%% file: subset-lmcs.tex
\documentclass{LMCS}

\usepackage{amsmath,amssymb}
\usepackage{proof} 
\newcommand\infertree[2]{\infer{#2}{#1}}

\usepackage{mathpartir}
\usepackage{stmaryrd}
\usepackage{diagrams}   
\usepackage{hyperref}
\usepackage{bigstrut}   
\usepackage{latexsym}   
\usepackage{url}        
\usepackage{enumerate}

\DeclareMathAlphabet{\mathsf}{OT1}{cmss}{m}{n}
\SetMathAlphabet{\mathsf}{bold}{OT1}{cmss}{b}{n}

\newtheorem{principle}[thm]{Principle}

\usepackage{shading}

\usepackage[usenames]{color}
\definecolor{tlcacolor}{rgb}{0.5,0.1,0.1}
\definecolor{trcolor}{rgb}{0.2,0.4,0.2}
\definecolor{wjlcolor}{rgb}{0.4,0.1,0.4}
\definecolor{Gray}{rgb}{0.8,0.8,0.8}

\newcommand{\boxassum}[1]{\ \ (#1)}

\newcommand{\inferaxiom}[1]{\inferrule{ }{#1}}
\newcommand{\rulename}[1]{\textbf{#1}}
\newcommand{\namedaxiom}[2]{\inferrule{ }{#1}\ (\rulename{#2})}
\newcommand{\namedrule}[3]{\inferrule{#1}{#2}\ (\rulename{#3})}
\newcommand{\rulesheader}[1]{\hfill\fbox{\ensuremath{#1}} \\}
\newcommand{\rulesheadernobreak}[1]{\hfill\fbox{\ensuremath{#1}} \nopagebreak}
\newcommand{\grammarheader}[1]{\begin{split} \hspace{-1em}\textbf{#1} \end{split}}

\newenvironment{indented}
    {\begin{list}{}{%
        \setlength{\topsep}{0.5em}%
        \setlength{\leftmargin}{1em}%
        \setlength{\rightmargin}{0em}%
        \setlength{\listparindent}{\parindent}%
        \setlength{\itemindent}{\parindent}%
        \setlength{\parsep}{\parskip}%
     }%
     \item[]
    }
    {\end{list}}
\newenvironment{proofcase}
    {\begin{tabbing}}
    {\end{tabbing}}

\newcommand\lit[1]{\mathsf{#1}}

\newcommand\arrow{\rightarrow}
\newcommand\intersect{\wedge}
\newcommand\rintersect{\mathrel{\wedge}}
\newcommand\refines{\mathrel{\sqsubset}}
\newcommand\hassort{\mathrel{::}}
\newcommand\subtype{\leq}
\newcommand\asubtype{\leqq}
\newcommand\Nspine{N_1 \ldots N_k}

\newcommand\decl[2]{#1{:}#2}
\newcommand\rdecl[3]{#1{\hassort}#2{\refines}#3}
\newcommand\arefdecl[3]{#1{\refines}#2{\hassort}#3}
\newcommand\mrefdecl[2]{#1{\hassort}#2}
\newcommand\subdecl[2]{#1{\subtype}{#2}}
\newcommand\bdot{.\,}
\newcommand\var[1]{\textit{#1}}
\newcommand\even{\var{e}}
\newcommand\odd{\var{o}}
\newcommand\pos{\var{p}}

\newcommand\ktype{\lit{type}}
\newcommand\kpi[3]{\Pi \decl{#1}{#2} \bdot #3}
\newcommand\kunit{1}
\newcommand\kprod[2]{#1 \times #2}

\newcommand\csort{\lit{sort}}
\newcommand\cpi[3]{\Pi \mrefdecl{#1}{#2} \bdot #3}
\newcommand\ccpi[4]{\Pi \rdecl{#1}{#2}{#3} \bdot #4}
\newcommand\ctop{\top}
\newcommand\cinter[2]{#1 \mathrel{\intersect} #2}

\newcommand\api[3]{\Pi \decl{#1}{#2} \bdot #3}
\newcommand\asigma[3]{\Sigma \decl{#1}{#2} \bdot #3}
\newcommand\asubset[3]{\{ \decl{#1}{#2} \mid #3 \}}
\newcommand\aapp[2]{#1\ #2}
\newcommand\aarrow[2]{#1 \arrow #2}

\newcommand\aunit{1}
\newcommand\aprod[2]{#1 \times #2}
\newcommand\afst[1]{\pi_1\, #1}
\newcommand\asnd[1]{\pi_2\, #1}

\newcommand\spi[3]{\Pi \mrefdecl{#1}{#2} \bdot #3}
\newcommand\sspi[4]{\Pi \rdecl{#1}{#2}{#3} \bdot #4}
\newcommand\sarrow[2]{#1 \arrow #2}
\newcommand\sapp[2]{#1\ #2}
\newcommand\sstop{\top}
\newcommand\sinter[2]{#1 \mathrel{\intersect} #2}

\newcommand\mapp[2]{#1\ #2}
\newcommand\mlam[2]{\lambda #1 \bdot #2}
\newcommand\munit{\langle \rangle}
\newcommand\mpair[2]{\langle #1, #2 \rangle}
\newcommand\mfst[1]{\pi_1\, #1}
\newcommand\msnd[1]{\pi_2\, #1}

\newcommand\synths{\Rightarrow}
\newcommand\checks{\Leftarrow}

\newcommand\kind{\lit{kind}}
\newcommand\jclass[4][] {#2 \vdash_{#1} #3 \refines #4}

\newcommand\jclasstform[5][] {#2 \vdash_{#1} #3 \refines #4
                              \stackrel{\mathrm{form}}{\leadsto} #5}

\newcommand\kindtpred[2] {#1 \stackrel{\mathrm{pred}}{\leadsto} #2}
\newcommand\kindtsub[2] {#1 \stackrel{\subtype}{\leadsto} #2}

\newcommand\jsynthsortclass[6][]
    {#2 \vdash_{#1} #3 \refines #4 \synths #5}
\newcommand\jsynthsortclasst[6][]
    {#2 \vdash_{#1} #3 \refines #4 \synths #5 \leadsto #6}
\newcommand\jchecksortclass[4][]{#2 \vdash_{#1} #3 \refines #4}
\newcommand\jchecksortclasst[5][]{#2 \vdash_{#1} #3 \refines #4 \leadsto #5}
\newcommand\jnchecksortclass[4][]{#2 \not\vdash_{#1} #3 \refines #4}

\newcommand\jsynthtermsort[4][]{#2 \vdash_{#1} #3 \synths #4}
\newcommand\jchecktermsort[4][]{#2 \vdash_{#1} #3 \checks #4}

\newcommand\jsynthtermsortt[5][]{#2 \vdash_{#1} #3 \synths #4 \leadsto #5}
\newcommand\jchecktermsortt[5][]{#2 \vdash_{#1} #3 \checks #4 \leadsto #5}

\newcommand\jsig[1]{\vdash #1\ \lit{sig}}
\newcommand\jnsig[1]{\not\vdash #1\ \lit{sig}}
\newcommand\jsigt[2]{\vdash #1\ \lit{sig} \leadsto #2}
\newcommand\jctx[2][]{\vdash_{#1} #2\ \lit{ctx}}
\newcommand\jctxt[3][]{\vdash_{#1} #2\ \lit{ctx} \leadsto #3}
\newcommand\eraserefs[1]{#1^*}
\newcommand\jkind[3][]{#2 \vdash_{#1} #3 \checks \kind}
\newcommand\jsynthtypekind[4][]{#2 \vdash_{#1} #3 \synths #4}
\newcommand\jchecktypekind[3][]{#2 \vdash_{#1} #3 \checks \ktype}
\newcommand\jsynthtermtype[4][]{#2 \vdash_{#1} #3 \synths #4} 
\newcommand\jchecktermtype[4][]{#2 \vdash_{#1} #3 \checks #4} 
\newcommand\jnchecktermtype[4][]{#2 \not\vdash_{#1} #3 \checks #4} 

\newcommand\asynths{\Rrightarrow}
\newcommand\achecks{\Lleftarrow}
\newcommand\jasynth[4][]{#2 \vdash_{#1} #3 \asynths #4}
\newcommand\jacheck[4][]{#2 \vdash_{#1} #3 \achecks #4}
\newcommand\jnacheck[4][]{#2 \not\vdash_{#1} #3 \achecks #4}
\newcommand\japply[3]{#1 \mathrel{@} #2 = #3}
\newcommand\under[3][]{#2 \vdash_{#1} #3}
\newcommand\jsubtype[2]{#1 \subtype #2}
\newcommand\jasubtype[2]{#1 \asubtype #2}
\newcommand\jnsubtype[2]{#1 \not\subtype #2}
\newcommand\jnasubtype[2]{#1 \not\asubtype #2}

\newcommand\jchecksubtype[5][]
            {#3 \subtype #4}
\newcommand\jsynthsubtype[7][]
            {#3 \subtype #4}

\newcommand\jchecksubtypet[5][]{#2 \vdash_{#1} #3 \subtype #4 \leadsto #5}
\newcommand\jsynthsubtypet[3]{#1 \subtype #2 \leadsto #3}

\newcommand\subst[3]{[#1/#2]\,#3}
\newcommand\hsubst[5]{[#2/#3]_{#4}^{\mathrm{#1}}\,#5}
\newcommand\jhsubst[6]{\hsubst{#1}{#2}{#3}{#4}{#5} = #6}
\newcommand\ssplit[1]{\mathrm{split}(#1)}
\newcommand\binter[1]{\bigwedge(#1)}
\newcommand\expand[2]{\eta_{#1}(#2)}

\newcommand\erasedeps[1]{(#1)^-}

\newcommand\treduce[4]{\decl{#1}{#2} \vdash #3 : #4}

\newcommand\defeq{\stackrel{\text{\scriptsize{def}}}{=}}
\newcommand\FV{\mathop{\mathrm{FV}}}
\newcommand\D{\mathop{\mathcal D}}
\newcommand\E{\mathop{\mathcal E}}

\newcommand\head{\mathop{\mathrm{head}}}

\newcommand\metalam[2]{\boldsymbol{\lambda} #1 \textbf{.}\  #2}
\newcommand{\metalp}{\boldsymbol{(}}
\newcommand{\metarp}{\boldsymbol{)}}
\newcommand{\metacomma}{\boldsymbol{,\,}}
\newcommand\metaapp[2]{#1 \metalp #2 \metarp}
\newcommand\metapair[2]{\metalp #1 \metacomma #2 \metarp}
\newcommand\metalamtwo[3]{\metalam {\metapair {#1} {#2}} {#3}}
\newcommand\metaapptwo[3]{#1 {\metapair {#2} {#3}}}

\newcommand\metalamfive[6]{\metalam {\metalp {#1}       %
                                     \metacomma {#2}    %
                                     \metacomma {#3}    %
                                     \metacomma {#4}    %
                                     \metacomma {#5}    %
                                     \metarp}           %
                                    {#6}}
\newcommand\metaappfive[6]{#1 \metalp {#2}      %
                                \metacomma {#3} %
                                \metacomma {#4} %
                                \metacomma {#5} %
                                \metacomma {#6} %
                              \metarp}
\newcommand\revapp[2]{#1 @ #2}
\newcommand\what\widehat
\newcommand\whatQQ{\what {Q_1{\textit-}Q_2}}
\newcommand\whatQQp{\what {Q_1{\textit-}Q'}}

\newcommand\f{\ensuremath{_{\mathrm{f}}}}
\newcommand\p{\ensuremath{_{\mathrm{p}}}}
\newcommand\longform\f
\newcommand\s{\ensuremath{_{\text{s}}}}
\newcommand\form[1]{\operatorname{form}(#1)}
\newcommand\judge[1]{\Gamma \vdash \mathcal{#1}}
\newcommand\judgetrans[2]{\Gamma \vdash \mathcal{#1} \leadsto #2}

\newcommand\GammaL{\Gamma_{\mathrm L}}
\newcommand\GammaR{\Gamma_{\mathrm R}}

\usepackage{listings}
\include{lfr-listings}
\newcommand\lst\lstinline

\newcommand{\promote}[1]{#1^{\oplus}}
 
\newcommand\idecl[2]{#1{\div}#2}

\newcommand\miapp[2]{#1\ [#2]}
\newcommand\aiapp[2]{#1\ [#2]}

\newlength{\arrowlength}
\newcommand\iarrow{\mathrel{\rlap{$\rightarrow$}\makebox[\arrowlength]{$\div$}}}
\newcommand\aiarrow[2]{#1 \iarrow #2}

\newcommand{\Ntwohat}{\ensuremath{\what{\var{N2}}}}
\newcommand{\hatted}[1]{\ensuremath{\what{\var{#1}}}}
\newcommand{\zhat}{\hatted{z}}
\newcommand{\shat}{\hatted{s}}
\newcommand{\doublestar}{\ensuremath{\var{double*}}}
\newcommand{\doublestarhat}{\hatted{double*}}
\newcommand{\dblzhat}{\hatted{dbl/z}}
\newcommand{\dblshat}{\hatted{dbl/s}}

\def\doi{6 (4:5) 2010}
\lmcsheading%
{\doi}
{1--50}
{}
{}
{Nov.~27, 2009}
{Dec.~\phantom05, 2010}
{}

\begin{document}

\title[Refinement Types for Logical Frameworks]
      {Refinement Types for Logical Frameworks and Their Interpretation as
       Proof Irrelevance\rsuper*}

\author[W.~Lovas]{William Lovas}
\address{Carnegie Mellon University \\
         Pittsburgh, PA 15213, USA}
\email{\{wlovas,fp\}@cs.cmu.edu}

\author[F.~Pfennig]{Frank Pfenning}
\address{\vskip-6 pt}

\titlecomment{{\lsuper*}The work was partially supported by the
Funda\c{c}\~{a}o para a Ci\^{e}ncia e Tecnologia (FCT), Portugal,
under a grant from the Information and Communications Technology Institute
(ICTI) at Carnegie Mellon University.}




%
%
%

\begin{abstract}
Refinement types sharpen systems of simple and dependent types by offering
expressive means to more precisely classify well-typed terms.
We present a system of refinement types for LF in the style of
recent formulations where only canonical forms are well-typed.  Both the
usual LF rules and the rules for type refinements are bidirectional,
leading to a straightforward proof of decidability of typechecking even in
the presence of intersection types.  Because we insist on canonical forms,
structural rules for subtyping can now be derived rather than being assumed
as primitive.  We illustrate the expressive power of our system with
examples and validate its design by demonstrating a precise correspondence
with traditional presentations of subtyping.

Proof irrelevance provides a mechanism for selectively hiding the identities
of terms in type theories.
We show that LF refinement types can be interpreted as
predicates using proof irrelevance,
establishing a uniform relationship between two previously studied
concepts in type theory.
The interpretation and its correctness proof are surprisingly complex,
lending support to the claim that refinement types are a fundamental
construct rather than just a convenient surface syntax for certain uses of
proof irrelevance.
\end{abstract}

\maketitle

%



\section{Introduction}

\noindent LF was created as a framework for defining logics and
programming languages \cite{Harper93jacm}.  Since its inception, it
has been used to represent and formalize reasoning about a number of
deductive systems, which are prevalent in the study of logics and
programming languages.\footnote{See \cite{Pfenning01handbook} for an
  introduction to logical frameworks and further references.}  In its
most recent incarnation as the Twelf metalogic \cite{Pfenning99cade},
it has been used to encode and mechanize the metatheory of programming
languages that are prohibitively complex to reason about on paper
\cite{Crary03popl,Lee07}.

It has long been recognized that some LF encodings would benefit from the
addition of a subtyping mechanism to LF \cite{Pfenning93types,Aspinall01}.
In LF encodings, judgments are represented by type families, and many
subsets of data types and judgmental inclusions can be elegantly
represented via subtyping.

Prior work has explored adding subtyping and intersection types to LF via
\emph{refinement types} \cite{Pfenning93types}.  Many of that system's
metatheoretic properties were proven indirectly by translation into other
systems, though, giving little insight into notions of adequacy or
implementation strategies.  We begin this paper by presenting a refinement
type system for LF based on the modern \emph{canonical forms} approach
\cite{Watkins02tr,hl07mechanizing}, and by doing so we obtain direct proofs of
important properties like decidability.  Moreover, the theory of canonical
forms provides the basis for a study of adequacy theorems exploiting
refinement types.

In canonical forms-based LF, only $\beta$-normal $\eta$-long terms are
well-typed --- the syntax restricts terms to being $\beta$-normal, while
the typing relation forces them to be $\eta$-long.  Since standard
substitution might introduce redexes even when substituting a normal term
into a normal term, it is replaced with a notion of \emph{hereditary
substitution} that contracts redexes along the way, yielding another normal
term.  Since only canonical forms are admitted, type equality is just
$\alpha$-equivalence, and typechecking is manifestly decidable.

Canonical forms are exactly the terms one cares about when adequately
encoding a language in LF, so this approach loses no expressivity.
Since all terms are normal, there is no notion of reduction, and thus the
metatheory need not directly treat properties related to reduction, such as
subject reduction, Church-Rosser, or strong normalization.  All of the
metatheoretic arguments become straightforward structural inductions, once
the theorems are stated properly.

By introducing a layer of refinements distinct from the usual layer of types,
we prevent subtyping from interfering with our extension's metatheory.  We
also follow the general philosophy of prior work on refinement types
\cite{Freeman91,Freeman94,Davies05} in only assigning refined types to terms
already well-typed in pure LF, ensuring that our extension is conservative.

As a simple example, we study the representation of natural numbers
as well as even and odd numbers.  In normal logical discourse, we
might define these with the following grammar:
\[\begin{array}{llcl}
\mbox{Natural numbers} & n & ::= & z \mid s(n) \\[1ex]
\mbox{Even numbers} & e & ::= & z \mid s(o) \\
\mbox{Odd numbers} & o & ::= & s(e)
\end{array}\]
The first line can be seen as defining the \emph{abstract syntax} of
natural numbers in unary form, the second and third lines as
defining two subsets of the natural numbers defined in the first line.
We will follow this informal convention, and represent the first
as a \emph{type} with two constructors.
\begin{quote}
$ \var{nat} : \ktype. $ \\
$ \var {z} : \var{nat}. $ \\
$ \var {s} : \var{nat} \rightarrow \var{nat}. $
\end{quote}
The second and third line define even and odd numbers as
a subset of the natural numbers, which we represent as \emph{refinements}
of the type \lst{nat}.
\begin{quote}
$ \var{even} \refines \var{nat}.\quad\var{odd} \refines \var{nat}. $\\
$ \var{z} :: \var{even}.$\\ 
$ \var{s} :: \var{even} \arrow \var{odd}\ \rintersect\ \var{odd} \arrow \var{even}. $
\end{quote}
In the above, \lst{even << nat} declares \lst{even} as a refinement of the
type \lst{nat}, and the declarations using ``\lst{::}'' give more precise
sorts for the constructors \lst{z} and \lst{s}.  Note that since the successor
function satisfies two unrelated properties, we give two refinements for it
using an intersection sort.  We can give similar representations of all
regular tree grammars as refinements, which then represent regular tree
types \cite{DartRegular}.  Our language generalizes this further to allow
binding operators and dependent types, both of which it inherits from LF,
thereby going far beyond what can be recognized with tree automata
\cite{tata2007}.

Already in this example we can see that it is natural to use refinements to
represent certain subsets of data types.  Conversely, refinements can be
interpreted as defining subsets.  In the second part of this paper, we exhibit
an interpretation of LF refinement types which we refer to as the
``\emph{subset interpretation}'', since a sort refining a type is interpreted
as a predicate embodying the refinement, and the set of terms having that sort
is simply the subset of terms of the refined type that also satisfy the
predicate.  For example, under
the subset interpretation, we translate the refinements \lst{even} and
\lst{odd} to predicates on natural numbers.
The refinement declarations for \lst{z} and \lst{s} turn into constructors for
proofs of these predicates.\smallskip
\begin{quote}
$    even : nat \rightarrow\ktype.\quad odd : nat \rightarrow\ktype. $\\
$    \zhat\phantom{_1} : even\,\, z. $\\
$    \shat_1 : \raisebox{2pt}{$\scriptstyle\prod$} x{:}nat.\,\, even\,\, x \rightarrow odd\, (s\, x). $\\
$    \shat_2 : \raisebox{2pt}{$\scriptstyle\prod$}x{:}nat.\,\, odd\,\, x \rightarrow even\, (s\, x). $
\end{quote}\smallskip
The successor function's two unrelated sorts translate to proof
constructors for two different predicates.

We show that our interpretation is correct by proving, for instance, that a
term $N$ has sort $S$ if and only if its translation $\what N$ has type
$\metaapp {\what S} N$, where $\metaapp {\what S} -$ is the translation of the
sort $S$ into a type family representing a predicate; thus, an adequate
encoding using refinement types remains adequate after translation.  The chief
complication in proving correctness is the dependency of types on terms, which
forces us to deal with a \textit{coherence} problem
\cite{Tannen91coercion,REYNOLDS91}.

Normally, subset interpretations are not subject to the issue of
coherence---that is, of ensuring that the interpretation of a judgment is
independent of its derivation---since the terms in the target of the
translation are \textit{the same} as the terms in the source, just with the
stipulation that a certain property hold of them.  The proofs of these
properties are computationally immaterial, so they may simply be ignored.  But
the presence of full dependent types in LF means that the interpretation of a
sort might depend on these proofs, potentially violating the adequacy of
representations.

In order to solve the coherence problem we employ \emph{proof irrelevance}, a
technique used in type theories to selectively hide the identities of terms
representing proofs \cite{Pfenning01lics,AwodeyBauer04}.  In the example, the
terms whose identity should be irrelevant are those constructing proofs of
\lst{odd($n$)} and \lst{even($n$)}, that is, those composed from $\zhat$,
$\shat_1$, and $\shat_2$.

The subset interpretation completes our intuitive understanding of refinement
types as representing subsets of types.  It turns out that in the presence of
variable binding and dependent types, this understanding is considerably more
difficult to attain than it might seem from the small example above.

In the remainder of the paper, we describe our refinement type system
alongside a few illustrative examples (Section~\ref{sect:system-examples}).
Then we explore its metatheory and sketch proofs of key results, including
decidability (Section~\ref{sect:metatheory}).  We note that our approach
leads to subtyping only being defined at base types, but we show that
this is no restriction at all: subtyping at higher types is intrinsically
present due to the use of canonical forms
(Section~\ref{sect:higher-subsorting}).  Next, we take a brief detour to
review prior work on proof irrelevance (Section~\ref{sect:proof-irrelevance}),
setting the stage for our subset interpretation and proofs of its
correctness (Section~\ref{sect:interpretation}).  Finally, we offer some
concluding remarks on the broader implications of our work
(Section~\ref{sect:conclusion}).

This paper represents a combination of the developments in a technical report
on the basic design of LF with refinement types \cite{Lovas07,Lovas08tr} and a
conference paper sketching the subset interpretation \cite{Lovas09subset}.

\section{System and Examples}
\label{sect:system-examples}

\noindent We present our system of LF with Refinements, LFR, through
several examples.  In what follows, $R$ refers to atomic terms and $N$
to normal terms.  Our atomic and normal terms are exactly the terms
from canonical presentations of LF\@.
\begin{align*}
R &::= c
    \mid x
    \mid \mapp R N
  & \text{atomic terms}
\\
N, M &::= R
    \mid \mlam x N
  & \text{normal terms}
\end{align*}
In this style of presentation, typing is defined bidirectionally by two
judgments: $R \synths A$, which says atomic term $R$ \emph{synthesizes}
type $A$, and $N \checks A$, which says normal term $N$ \emph{checks}
against type $A$.  Since $\lambda$-abstractions are always checked against
a given type, they need not be decorated with their domain types.

Types are similarly stratified into atomic and normal types.
\begin{align*}
P &::= a
    \mid \aapp P N
  & \text{atomic type families}
\\
A, B &::= P
    \mid \api x A B
  & \text{normal type families}
\end{align*}

The operation of hereditary substitution, written $\hsubst {} N x A$,
is a partial function which computes the normal form of the standard
capture-avoiding substitution of $N$ for $x$.  It is indexed by the
putative type of $x$, $A$, to ensure termination, but neither the variable
$x$ nor the substituted term $N$ are required to bear any relation to this
type index for the operation to be defined.  We show in
Section~\ref{sect:metatheory} that when $N$ and $x$ \emph{do} have type
$A$, hereditary substitution is a total function on well-formed terms.

As a philosophical aside, we note that restricting our attention to normal
terms in this way is similar to the idea of restricting one's attention to
cut-free proofs in a sequent calculus~\cite{pfenning00structural}.  Showing
that hereditary substitution can always compute a canonical form is
analogous to showing the cut rule admissible.  And just as cut
admissibility may be used to prove a \textit{cut elimination} theorem,
hereditary substitution may be used to prove a \textit{normalization}
theorem relating the canonical approach to traditional formulations.  We
will not explore the relationship any further in the present work: the
canonical terms are the only ones we care about when formalizing deductive
systems in a logical framework, so we simply take the canonical
presentation as primary.

Our layer of refinements uses metavariables $Q$ for atomic sorts and $S$
for normal sorts.  These mirror the definition of types above, except
for the addition of intersection and ``top'' sorts.
\begin{align*}
Q &::= s
    \mid \sapp Q N
  & \text{atomic sort families}
\\
S, T &::= Q
    \mid \sspi x S A T
    \mid \sstop
    \mid \sinter {S_1} {S_2}
  & \text{normal sort families}
\end{align*}
Sorts are related to types by a refinement relation, $S \refines A$ (``$S$
refines $A$''), discussed below.  We only sort-check well-typed terms, and
a term of type $A$ can be assigned a sort $S$ only when $S \refines A$.
These constraints are collectively referred to as the ``refinement
restriction''.  We occasionally omit the ``$\refines A$'' from function
sorts when it is clear from context.

Deductive systems are encoded in LF using the \textit{judgments-as-types}
principle~\cite{Harper93jacm,hl07mechanizing}: syntactic categories are
represented by simple types, and judgments over syntax are represented by
dependent type families.  Derivations of judgments are inhabitants of those
type families, and well-formed derivations correspond to well-typed LF
terms.  An LF signature is a collection of kinding declarations $a : K$ and
typing declarations $c : A$ that establishes a set of syntactic categories,
a set of judgments, and inhabitants of both.  In LFR, we can represent
syntactic subsets or sets of derivations that have certain
\textit{properties} using sorts.  Thus one might say that the methodology
of LFR is \textit{properties-as-sorts}.

\subsection{Example: Natural Numbers}

For the first running example we will use the natural numbers in
unary notation.  In LF, they would be specified as follows
\begin{quote}
$ \var{nat} : \ktype. $ \\
$ \var {z} : \var{nat}. $ \\
$ \var {s} : \var{nat} \rightarrow \var{nat}. $
\end{quote}
These declarations establish a syntactic category of natural numbers
populated by two constructors, a constant constructor representing
zero and a unary constructor representing the successor function.

Suppose we would like to distinguish the odd and the even numbers
as refinements of the type of all numbers.
\begin{quote}
$ \var{even} \refines \var{nat}. $ \\
$ \var{odd} \refines \var{nat}. $
\end{quote}
The form of the declaration is $s \refines a$ where $a$ is
a type family already declared and $s$ is a new sort family.
Sorts headed by $s$ are declared in this way to refine types headed by $a$.
The relation $S \refines A$ is extended through
the whole sort hierarchy in a compositional way.

Next we declare the sorts of the constructors.  For zero, this is easy:
\begin{quote}
$ \var{z} :: \var{even}. $
\end{quote}
The general form of this declaration is $c :: S$, where
$c$ is a constant already declared in the form
$c : A$, and where $S \refines A$.  The declaration for the successor is
slightly more difficult, because it maps even numbers to odd
numbers and vice versa.  In order to capture both properties
simultaneously we need to use an \emph{intersection sort}, written
as $S_1 \rintersect S_2$.\footnote{Intersection has lower precedence
than arrow.}
\begin{quote}
$ \var{s} ::
  \var{even} \arrow \var{odd}\ \rintersect\ \var{odd} \arrow \var{even}. $
\end{quote}
In order for an intersection to be well-formed, both components must refine
the same type.  The nullary intersection $\sstop$ can refine any type, and
represents the maximal refinement of that type.\footnote{As usual in LF, we
use $A \arrow B$ as shorthand for the dependent type $\api x A B$ when $x$
does not occur in $B$.}
\begin{mathpar}
\inferrule{s \refines a \in \Sigma}
          {s\ \Nspine \refines a\ \Nspine}
\and
\inferrule{S \refines A \\ T \refines B}
          {\spi x S T \refines \api x A B}
\and
\inferrule{S_1 \refines A \\ S_2 \refines A}
          {\sinter {S_1} {S_2} \refines A}
\and
\inferaxiom  {\sstop \refines A}
\end{mathpar}
To show that the declaration for $\var{s}$ is well-formed, we establish that
$\var{even} \arrow \var{odd} \rintersect \var{odd} \arrow \var{even} \
\refines\  \var{nat} \rightarrow \var{nat}$.

The \emph{refinement relation} $S \refines A$ should not be confused with
the usual \emph{subtyping relation}.  Although each is a kind of subset
relation%
\footnote{%
It may help to recall the interpretation of $S \refines A$: for a term to
be judged to have sort $S$, it must already have been judged to have type
$A$ for some $A$ such that $S \refines A$.  Thus, the refinement relation
represents an inclusion ``by fiat'': every term with sort $S$ is also a
term of type $A$, by invariant.  By contrast, subsorting $S_1 \subtype S_2$
is a more standard sort of inclusion: every term with sort $S_1$ is also a
term of sort $S_2$, by subsumption (see
Section~\ref{sect:higher-subsorting}).%
}%
, they are quite different:  Subtyping relates two types, is
contravariant in the domains of function types, and is transitive, while
refinement relates a sort to a type, so it does not make sense to consider
its variance or whether it is transitive.  We will discuss subtyping
below and in Section~\ref{sect:higher-subsorting}.

Now suppose that we also wish to distinguish the strictly positive natural
numbers.  We can do this by introducing a sort $\var{pos}$ refining
$\var{nat}$ and declaring that the successor function yields a $\var{pos}$
when applied to anything, using the maximal sort.
\begin{quote}
$ \var{pos} \refines \var{nat}. $ \\
$ \var{s} :: \cdots \rintersect \sstop \arrow \var{pos}. $
\end{quote}
Since we only sort-check well-typed programs and $\var{s}$ is declared
to have type $\var{nat} \arrow \var{nat}$, the sort $\sstop$ here acts as a
sort-level reflection of the entire $\var{nat}$ type.

We can specify that all odds are positive by declaring $\var{odd}$ to
be a subsort of $\var{pos}$.
\begin{quote}
$ \var{odd} \subtype \var{pos}. $
\end{quote}
Although any ground instance of $\var{odd}$ is evidently $\var{pos}$, we
need the subsorting declaration to establish that variables of sort
$\var{odd}$ are also $\var{pos}$.

Putting it all together, we have the following:
\begin{quote}
$ \var{even} \refines \var{nat}. $ \quad
$ \var{odd} \refines \var{nat}. $ \quad
$ \var{pos} \refines \var{nat}. $ \\
$ \var{odd} \subtype \var{pos}. $ \\
$ \var{z} :: \var{even}. $ \\
$ \var{s} ::
      \var{even} \arrow \var{odd}\ \rintersect\ 
      \var{odd} \arrow \var{even}\ \rintersect\ 
      \sstop \arrow \var{pos}. $
\end{quote}

Now we should be able to verify that, for example, $\var{s}\
(\var{s}\ \var{z}) \checks \var{even}$.
To explain how, we analogize with pure canonical LF\@.
Recall that atomic types have the form $a\ \Nspine$ for a type family $a$
and are denoted by
$P$.  Arbitrary types $A$ are either atomic ($P$) or
(dependent) function types ($\api x A B$).  Canonical terms are
then characterized by the rules shown in the left column above.

\begin{figure}
\begin{center}
\begin{tabular*}{\textwidth}{@{\extracolsep{\fill}}c|c}
Canonical LF                    &   LF with Refinements \\
\hline
  \begin{minipage}{0.419\textwidth}
    \begin{mathpar}
    \inferrule{\jchecktermtype {\Gamma, \decl x A} N B}
              {\jchecktermtype \Gamma {\mlam x N} {\api x A B}}
    \and
    \inferrule{\jsynthtermtype \Gamma R {P'} \\ P' = P}
              {\jchecktermtype \Gamma R P}
    \end{mathpar}
    \hrule
    \begin{mathpar}
    \inferrule{\decl x A \in \Gamma}
              {\jsynthtermtype \Gamma x A}
    \and
    \inferrule{\decl c A \in \Sigma}
              {\jsynthtermtype \Gamma c A}
    \and
    \inferrule{\jsynthtermtype \Gamma R {\api x A B}
            \\ \jchecktermtype \Gamma N A
              }
              {\jsynthtermtype \Gamma {\mapp R N} {\hsubst {} N x A B}}
    \end{mathpar}
  \end{minipage}
    &
    \begin{minipage}{0.53\textwidth}
      \begin{mathpar}
        \namedrule{\jchecktermtype {\Gamma, \rdecl x S A} N T}
                  {\jchecktermtype \Gamma {\mlam x N} {\sspi x S A T}}
                  {$\Pi$-I}
        \and
        \namedrule{\jsynthtermtype \Gamma R {Q'} \\ Q' \subtype Q}
                  {\jchecktermtype \Gamma R Q}
                  {switch}
      \end{mathpar}
      \hrule
      \begin{mathpar}
        \namedrule{\rdecl x S A \in \Gamma}
                  {\jsynthtermtype \Gamma x S}
                  {var}
        \and
        \namedrule{c :: S \in \Sigma}
                  {\jsynthtermtype \Gamma c S}
                  {const}
        \and
        \namedrule{\jsynthtermtype \Gamma R {\sspi x S A T}
                \\ \jchecktermtype \Gamma N S
                  }
                  {\jsynthtermtype \Gamma {\mapp R N} {\hsubst {} N x A T}}
                  {$\Pi$-E}
      \end{mathpar}
    \end{minipage}
\end{tabular*}
\end{center}
\end{figure}

There are two typing judgments, $N \checks A$ which means that
$N$ checks against $A$ (both given) and $R \synths A$ which
means that $R$ synthesizes type $A$ ($R$ given as input, $A$ produced as
output).  Both take place in a
context $\Gamma$ assigning types to variables.
To force terms to be $\eta$-long, the rule for checking an atomic term $R$
only checks it at an atomic type $P$.  It does so by synthesizing a type
$P'$ and comparing it to the given type $P$.  In canonical LF, all types
are already canonical, so this comparison is just $\alpha$-equality.

On the right-hand side we have shown the corresponding rules for
sorts.  First, note that the format of the context $\Gamma$ is slightly
different, because it declares sorts for variables, not just types.
The rules for functions and applications are straightforward analogues
to the rules in ordinary LF\@.  The rule \rulename{switch} for checking
atomic terms $R$ at atomic sorts $Q$ replaces the equality check with a
subsorting check and is the only place where we appeal to subsorting
(defined below).  For applications, we use the type $A$ that refines the
type $S$ as the index parameter of the hereditary substitution.

Subsorting is exceedingly simple: it only needs to be defined
on atomic sorts, and is just the reflexive and transitive closure
of the declared subsorting relationship.
\begin{mathpar}
\inferrule{\subdecl{s_1}{s_2} \in \Sigma}
          {\sapp {s_1} \Nspine \subtype \sapp {s_2} \Nspine}
\and
\inferaxiom {Q \subtype Q}
\and
\inferrule{Q_1 \subtype Q'
        \\ Q' \subtype Q_2}
          {Q_1 \subtype Q_2}
\end{mathpar}

\noindent The sorting rules do not yet treat intersections.  In line
with the general bidirectional nature of the system,
the introduction rules are part of the
\emph{checking} judgment, and the elimination rules are
part of the \emph{synthesis} judgment.
Binary intersection $\sinter {S_1} {S_2}$ has one introduction and two
eliminations, while nullary intersection $\sstop$ has just one introduction.
\begin{mathpar}
\and
\namedrule{\jchecktermtype \Gamma N {S_1} \\ \jchecktermtype \Gamma N {S_2}}
          {\jchecktermtype \Gamma N {\sinter {S_1} {S_2}}}
          {$\intersect$-I}
\and
\namedaxiom{\jchecktermtype \Gamma N \sstop}
           {$\sstop$-I}
\and
\\
\namedrule{\jsynthtermtype \Gamma R {\sinter {S_1} {S_2}}}
          {\jsynthtermtype \Gamma R {S_1}}
          {$\intersect$-E$_1$}
\and
\namedrule{\jsynthtermtype \Gamma R {\sinter {S_1} {S_2}}}
          {\jsynthtermtype \Gamma R {S_2}}
          {$\intersect$-E$_2$}
\end{mathpar}
Note that although (canonical forms-style) LF type synthesis is unique, LFR
sort synthesis is not, due to the intersection elimination rules.

Now we can see how these rules generate a deduction of
$\var{s}\ (\var{s}\ \var{z}) \checks \var{even}$.
The context is always empty and therefore omitted.
To save space, we abbreviate $\var{even}$ as \even, $\var{odd}$ as \odd,
and
$\var{pos}$ as \pos, 
and we omit reflexive uses of subsorting.
\begin{mathpar}
\infertree{
    \infertree{
         \infertree{
          \infertree{
           \infertree{}{
           \jsynthtermtype {} s {\sinter {\aarrow \even \odd}
                                         {(\sinter {\aarrow \odd \even}
                                                   {\aarrow \sstop \pos})}}
           }
          }
          {\jsynthtermtype {} s {\sinter {\aarrow \odd \even}
                                        {\aarrow \sstop \pos}}}
         }
         {\jsynthtermtype {} s {\aarrow \odd \even}}
        &
         \infertree{
          \infertree{
           \infertree{
            \infertree{}{
            \jsynthtermtype {} s {\sinter {\aarrow \even \odd}
                                          {(\ldots)}}
            }
           }
           {\jsynthtermtype {} s {\aarrow \even \odd}}
           &
           \infertree{
            \infertree{}{\jsynthtermtype {} z \even}
           }
           {\jchecktermtype {} z \even}
          }
          {\jsynthtermtype {} {s\ z} \odd}
         }
         {\jchecktermtype {} {s\ z} \odd}
        }
        {\jsynthtermtype {} {s\ (s\ z)} \even}
   }
   {\jchecktermtype {} {s\ (s\ z)} \even}
\end{mathpar}

\noindent Using the \rulename{$\intersect$-I} rule, we can check that $\var{s}\
\var{z}$ is both odd and positive:
\begin{mathpar}
    \infer{\jchecktermtype {} {s\ z} {\sinter \odd \pos}}
        { \deduce{\jchecktermtype {} {s\ z} \odd}{\vdots}
        & \deduce{\jchecktermtype {} {s\ z} \pos}{\vdots} }
\end{mathpar}
Each remaining subgoal now proceeds similarly to the above example.

To illustrate the use of sorts with non-trivial type \emph{families},
consider the definition of the $\var{double}$ relation in LF\@.  We declare
a type family representing the doubling judgment and populate it with two
proof rules.
\begin{quote}
$ \var{double} : \var{nat} \arrow \var{nat} \arrow \ktype. $ \\
$ \var{dbl/z} : \var{double}\ \var{z}\ \var{z}. $ \\
$ \var{dbl/s} : \api X {\var{nat}} {\api Y {\var{nat}} {\var{double}\ X\ Y}
       \arrow \var{double}\ (\var{s}\ X)\ (\var{s}\ (\var{s}\ Y))}. $
\end{quote}
With sorts, we can now directly express the property that the second
argument to $\var{double}$ must be even.  But to do so, we require a notion
analogous to \emph{kinds} that may contain sort information.  We call these
\emph{classes} and denote them by $L$.
\begin{align*}
K &::= \ktype
    \mid \kpi x A K
  & \text{kinds}
\\
L &::= \csort
    \mid \ccpi x S A L
    \mid \ctop
    \mid \cinter {L_1} {L_2}
  & \text{classes}
\end{align*}
Classes $L$ mirror kinds $K$, and they have a refinement relation $L
\refines K$ similar to $S \refines A$.  (We elide the rules here, but they
are included in Appendix~\ref{app:lfr-rules}.)
Now, the general form of the $s \refines a$ declaration is $s \refines a ::
L$, where $a : K$ and $L \refines K$; this declares sort constant $s$ to
refine type constant $a$ and to have class $L$.

For now, we reuse the type name $\var{double}$ as a sort, as no ambiguity
can result.  As before, we use $\sstop$ to represent a $\var{nat}$ with no
additional restrictions.
\begin{quote}
$ \var{double} \refines \var{double} ::
        \sstop \arrow \var{even} \arrow \csort. $ \\
$ \var{dbl/z} :: \var{double}\ \var{z}\ \var{z}. $ \\
$ \var{dbl/s} :: \spi X \sstop {\spi Y {\var{even}} {\var{double}\ X\ Y}
        \arrow \var{double}\ (\var{s}\ X)\ (\var{s}\ (\var{s}\ Y))}. $
\end{quote}

\noindent After these declarations, it would be a static \emph{sort error}
to pose a query such as 
\begin{center}
``$\texttt{?-}\ \var{double}\ X\ (\var{s}\ 
(\var{s}\ (\var{s}\ \var{z}))).$''
\end{center}
before any search is ever attempted.  In LF, queries like this
could fail after a long search or even not
terminate, depending on the search strategy.
One of the
important motivations for considering sorts for LF is
to avoid uncontrolled search in favor of decidable
static properties whenever possible.

The tradeoff for such precision is that now
sort checking itself is non-deterministic and has to
perform search because of the choice between the two intersection
elimination rules.  
As Reynolds has shown, this non-determinism causes intersection type
checking to be PSPACE-hard \cite{REYNOLDS96}, even for normal terms
as we have here \cite{REYNOLDS89B}.  Using techniques such as
focusing, we believe that for practical cases they can be analyzed
efficiently for the purpose of sort checking.\footnote{The present paper
concentrates primarily on decidability, though, not efficiency.}

\subsection{A Second Example: The \texorpdfstring{$\lambda$}{lambda}-Calculus}

As a second example, we use an intrinsically typed version of the
call-by-value simply-typed $\lambda$-calculus.  This means every object language
expression is indexed by its object language type.  We use sorts to
distinguish the set of \emph{values} from the set of arbitrary
\emph{computations}.  While this can be encoded in LF in a variety
of ways, it is significantly more cumbersome.

\newcommand\RA{\Mapsto}
\newcommand\ra{\arrow}
\begin{quote}
\begin{tabbing}
$ \var{tp} : \ktype. $    \hspace{6em} \= \%\ the type of object language
                                               types \\
$ \RA\ : \var{tp} \ra \var{tp} \ra \var{tp}. $
                                        \> \%\ object language function space \\
$ \textbf{\%infix}\ \var{right}\ 10\ \RA. $ \\
 \\
$ \var{exp} : \var{tp} \ra \ktype. $    \> \%\ the type of expressions \\
$ \var{cmp} \refines \var{exp}. $       \> \%\ the sort of computations \\
$ \var{val} \refines \var{exp}. $       \> \%\ the sort of values \\
 \\
$ \var{val} \subtype \var{cmp}. $       \> \%\ every value is a (trivial)
                                               computation \\
 \\
$ \var{lam} :: (\var{val}\ A \ra \var{cmp}\ B) \ra \var{val}\ (A \RA B). $ \\
$ \var{app} :: \var{cmp}\ (A \RA B) \ra \var{cmp}\ A \ra \var{cmp}\ B. $
\end{tabbing}
\end{quote}

\noindent In the last two declarations, we follow Twelf convention and
leave the quantification over $A$ and $B$ implicit, to be inferred by
type reconstruction.  Also, we did not explicitly declare a type for
$\var{lam}$ and $\var{app}$.  We posit a front end that can recover
this information from the refinement declarations for $\var{val}$ and
$\var{cmp}$, avoiding redundancy.

The most interesting declaration is the one for the constant $\var{lam}$.
The argument type $(\var{val}\ A \arrow \var{cmp}\ B)$ indicates that
$\var{lam}$ binds a variable which stands for a value
of type $A$ and the body is an arbitrary computation of type $B$.  The result
type $\var{val}\ (A \RA B)$ indicates that any $\lambda$-abstraction is a
value.  Now we have, for example (parametrically in $A$ and
$B$):
$\jchecktermtype {\rdecl A \sstop {\var{tp}}, \rdecl B \sstop {\var{tp}}}
                    {\var{lam}\ \mlam x \var{lam}\ \mlam y x}
                    {\var{val}\ (A \RA (B \RA A))}$.

Now we can express that evaluation must always returns a value.
Since the declarations below are intended to represent a logic program,
we follow the logic programming convention of reversing the arrows in the
declaration of $\var{ev-app}$.
\begin{quote}
\begin{tabbing}
$ \var{eval}
    :: \var{cmp}\ A \arrow \var{val}\ A \arrow \csort. $ \\
$ \var{ev-lam}
    :: \var{eval}\ (\var{lam}\ \mlam x {E\ x})\ (\var{lam}\ \mlam x {E\ x}). $\\
$ \var{ev-app}
      :: $ \= $ \var{eval}\ (\var{app}\ {E_1}\ {E_2})\ V $ \\
      \> $ \leftarrow \var{eval}\ E_1\ (\var{lam}\ \mlam x {E_1'\ x}) $ \\
      \> $ \leftarrow \var{eval}\ E_2\ V_2 $ \\
      \> $ \leftarrow \var{eval}\ (E_1'\ V_2)\ V. $
\end{tabbing}
\end{quote}
Sort checking the above declarations demonstrates that when evaluation
returns at all, it returns a syntactic value.
%
%
Moreover, if sort reconstruction gives $E_1'$ the ``most general'' sort
$\var{val}\ A \arrow \var{cmp}\ B$,
the declarations also
ensure that the language is indeed call-by-value: it would
be a sort error to ever substitute a computation for
a $\var{lam}$-bound variable, for example, by evaluating
$(E_1'\ E_2)$ instead of $(E_1'\ V_2)$ in the $\var{ev-app}$ rule.
An interesting question for future work is whether type reconstruction can
always find such a ``most general'' sort for implicitly quantified
metavariables.

A side note: through the use of
sort families indexed by object language types,
the sort checking not only guarantees that the language
is call-by-value and that evaluation, if it succeeds, will
always return a value, but also that the object language
type of the result remains the same (type preservation).

\section{Metatheory}
\label{sect:metatheory}

\noindent In this section, we present some metatheoretic results about our
framework.  These follow a similar pattern as previous work using
hereditary substitutions \cite{Watkins02tr,Nanevski07tocl,hl07mechanizing}.
We give sketches of all proofs.  Technically tricky proofs are available
from a companion technical report \cite{Lovas08tr}.

\subsection{Hereditary Substitution}

Recall that we replace ordinary capture-avoiding substitution with
\emph{hereditary substitution}, $\hsubst {} N x A$, an operation which
substitutes a normal term into a canonical form yielding another canonical
form, contracting redexes ``in-line''.  The
operation is indexed by the putative type of $N$ and $x$ to facilitate a
proof of termination.
In fact, the type index on hereditary substitution need only be
a simple type to ensure termination.  To that end, we denote simple types
by $\alpha$ and define an erasure to simple types $\erasedeps A$.
\begin{align*}
\alpha &::= a \mid \aarrow {\alpha_1} {\alpha_2} 
&\erasedeps {\aapp a \Nspine} &= a 
&\erasedeps {\api x A B} &= {\aarrow {\erasedeps A} {\erasedeps B}}
\end{align*}
For clarity, we also index hereditary substitutions
by the syntactic category on which they operate, so for example we have
$\jhsubst n N x A M {M'}$ and $\jhsubst s N x A S {S'}$;
Table~\ref{tab:hsubst} lists all of the judgments defining substitution.
We write $\jhsubst n N x A M {M'}$ as short-hand for $\jhsubst n N x
{\erasedeps A} M {M'}$.

\newlength{\oldarrayrulewidth}
\setlength{\oldarrayrulewidth}{\arrayrulewidth}
\setlength{\arrayrulewidth}{0.2pt}
\setlength{\doublerulesep}{\arrayrulewidth}

\begin{table}
\begin{center}
\begin{tabular*}{0.80\textwidth}%
                {l@{\extracolsep{\fill}}c@{\extracolsep{0pt}}l}
\hline\hline\hline\hline
{\bf Judgment:}                            && {\bf Substitution into:}
                                               \bigstrut \\
\hline\hline\hline\hline
$\jhsubst {rr} {N_0} {x_0} {\alpha_0} R {R'}$
                                            && Atomic terms (yielding atomic)
                                            \bigstrut[t] \\
$\jhsubst {rn} {N_0} {x_0} {\alpha_0} R (N', \alpha')$
                                            && Atomic terms (yielding normal)
                                            \\
$\jhsubst n {N_0} {x_0} {\alpha_0} N {N'}$  && Normal terms \bigstrut[b] \\
\hline
$\jhsubst p {N_0} {x_0} {\alpha_0} P {P'}$  && Atomic types \bigstrut[t] \\
$\jhsubst a {N_0} {x_0} {\alpha_0} A {A'}$  && Normal types \bigstrut[b] \\
\hline
$\jhsubst q {N_0} {x_0} {\alpha_0} Q {Q'}$  && Atomic sorts \bigstrut[t] \\
$\jhsubst s {N_0} {x_0} {\alpha_0} S {S'}$  && Normal sorts \bigstrut[b] \\
\hline
$\jhsubst k {N_0} {x_0} {\alpha_0} K {K'}$  && Kinds    \bigstrut[t] \\
$\jhsubst l {N_0} {x_0} {\alpha_0} L {L'}$  && Classes  \bigstrut[b] \\
\hline
$\jhsubst \gamma {N_0} {x_0} {\alpha_0} \Gamma {\Gamma'}$ && Contexts \bigstrut
\\
\hline\hline\hline\hline
\end{tabular*}
\end{center}
\caption{Judgments defining hereditary substitution.}
\label{tab:hsubst}
\end{table}

\setlength{\arrayrulewidth}{\oldarrayrulewidth}

Our formulation of hereditary substitution is defined judgmentally by
inference rules.  The only place $\beta$-redexes might be introduced is
when substituting a normal term $N$ into an atomic term $R$: $N$ might be a
$\lambda$-abstraction, and the variable being substituted for may occur at
the head of $R$.  Therefore, the judgments defining substitution into
atomic terms are the most interesting ones.

We denote substitution into atomic terms by two judgments: $\jhsubst {rr}
{N_0} {x_0} {\alpha_0} R {R'}$, for when the head of $R$ is \emph{not}
$x_0$, and $\jhsubst {rn} {N_0} {x_0} {\alpha_0} R {(N', \alpha')}$, for when
the head of $R$ \emph{is} $x_0$, where $\alpha'$ is the simple type of the
output $N'$.  The former is just defined
compositionally; the latter is defined by two rules:
\begin{mathpar}
\namedaxiom {\jhsubst {rn} {N_0} {x_0} {\alpha_0} {x_0} {(N_0, \alpha_0)}}
            {subst-rn-var}
\and
\namedrule{\jhsubst {rn} {N_0} {x_0} {\alpha_0} {R_1}
                            {(\mlam x {N_1}, \aarrow {\alpha_2} {\alpha_1})}
        \\\\ \jhsubst n {N_0} {x_0} {\alpha_0} {N_2} {N_2'}
        \\ \jhsubst n {N_2'} x {\alpha_2} {N_1} {N_1'}}
          {\jhsubst {rn} {N_0} {x_0} {\alpha_0} {\mapp {R_1} {N_2}}
                                                {(N_1', \alpha_1)}}
          {subst-rn-$\beta$}
\end{mathpar}
The rule \rulename{subst-rn-var} just returns the substitutend $N_0$ and
its putative type index ${\alpha_0}$.  The rule \rulename{subst-rn-$\beta$}
applies when the result of substituting into the head of an application is
a $\lambda$-abstraction; it avoids creating a redex by hereditarily
substituting into the body of the abstraction.

A simple lemma establishes that these two judgments are mutually
exclusive by examining the head of the input atomic term.
\begin{align*}
\head(x) &= x
& \head(c) &= c
& \head(\mapp R N) &= \head(R)
\end{align*}
\begin{lem} ~
\label{lem:rr-rn-excl}
\begin{enumerate}[\em(1)]
\item If\/ $\jhsubst {rr} {N_0} {x_0} {\alpha_0} R {R'}$,
      then $\head(R) \neq x_0$.
\item If\/ $\jhsubst {rn} {N_0} {x_0} {\alpha_0} R {(N', \alpha')}$,
      then $\head(R) = x_0$.
\end{enumerate}
\end{lem}
\begin{proof}
By induction on the given derivation.
\end{proof}

\noindent Substitution into normal terms has two rules for atomic
terms $R$, one which calls the ``rr'' judgment and one which calls the
``{rn}'' judgment.
\begin{mathpar}
\namedrule{\jhsubst {rr} {N_0} {x_0} {\alpha_0} R {R'}}
          {\jhsubst n {N_0} {x_0} {\alpha_0} R {R'}}
          {subst-n-atom}
\and
\namedrule{\jhsubst {rn} {N_0} {x_0} {\alpha_0} R {(R', a')}}
          {\jhsubst n {N_0} {x_0} {\alpha_0} R {R'}}
          {subst-n-atom-norm}
\end{mathpar}
Note that the latter rule requires both the term and the type returned by
the ``{rn}'' judgment to be atomic.

Every other syntactic category's substitution judgment is defined
compositionally, tacitly renaming bound variables to avoid capture.  For
example, the remaining rule defining substitution into normal terms, the
rule for substituting into a $\lambda$-abstraction, just recurses on the
body of the abstraction.
\begin{mathpar}
\inferrule{\jhsubst n {N_0} {x_0} {\alpha_0} N {N'}}
          {\jhsubst n {N_0} {x_0} {\alpha_0} {\mlam x N} {\mlam x {N'}}}
\end{mathpar}


\noindent Although we have only defined hereditary substitution
relationally, it is easy to show that it is in fact a partial function
by proving that there only ever exists one ``output'' for a given set
of ``inputs''.

\begin{thm}[Functionality of Substitution]
\label{thm:functionality}
Hereditary substitution is a functional relation.  In particular:
\begin{enumerate}[\em(1)]
    \item If\/ $\jhsubst {rr} {N_0} {x_0} {\alpha_0} R {R_1}$
          and $\jhsubst {rr} {N_0} {x_0} {\alpha_0} R {R_2}$,
          then ${R_1} = {R_2}$,
    \item If\/ $\jhsubst {rn} {N_0} {x_0} {\alpha_0} R {({N_1}, {\alpha_1})}$
          and $\jhsubst {rn} {N_0} {x_0} {\alpha_0} R {({N_2}, {\alpha_2})}$,
          then ${N_1} = {N_2}$ and ${\alpha_1} = {\alpha_2}$,
    \item If\/ $\jhsubst n {N_0} {x_0} {\alpha_0} N {N_1}$
          and $\jhsubst n {N_0} {x_0} {\alpha_0} N {N_2}$,
          then ${N_1} = {N_2}$,
\end{enumerate}
and similarly for other syntactic categories.
\end{thm}
\begin{proof}
Straightforward induction on the first derivation, applying inversion to
the second derivation.  The cases for rules \rulename{subst-n-atom} and
\rulename{subst-n-atom-norm} require Lemma~\ref{lem:rr-rn-excl} to show
that the second derivation ends with the same rule as the first one.
\end{proof}

\noindent Additionally, it is worth noting that hereditary substitution 
behaves just like ``ordinary'' substitution on terms that do not contain the
distinguished free variable.
\begin{thm}[Trivial Substitution]
Hereditary substitution for a non-occurring variable has no effect.
    \begin{enumerate}[\em(1)]
    \item If\/ ${x_0} \not\in \FV(R)$,
        then $\jhsubst {rr} {N_0} {x_0} {\alpha_0} R R$,
    \item If\/ ${x_0} \not\in \FV(N)$,
        then $\jhsubst n {N_0} {x_0} {\alpha_0} N N$,
    \end{enumerate}
and similarly for other syntactic categories.
\end{thm}
\begin{proof}
Straightforward induction on term structure.
\end{proof}

\subsection{Decidability}
\label{sect:decidability}

A hallmark of the canonical forms/hereditary substitution
approach is that it allows a decidability proof to be carried out
comparatively early, before proving anything about the behavior of
substitution, and without dealing with any complications introduced by
$\beta$/$\eta$-conversions inside types.  Ordinarily in a dependently typed
calculus, one must first prove a substitution theorem before proving
typechecking decidable,
since typechecking relies on type equality, type
equality relies on $\beta$/$\eta$-conversion, and
$\beta$/$\eta$-conversions rely on substitution preserving well-formedness.
(See for example \cite{Harper05tocl} for a typical non-canonical forms-style
account of LF definitional equality.)

In contrast, if only canonical forms are permitted, then type equality is
just $\alpha$-con\-vert\-i\-bil\-i\-ty, so one only needs to show \emph{decidability}
of substitution in order to show decidability of typechecking.  Since LF
encodings represent judgments as type families and proof-checking as
typechecking, it is comforting to have a decidability proof that relies on
so few assumptions.

\begin{lem}
If\/ $\jhsubst {rn} {N_0} {x_0} {\alpha_0} R {(N', \alpha')}$,
then $\alpha'$ is a subterm of $\alpha_0$.
\label{lem:subst-rn-alpha-subterm}
\end{lem}
\begin{proof}
By induction on the derivation of $\jhsubst {rn} {N_0} {x_0} {\alpha_0} R
{(N', \alpha')}$.  In rule \rulename{subst-rn-var}, $\alpha'$ is the same
as $\alpha_0$.  In rule \rulename{subst-rn-$\beta$}, our inductive
hypothesis tells us that $\aarrow {\alpha_2} {\alpha_1}$ is a subterm of
$\alpha_0$, so $\alpha_1$ is as well.
\end{proof}

\noindent By working in a constructive metalogic, we are able to prove
decidability of a judgment by proving an instance of the law of the
excluded middle; the computational content of the proof then
represents a decision procedure.

\begin{thm}[Decidability of Substitution]
Hereditary substitution is decidable.  In particular:
    \begin{enumerate}[\em(1)]
     \item Given $N_0$, $x_0$, $\alpha_0$, and $R$, either
         $\exists R' \bdot \jhsubst {rr} {N_0} {x_0} {\alpha_0} R {R'}$, or 
         $\not\exists R' \bdot \jhsubst {rr} {N_0} {x_0} {\alpha_0} R {R'}$,
     \item Given $N_0$, $x_0$, $\alpha_0$, and $R$, either
         $\exists (N', \alpha') \bdot \jhsubst {rn} {N_0} {x_0} {\alpha_0}
         R {(N', \alpha')}$, or \\
         $\not\exists (N', \alpha') \bdot \jhsubst {rn} {N_0} {x_0} {\alpha_0}
         R {(N', \alpha')}$,
     \item Given $N_0$, $x_0$, $\alpha_0$, and $N$, either
        $\exists N' \bdot \jhsubst n {N_0} {x_0} {\alpha_0} N {N'}$, or
        $\not\exists N' \bdot \jhsubst n {N_0} {x_0} {\alpha_0} N {N'}$,
    \end{enumerate}
    and similarly for other syntactic categories.
\end{thm}
\begin{proof}
By lexicographic induction on the type subscript $\alpha_0$, the main
subject of the substitution judgment, and the clause number.  For each
applicable rule defining hereditary substitution, the premises are at a
smaller type subscript, or if the same type subscript, then a smaller term,
or if the same term, then an earlier clause.  The case for rule
\rulename{subst-rn-$\beta$} relies on
Lemma~\ref{lem:subst-rn-alpha-subterm} to know that $\alpha_2$ is a strict
subterm of $\alpha_0$.
\end{proof}

\begin{thm}[Decidability of Subsorting]
Given $Q_1$ and $Q_2$, either $\jchecksubtype \Gamma {Q_1} {Q_2} P$ or
$\jnsubtype {Q_1} {Q_2}$.
\end{thm}
\begin{proof}
Since the subsorting relation $\jchecksubtype \Gamma {Q_1} {Q_2} P$ is just
the reflexive, transitive closure of the declared subsorting relation $s_1
\subtype s_2$, it suffices to compute this closure, check that the heads of
$Q_1$ and $Q_2$ are related by it, and ensure that all of the arguments of
$Q_1$ and $Q_2$ are equal.
\end{proof}

We prove decidability of typing by exhibiting a deterministic algorithmic
system that is equivalent to the original.  Instead of synthesizing a
single sort for an atomic term, the algorithmic system synthesizes an
intersection-free list of sorts, $\Delta$.
\begin{align*}
  \Delta &::= \cdot \mid \Delta, Q \mid \Delta, \sspi x S A T
\end{align*}
(As usual, we freely overload comma to mean list concatenation, as no
ambiguity can result.)
One can think of $\Delta$ as the intersection of all its elements.  Instead
of applying intersection eliminations, the algorithmic system eagerly breaks
down intersections using a ``split'' operator, leading to a deterministic
``minimal-synthesis'' system.
\begin{align*}
  \ssplit Q &= Q
& \ssplit {\sinter {S_1} {S_2}} &= \ssplit {S_1}, \ssplit {S_2} \\
  \ssplit {\sspi x S A T} &= \sspi x S A T
& \ssplit \sstop &= \cdot
\end{align*}
\begin{mathpar}
\inferrule{\mrefdecl c S \in \Sigma
          }
          {\jasynth \Gamma c {\ssplit S}}
\and
\inferrule{\rdecl x S A \in \Gamma}
          {\jasynth \Gamma x {\ssplit S}}
\and
\inferrule{\jasynth \Gamma R \Delta
        \\ \under \Gamma {\japply \Delta N {\Delta'}}}
          {\jasynth \Gamma {\mapp R N} {\Delta'}}
\end{mathpar}
The rule for applications uses an auxiliary judgment $\under \Gamma
\japply \Delta N {\Delta'}$ which computes the possible types of $\mapp R
N$ given that $R$ synthesizes to all the sorts in $\Delta$.  It has
two key rules:
\begin{mathpar}
\inferaxiom {\under \Gamma {\japply \cdot N \cdot}}
\and
\inferrule{\under \Gamma \japply \Delta N {\Delta'}
        \\ \jacheck \Gamma N S
        \\ \jhsubst s N x A T {T'}}
          {\under \Gamma \japply {(\Delta, \sspi x S A T)} N
                                 {\Delta', \ssplit {T'}}}
\end{mathpar}
The other rules force the judgment to be defined when neither of the above
two rules apply.
\begin{mathpar}
\inferrule{\under \Gamma \japply \Delta N {\Delta'}
        \\ \jnacheck \Gamma N S}
          {\under \Gamma \japply {(\Delta, \sspi x S A T)} N {\Delta'}}
\and
\inferrule{\under \Gamma \japply \Delta N {\Delta'}
        \\ \not\exists {T'} \bdot \jhsubst s N x A T {T'}}
          {\under \Gamma \japply {(\Delta, \sspi x S A T)} N {\Delta'}}
\and
\inferrule{\under \Gamma \japply \Delta N {\Delta'}}
          {\under \Gamma \japply {(\Delta, Q)} N {\Delta'}}
\end{mathpar}
Finally, to tie everything together, we define a new checking judgment
$\jacheck \Gamma N S$ that makes use of the algorithmic synthesis
judgment; it looks just like $\jchecktermtype \Gamma N S$ except for the
rule for atomic terms.
\begin{mathpar}
\inferrule{\jasynth \Gamma R \Delta
        \\ Q' \in \Delta
        \\ {Q'} \subtype Q}
          {\jacheck \Gamma R Q}
\and
\inferrule{\jacheck {\Gamma, \rdecl x S A} N T}
          {\jacheck \Gamma {\mlam x N} {\sspi x S A T}}
\\
\and
\inferaxiom {\jacheck \Gamma N \sstop}
\and
\inferrule{\jacheck \Gamma N {S_1}
        \\ \jacheck \Gamma N {S_2}}
          {\jacheck \Gamma N {\sinter {S_1} {S_2}}}
\end{mathpar}

\noindent This new algorithmic system is manifestly decidable:
despite the negative conditions in some of the premises, the definitions of
the judgments are well-founded by the ordering used in the following proof.
(If we wished, we could also explicitly synthesize a definition of $\jnacheck
\Gamma N S$, but it would not illuminate the algorithm any further.)
\begin{thm}
\label{thm:checking-decidable}
Algorithmic sort checking is decidable.  In particular:
\begin{enumerate}[\em(1)]
\item Given $\Gamma$ and $R$, either
    $\exists \Delta \bdot \jasynth \Gamma R \Delta$ or
    $\not\exists \Delta \bdot \jasynth \Gamma R \Delta$.
\item Given $\Gamma$, $N$, and $S$, either $\jacheck \Gamma N S$ or
    $\jnacheck \Gamma N S$.
\item Given $\Gamma$, $\Delta$, and $N$, $\exists \Delta' \bdot
    \under \Gamma \japply \Delta N {\Delta'}$.
\label{cl:japply}
\end{enumerate}
\end{thm}
\begin{proof}
By lexicographic induction on the term $R$ or $N$, the clause number, and
the sort $S$ or the list of sorts $\Delta$.  For each applicable rule, the
premises are either known to be decidable, or at a smaller term, or if the
same term, then an earlier clause, or if the same clause, then either a
smaller $S$ or a smaller $\Delta$.  For clause \ref{cl:japply}, we must use
our inductive hypothesis to argue that the rules cover all possibilities,
and so a derivation always exists.
\end{proof}
\noindent Note that the algorithmic synthesis system sometimes outputs
an empty $\Delta$ even when the given term is ill-typed, since the
$\under \Gamma \japply \Delta N {\Delta'}$ judgment is always defined.

It is straightforward to show that the algorithm is sound and complete with
respect to the original bidirectional system.
\begin{lem}
\label{lem:synth-split}
If\/ $\jsynthtermtype \Gamma R S$, then for all $S' \in \ssplit S$,
$\jsynthtermtype \Gamma R {S'}$.
\end{lem}
\begin{proof}
By induction on $S$, making use of the
\rulename{$\intersect$-E$_1$} and \rulename{$\intersect$-E$_2$} rules.
\end{proof}\medskip

\begin{thm}[Soundness of Algorithmic Typing] ~
\label{thm:sound-typing}
\begin{enumerate}[\em(1)]
\item If\/ $\jasynth \Gamma R \Delta$, then for all $S \in \Delta$,
    $\jsynthtermtype \Gamma R S$.
\item If\/ $\jacheck \Gamma N S$, then $\jchecktermtype \Gamma N S$.
\item If\/ $\under \Gamma \japply \Delta N {\Delta'}$,
    and for all $S \in \Delta$, $\jsynthtermtype \Gamma R S$,
    then for all $S' \in \Delta'$, $\jsynthtermtype \Gamma {\mapp R N} {S'}$.
\end{enumerate}
\end{thm}
\begin{proof}
By induction on the given derivation, using Lemma~\ref{lem:synth-split}.
\end{proof}

\noindent For completeness, we use the notation $\Delta \subseteq
\Delta'$ to mean that $\Delta$ is a sublist of $\Delta'$.

\begin{lem}
\label{lem:apply}
If\/ $\under \Gamma \japply \Delta N {\Delta'}$
    and $\jasynth \Gamma R \Delta$
    and $\sspi x S A T \in \Delta$
    and $\jacheck \Gamma N S$
    and $\jhsubst s N x A T {T'}$,
    then $\ssplit {T'} \subseteq \Delta'$.
\end{lem}
\begin{proof}
By straightforward induction on the derivation of $\under \Gamma \japply
\Delta N {\Delta'}$.
\end{proof}\medskip

\begin{thm}[Completeness for Algorithmic Typing] ~
\label{thm:complete-typing}
\begin{enumerate}[\em(1)]
\item If\/ $\jsynthtermtype \Gamma R S$, then $\jasynth \Gamma R \Delta$
    and $\ssplit S \subseteq \Delta$.
\item If\/ $\jchecktermtype \Gamma N S$, then $\jacheck \Gamma N S$.
\end{enumerate}
\end{thm}

\begin{proof}
By straightforward induction on the given derivation.  In the application
case, we make use of the fact that $\under \Gamma \japply \Delta N
{\Delta'}$ is always defined and apply Lemma~\ref{lem:apply}.
\end{proof}

\noindent Soundness, completeness, and decidability of the algorithmic
system gives us a decision procedure for the judgment $\jchecktermtype
\Gamma N S$.  First, decidability tells us that either $\jacheck
\Gamma N S$ or $\jnacheck \Gamma N S$.  Then soundness tells us that
if $\jacheck \Gamma N S$ then $\jchecktermtype \Gamma N S$, while
completeness tells us that if $\jnacheck \Gamma N S$ then
$\jnchecktermtype \Gamma N S$.

Decidability theorems and proofs for other syntactic categories' formation
judgments proceed similarly.  When all is said and done, we have enough to
show that the problem of sort checking an LFR signature is decidable.

\begin{thm}[Decidability of Sort Checking]  Sort checking is decidable.
In particular:
\begin{enumerate}[\em(1)]
    \item Given $\Gamma$, $N$, and $S$, either $\jchecktermtype \Gamma N S$
    or $\jnchecktermtype \Gamma N S$,
    \label{cl:term-checking-decidable}
    \item Given $\Gamma$, $S$, and $A$, either $\jchecksortclass \Gamma S
    A$ or $\jnchecksortclass \Gamma S A$, and
    \item Given $\Sigma$, either $\jsig \Sigma$ or $\jnsig \Sigma$.
\end{enumerate}
\end{thm}

\subsection{Identity and Substitution Principles}

Since well-typed terms in our framework must be canonical, that is
$\beta$-normal and $\eta$-long, it is non-trivial to prove $\aarrow S S$
for non-atomic $S$, or to compose proofs of $\aarrow {S_1} {S_2}$ and
$\aarrow {S_2} {S_3}$.  The Identity and Substitution principles ensure
that our type theory makes logical sense by demonstrating the reflexivity
and transitivity of entailment.  Reflexivity is witnessed by
$\eta$-expansion, while transitivity is witnessed by hereditary
substitution.

The Identity principle effectively says that synthesizing (atomic) objects
can be made to serve as checking (normal) objects.  The Substitution
principle dually says that checking objects may stand in for synthesizing
assumptions, that is, variables.

\subsubsection{Substitution}

The goal of this section is to give a careful proof of the following
substitution theorem.
   Suppose $\jchecktermsort \GammaL {N_0} {S_0}$~.  Then:
   \begin{enumerate}[(1)]
   \item
        If
          \begin{enumerate}[$\bullet$]
          \item
          $\jctx {\GammaL, \rdecl {x_0} {S_0} {A_0}, \GammaR}$~, and
          \item
          $\jchecksortclass {\GammaL, \rdecl {x_0} {S_0} {A_0}, \GammaR} S A$~,
          and
          \item
          $\jchecktermsort {\GammaL, \rdecl {x_0} {S_0} {A_0}, \GammaR} N S$~,
          \end{enumerate}
          then
          \begin{enumerate}[$\bullet$]
           \item 
                 $\jhsubst \gamma {N_0} {x_0} {A_0} \GammaR {\GammaR'}$ and
                 $\jctx {\GammaL, \GammaR'}$~, and
           \item 
                 $\jhsubst s {N_0} {x_0} {A_0} S {S'}$ and
                 $\jhsubst a {N_0} {x_0} {A_0} A {A'}$ and
                 $\jchecksortclass {\GammaL, \GammaR'} {S'} {A'}$~, and
           \item 
                 $\jhsubst n {N_0} {x_0} {A_0} N {N'}$ and
                 $\jchecktermsort {\GammaL, \GammaR'} {N'} {S'}$~,
          \end{enumerate}
    \item
        If
            \begin{enumerate}[$\bullet$]
            \item $\jctx {\GammaL, \rdecl {x_0} {S_0} {A_0}, \GammaR}$~ and
            \item $\jsynthtermtype {\GammaL, \rdecl {x_0} {S_0} {A_0}, \GammaR}
                    R S~$,
            \end{enumerate}
            then
            \begin{enumerate}[$\bullet$]
            \item $\jhsubst \gamma {N_0} {x_0} {A_0} \GammaR {\GammaR'}$ and
                  $\jctx {\GammaL, \GammaR'}$~, and
                  $\jhsubst s {N_0} {x_0} {A_0} S {S'}$~, and either
                \begin{enumerate}[$-$]
                    \item $\jhsubst {rr} {N_0} {x_0} {A_0} R {R'}$ and
                          $\jsynthtermtype {\GammaL, \GammaR'} {R'} {S'}$~,
                          or
                    \item $\jhsubst {rn} {N_0} {x_0} {A_0} R {(N', \alpha')}$
                          and
                          $\jchecktermtype {\GammaL, \GammaR'} {N'} {S'}$~,
                \end{enumerate}
            \end{enumerate}
   \end{enumerate}
    and similarly for other syntactic categories.
    (\textbf{Theorem~\ref{thm:substitution}} below.)

To prove the substitution theorem, we require a lemma about how
substitutions compose.  The corresponding property for a ordinary non-hereditary
substitution says that $\subst {N_0} {x_0} {\subst {N_2} {x_2} N} = \subst
{\subst {N_0} {x_0} {N_2}} {x_2} {\subst {N_0} {x_0} N}$.  For hereditary
substitutions, the situation is analogous, but we must be clear about which
substitution instances we must assume to be defined and which we may
conclude to be defined:  If the three ``inner'' substitutions are defined,
then the two ``outer'' ones are also defined, and equal.  Note that the
composition lemma is something like a diamond property; the notation below
is meant to suggest this connection.
\newcommand\pr[1]{#1^\prime}
\newcommand\phat[1]{\hat{\pr {#1}}}
\newcommand\bp[1]{#1^\backprime}
\newcommand\bpp[1]{#1^{\backprime\prime}}
\newcommand\bpphat[1]{\hat{\bpp {#1}}}
\begin{lem}[Composition of Substitutions]
  \label{lem:composition}
  Suppose {$\jhsubst n {N_0} {x_0} {\alpha_0} {N_2} {\bp {N_2}}$} and
  ${x_2} \not\in \FV({N_0})$.  Then:
  \begin{enumerate}[\em(1)]
    \item If\/ {$\jhsubst n {N_0} {x_0} {\alpha_0} N {\bp N}$}
           and {$\jhsubst n {N_2} {x_2} {\alpha_2} N {\pr N}$},
          then for some $\bpp N$, \\
               {$\jhsubst n {\bp {N_2}} {x_2} {\alpha_2} {\bp N} {\bpp N}$}
           and {$\jhsubst n {N_0} {x_0} {\alpha_0} {\pr N} {\bpp N}$}~,
    \label{cl:n-n}

    \item If\/ {$\jhsubst {rr} {N_0} {x_0} {\alpha_0} R {\bp R}$}
           and {$\jhsubst {rr} {N_2} {x_2} {\alpha_2} R {\pr R}$},
          then for some $\bpp R$, \\
               {$\jhsubst {rr} {\bp N_2} {x_2} {\alpha_2} {\bp R} {\bpp R}$}
           and {$\jhsubst {rr} {N_0} {x_0} {\alpha_0} {\pr R} {\bpp R}$}~,
    \label{cl:rr-rr}

    \item If\/ {$\jhsubst {rr} {N_0} {x_0} {\alpha_0} R {\bp R}$}
           and {$\jhsubst {rn} {N_2} {x_2} {\alpha_2} R {(\pr N, \beta)}$},
          then for some $\bpp N$, \\
               {$\jhsubst {rn} {\bp N_2} {x_2} {\alpha_2} {\bp R}
                               {(\bpp N, \beta)}$}
           and {$\jhsubst n {N_0} {x_0} {\alpha_0} {\pr N} {\bpp N}$}~,
    \label{cl:rr-rn}

    \item If\/ {$\jhsubst {rn} {N_0} {x_0} {\alpha_0} R {(\bp N, \beta)}$}
           and {$\jhsubst {rr} {N_2} {x_2} {\alpha_2} R {\pr R}$},
          then for some $\bpp N$, \\
               {$\jhsubst {n} {\bp N_2} {x_2} {\alpha_2} {\bp N} {\bpp N}$}
           and {$\jhsubst {rn} {N_0} {x_0} {\alpha_0} {\pr R}
                               {(\bpp N, \beta)}$}~,
    \label{cl:rn-rr}
  \end{enumerate}
and similarly for other syntactic categories.
\end{lem}

\begin{proof}[Proof (sketch)]
By lexicographic induction on the unordered pair of ${\alpha_0}$ and
${\alpha_2}$, and on the first substitution derivation in each clause.
The cases
for rule \rulename{subst-rn-$\beta$}
in clauses \ref{cl:rr-rn} and \ref{cl:rn-rr}
appeal to the induction hypothesis at a smaller type using
Lemma~\ref{lem:subst-rn-alpha-subterm}.  The case in clause \ref{cl:rn-rr}
swaps the roles of ${\alpha_0}$ and ${\alpha_2}$, necessitating the
unordered induction metric.
\end{proof}

We also require a simple lemma about substitution into subsorting
derivations:
\begin{lem}[Substitution into Subsorting]
\label{lem:subst-subtype}
If\/ $\jsubtype {Q_1} {Q_2}$
and $\jhsubst q {N_0} {x_0} {\alpha_0} {Q_1} {Q_1'}$
and $\jhsubst q {N_0} {x_0} {\alpha_0} {Q_2} {Q_2'}$,
then $\jsubtype {Q_1'} {Q_2'}$.
\end{lem}
\begin{proof}
Straightforward induction using Theorem~\ref{thm:functionality}
(Functionality of Substitution), since the subsorting rules depend only on
term equalities, and not on well-formedness.
\end{proof}

Next, we must state the substitution theorem in a form general enough to
admit an inductive proof.  Following previous work on canonical forms-based
LF \cite{Watkins02tr,hl07mechanizing}, we strengthen its statement to one
that does not presuppose the well-formedness of the context or the
classifying types, but instead merely presupposes that hereditary
substitution is defined on them.  We call this strengthened theorem
``proto-substitution'' and prove it in several parts.  In order to capture
the convention that we only sort-check well-typed terms, proto-substitution
includes hypotheses about well-typedness of terms; these hypotheses use an
erasure $\eraserefs \Gamma$ that transforms an LFR context into an LF
context.
\begin{align*}
    \eraserefs \cdot &= \cdot &
    \eraserefs {(\Gamma, \rdecl x S A)} &= \eraserefs \Gamma, \decl x A
\end{align*}
The structure of the proof under this convention requires that we
interleave the proof of the core LF proto-substitution theorem.
Generally, reasoning related to core LF presuppositions is analogous to
refinement-related reasoning and can be dealt with mostly orthogonally, but
the presuppositions are necessary in certain cases.

\begin{thm}[Proto-Substitution, terms] ~
    \label{thm:proto-terms}
    \begin{enumerate}[\em(1)]
    \item If
        \begin{enumerate}[$\bullet$]
            \item $\jchecktermsort \GammaL {N_0} {S_0}$
                  \boxassum{and $\jchecktermtype {\eraserefs{\GammaL}}
                                  {N_0} {A_0}$}~,
                  and
            \item $\jchecktermsort {\GammaL, \rdecl {x_0} {S_0} {A_0}, \GammaR}
                    N S$
                  \boxassum{and $\jchecktermtype
                                    {\eraserefs\GammaL, \decl {x_0} {A_0},
                                     \eraserefs\GammaR} N A$}~,
                  and
            \item $\jhsubst \gamma {N_0} {x_0} {A_0} \GammaR {\bp \GammaR}$~, and
            \item $\jhsubst s {N_0} {x_0} {A_0} S {\bp S}$~
                  \boxassum{and $\jhsubst a {N_0} {x_0} {A_0} A {\bp A}$}~,
        \end{enumerate}
        then
        \begin{enumerate}[$\bullet$]
            \item $\jhsubst n {N_0} {x_0} {A_0} N {\bp N}$~, and
            \item $\jchecktermsort {\GammaL, \bp \GammaR} {\bp N} {\bp S}$
                \boxassum{and $\jchecktermtype {\eraserefs\GammaL,
                                                \eraserefs{(\bp\GammaR)}}
                                            {\bp N} {\bp A}$}~.
        \end{enumerate}
    \label{cl:proto-n}
    \item If
        \begin{enumerate}[$\bullet$]
            \item $\jchecktermsort \GammaL {N_0} {S_0}$
                \boxassum{and $\jchecktermtype {\eraserefs\GammaL} {N_0} {A_0}$}~,
                and
            \item $\jsynthtermsort {\GammaL, \rdecl {x_0} {S_0} {A_0}, \GammaR}
                    R S$~
                \boxassum{and $\jsynthtermtype {\eraserefs\GammaL,
                                                \decl {x_0} {A_0},
                                                \eraserefs\GammaR} R A$}~, and
            \item $\jhsubst \gamma {N_0} {x_0} {A_0} \GammaR {\bp \GammaR}$~,
        \end{enumerate}
        then
        \begin{enumerate}[$\bullet$]
            \item $\jhsubst s {N_0} {x_0} {A_0} S {\bp S}$
                \boxassum{and $\jhsubst a {N_0} {x_0} {A_0} A {\bp A}$~}, and
            \item either
                \begin{enumerate}[$-$]
                    \item $\jhsubst {rr} {N_0} {x_0} {A_0} R {\bp R}$ and
                    \item $\jsynthtermsort {\GammaL, \bp \GammaR} {\bp R}
                        {\bp S}$~
                        \boxassum{and $\jsynthtermtype
                                            {\eraserefs\GammaL,
                                             \eraserefs{(\bp\GammaR)}}
                                            {\bp R} {\bp A}$},
                \end{enumerate}
                  or
                \begin{enumerate}[$-$]
                    \item $\jhsubst {rn} {N_0} {x_0} {A_0} R {(\bp N,
                    \erasedeps{\bp A})}$ and 
                    \item $\jchecktermsort {\GammaL, \bp \GammaR} {\bp N}
                        {\bp S}$
                        \boxassum{and $\jchecktermtype
                                            {\eraserefs\GammaL,
                                             \eraserefs{(\bp\GammaR)}}
                                            {\bp N} {\bp A}$}~.
                \end{enumerate}
        \end{enumerate}
    \label{cl:proto-r}
    \end{enumerate}
\end{thm}

\noindent \textbf{Note:} We tacitly assume the implicit signature $\Sigma$ is
well-formed.  We do \emph{not} tacitly assume that any of the contexts,
sorts, or types are well-formed.  We \emph{do} tacitly assume that contexts
respect the usual variable conventions in that bound variables are always
fresh, both with respect to other variables bound in the same context and
with respect to other free variables in terms outside the scope of the
binding.

\begin{proof}[Proof (sketch)]
By lexicographic induction on $\erasedeps {A_0}$ and the derivation $\D$
hypothesizing $\rdecl {x_0} {S_0} {A_0}$.

The most involved case is that for application $\mapp {R_1} {N_2}$.  When
$\head(R_1) = {x_0}$ hereditary substitution carries out a
$\beta$-reduction, and the proof invokes the induction hypothesis at a
smaller type but not a subderivation.  This case also requires
Lemma~\ref{lem:composition} (Composition): since function sorts are
dependent, the typing rule for application carries out a substitution, and
we need to compose this substitution with the $\hsubst s {N_0} {x_0}
{\alpha_0}$ substitution.

In the case where we check a term at sort $\sstop$, we require the core LF 
assumptions in order to invoke the core LF proto-substitution theorem.
%
\end{proof}

Next, we can prove analogous proto-substitution theorems for sorts/types
and for classes/kinds.

\begin{thm}[Proto-Substitution, sorts and types] ~
\label{thm:proto-types}
    \begin{enumerate}[\em(1)]
    \item If
        \begin{enumerate}[$\bullet$]
            \item $\jchecktermtype \GammaL {N_0} {S_0}$
                  \boxassum{and $\jchecktermtype {\eraserefs\GammaL}
                                {N_0} {A_0}$}~,
            \item $\jchecksortclass {\GammaL, \rdecl {x_0} {S_0} {A_0}, \GammaR}
                                S A$
                  \boxassum{and $\jchecktypekind {\eraserefs\GammaL,
                                    \decl {x_0} {A_0}, \eraserefs\GammaR} A$}~,
            and
            \item $\jhsubst \gamma {N_0} {x_0} {A_0} \GammaR {\bp\GammaR}$~,
        \end{enumerate}
        then
        \begin{enumerate}[$\bullet$]
            \item $\jhsubst s {N_0} {x_0} {A_0} S {\bp S}$
                  \boxassum{and $\jhsubst a {N_0} {x_0} {A_0} A {\bp A}$}~,
                  and
            \item $\jchecksortclass {\GammaL, \bp\GammaR} {\bp S} {\bp A}$~,
                  \boxassum{and $\jchecktypekind {\eraserefs\GammaL,
                                    \eraserefs{(\bp\GammaR)}} {\bp A}$}~.
        \end{enumerate}

    \item If
        \begin{enumerate}[$\bullet$]
            \item $\jchecktermtype \GammaL {N_0} {S_0}$
                  \boxassum{and $\jchecktermtype {\eraserefs\GammaL}
                                {N_0} {A_0}$}~,
            \item $\jsynthsortclass {\GammaL, \rdecl {x_0} {S_0} {A_0}, \GammaR}
                                Q P L K$
                  \boxassum{and $\jsynthtypekind \Gamma P K$}~, and
            \item $\jhsubst \gamma {N_0} {x_0} {A_0} \GammaR {\bp\GammaR}$~,
        \end{enumerate}
        then
        \begin{enumerate}[$\bullet$]
            \item $\jhsubst q {N_0} {x_0} {A_0} Q {\bp Q}$
                  \boxassum{and $\jhsubst p {N_0} {x_0} {A_0} P {\bp P}$}~,
                  and
            \item $\jhsubst l {N_0} {x_0} {A_0} L {\bp L}$
                  \boxassum{and $\jhsubst k {N_0} {x_0} {A_0} K {\bp K}$}~,
                  and
            \item $\jsynthsortclass {\GammaL, \bp\GammaR} {\bp Q} {\bp P}
                        {\bp L} {\bp K}$
                  \boxassum{and $\jsynthtypekind {\eraserefs\GammaL,
                                    \eraserefs{(\bp\GammaR)}}
                                    {\bp P} {\bp K}$}~.
        \end{enumerate}
    \end{enumerate}
\end{thm}

\begin{proof}
By induction on the derivation hypothesizing $\rdecl {x_0} {S_0} {A_0}$,
using Theorem~\ref{thm:proto-terms} (Proto-Substitution, terms).  The
reasoning is essentially the same as the reasoning for
Theorem~\ref{thm:proto-terms}.
\end{proof}

\begin{thm}[Proto-Substitution, classes and kinds]
\label{thm:proto-kinds}
    ~ \\
    If
        \begin{enumerate}[$\bullet$]
            \item $\jchecktermtype \GammaL {N_0} {S_0}$
                  \boxassum{and $\jchecktermtype {\eraserefs\GammaL} {N_0} {A_0}$}~,
            \item $\jclass {\GammaL, \rdecl {x_0} {S_0} {A_0}, \GammaR} L K$
                  \boxassum{and $\jkind {\eraserefs\GammaL,
                                    \decl {x_0} {A_0}, \eraserefs\GammaR} K$}~,
            and
            \item $\jhsubst \gamma {N_0} {x_0} {A_0} \GammaR {\bp\GammaR}$~,
        \end{enumerate}
        then
        \begin{enumerate}[$\bullet$]
            \item $\jhsubst l {N_0} {x_0} {A_0} L {\bp L}$
                  \boxassum{and $\jhsubst k {N_0} {x_0} {A_0} K {\bp K}$}~,
                  and
            \item $\jclass {\GammaL, \bp\GammaR} {\bp L} {\bp K}$~,
                  \boxassum{and $\jkind {{\eraserefs\GammaL},
                                    \eraserefs{(\bp\GammaR)}} {\bp K}$}~.
        \end{enumerate}
\end{thm}

\begin{proof}
By induction on the derivation hypothesizing $\rdecl {x_0} {S_0} {A_0}$,
using Theorem~\ref{thm:proto-types} (Proto-Substitution, sorts and types).
\end{proof}

Then, we can finish proto-substitution by proving a proto-substitution
theorem for contexts.

\pagebreak[2]
\begin{thm}[Proto-Substitution, contexts]
\label{thm:proto-ctxs}
~\\
If
    \begin{enumerate}[$\bullet$]
        \item $\jchecktermtype \GammaL {N_0} {S_0}$
              \boxassum{and $\jchecktermtype {\eraserefs\GammaL} {N_0} {A_0}$}~,
        and
        \item $\jctx {\GammaL, \rdecl {x_0} {S_0} {A_0}}$
              \boxassum{and $\jctx {\eraserefs\GammaL, \decl {x_0} {A_0},
                                    \eraserefs\GammaR}$}~,
    \end{enumerate}
then
    \begin{enumerate}[$\bullet$]
        \item $\jhsubst \gamma {N_0} {x_0} {A_0} \GammaR {\bp\GammaR}$~,
        and
        \item $\jctx {\GammaL, \bp\GammaR}$
              \boxassum{and $\jctx {\eraserefs\GammaL,
                                    \eraserefs{(\bp\GammaR)}}$}~.
    \end{enumerate}
\end{thm}

\begin{proof}
Straightforward induction on $\GammaR$.
\end{proof}

Finally, we have enough obtain a proof of the desired substitution theorem.

\begin{thm}[Substitution]
\label{thm:substitution}
   Suppose $\jchecktermsort \GammaL {N_0} {S_0}$~.  Then:
   \begin{enumerate}[\em(1)]
   \item
        If
          \begin{enumerate}[$\bullet$]
          \item
          $\jctx {\GammaL, \rdecl {x_0} {S_0} {A_0}, \GammaR}$~, and
          \item
          $\jchecksortclass {\GammaL, \rdecl {x_0} {S_0} {A_0}, \GammaR} S A$~,
          and
          \item
          $\jchecktermsort {\GammaL, \rdecl {x_0} {S_0} {A_0}, \GammaR} N S$~,
          \end{enumerate}
          then
          \begin{enumerate}[$\bullet$]
           \item 
                 $\jhsubst \gamma {N_0} {x_0} {A_0} \GammaR {\GammaR'}$ and
                 $\jctx {\GammaL, \GammaR'}$~, and
           \item 
                 $\jhsubst s {N_0} {x_0} {A_0} S {S'}$ and
                 $\jhsubst a {N_0} {x_0} {A_0} A {A'}$ and
                 $\jchecksortclass {\GammaL, \GammaR'} {S'} {A'}$~, and
           \item 
                 $\jhsubst n {N_0} {x_0} {A_0} N {N'}$ and
                 $\jchecktermsort {\GammaL, \GammaR'} {N'} {S'}$~,
          \end{enumerate}
    \item
        If
            \begin{enumerate}[$\bullet$]
            \item $\jctx {\GammaL, \rdecl {x_0} {S_0} {A_0}, \GammaR}$~ and
            \item $\jsynthtermtype {\GammaL, \rdecl {x_0} {S_0} {A_0}, \GammaR}
                    R S~$,
            \end{enumerate}
            then
            \begin{enumerate}[$\bullet$]
            \item $\jhsubst \gamma {N_0} {x_0} {A_0} \GammaR {\GammaR'}$ and
                  $\jctx {\GammaL, \GammaR'}$~, and
                  $\jhsubst s {N_0} {x_0} {A_0} S {S'}$~, and either
                \begin{enumerate}[$-$]
                    \item $\jhsubst {rr} {N_0} {x_0} {A_0} R {R'}$ and
                          $\jsynthtermtype {\GammaL, \GammaR'} {R'} {S'}$~,
                          or
                    \item $\jhsubst {rn} {N_0} {x_0} {A_0} R {(N', \alpha')}$
                          and
                          $\jchecktermtype {\GammaL, \GammaR'} {N'} {S'}$~,
                \end{enumerate}
            \end{enumerate}
   \end{enumerate}
    and similarly for other syntactic categories.
\end{thm}
\begin{proof}
Straightforward corollary of Proto-Substitution Theorems~\ref{thm:proto-terms},
\ref{thm:proto-types}, \ref{thm:proto-kinds}, and \ref{thm:proto-ctxs}.
\end{proof}

\noindent Having proven substitution, we henceforth tacitly assume that all subjects
of a judgment are sufficiently well-formed for the judgment to make
sense.  In particular, we assume that all contexts are well-formed, and
whenever we assume $\jchecktermtype \Gamma N S$, we assume that for some
well-formed type $A$, we have $\jchecksortclass \Gamma S A$ and
$\jchecktermtype \Gamma N A$.  These assumptions embody our refinement
restriction: we only sort-check a term if it is already well-typed and even
then only at sorts that refine its type.

Similarly, whenever we assume $\jchecksortclass \Gamma S A$, we tacitly
assume that $\jchecktypekind \Gamma A$, and whenever we assume
$\jclass \Gamma L K$, we tacitly assume that $\jkind \Gamma K$.


\subsubsection{Identity}

Just as we needed a composition lemma to prove the substitution theorem, in
order to prove the identity theorem we need a lemma about how
$\eta$-expansion commutes with substitution.\footnote{The
categorically-minded reader might think of this as the right and left unit
laws for $\circ$ while thinking of the composition lemma above as the
associativity of $\circ$, where $\circ$ in the category represents
substitution, as usual.}

In stating this lemma, we require a judgment that predicts the simple type
output of ``{rn}'' substitution.  This judgment just computes the simple
type as in ``{rn}'' substitution, but without computing anything having to
do with substitution.  Since it resembles a sort of ``approximate typing
judgment'', we write it $\treduce {x_0} {\alpha_0} R \alpha$.  As with
``{rn}'' substitution, it is only defined when the head of $R$ is ${x_0}$.%
\begin{mathpar}
  \inferaxiom{\treduce {x_0} {\alpha_0} {x_0} {\alpha_0}}
  \and
  \inferrule{\treduce {x_0} {\alpha_0} R {\aarrow \alpha \beta}}
            {\treduce {x_0} {\alpha_0} {\mapp R N} \beta}
\end{mathpar}

\begin{lem}
\label{lem:subst-treduce}
  If\/ $\jhsubst {rn} {N_0} {x_0} {\alpha_0} R {(N', \alpha')}$
  and $\treduce {x_0} {\alpha_0} R \alpha$, then $\alpha' = \alpha$.
\end{lem}
\begin{proof}
Straightforward induction.
\end{proof}

\begin{lem}[Commutativity of Substitution and $\eta$-expansion]
Substitution commutes with $\eta$-expansion.  In particular:
  \label{lem:commutativity}
  ~
  \begin{enumerate}[\em(1)]
    \item
    \begin{enumerate}[\em(a)]
      \item If\/ $\jhsubst n {\expand \alpha x} x \alpha N {N'}$,
            then $N = N'$~,
            \label{cl:left-n}
      \item If\/ $\jhsubst{rr} {\expand \alpha x} x \alpha R {R'}$,
            then $R = R'$~,
            \label{cl:left-rr}
      \item If\/ $\jhsubst{rn} {\expand \alpha x} x \alpha R {(N, \beta)}$,
            then $\expand {\beta} R = N$~,
            \label{cl:left-rn}
    \end{enumerate}
    \item If\/ $\jhsubst n {N_0} {x_0} {\alpha_0} {\expand \alpha R} {N'}$, then
    \begin{enumerate}[\em(a)]
      \item if\/ $\head(R) \neq {x_0}$, then
            $\jhsubst {rr} {N_0} {x_0} {\alpha_0} R {R'}$ and
            $\expand \alpha {R'} = N'$~,
            \label{cl:right-rr}
      \item if\/ $\head(R) = {x_0}$ and
            $\treduce {x_0} {\alpha_0} R \alpha$, then
            $\jhsubst {rn} {N_0} {x_0} {\alpha_0} R {(N', \alpha)}$~,
            \label{cl:right-rn}
    \end{enumerate}
  \end{enumerate}
and similarly for other syntactic categories.
\end{lem}

\begin{proof}[Proof (sketch)]
By lexicographic induction on $\alpha$ and the given substitution
derivation.  The proofs of clauses~\ref{cl:left-n}, \ref{cl:left-rr}, and
\ref{cl:left-rn} analyze the substitution derivation, while the proofs of
clauses~\ref{cl:right-rr} and \ref{cl:right-rn} analyze the simple type
$\alpha$ at which $R$ is $\eta$-expanded.
\end{proof}

\noindent \textbf{Note:} By considering the variable being substituted
for to be a bound variable subject to $\alpha$-conversion\footnote{In
  other words, by reading $\jhsubst n {N_0} {x_0} {\alpha_0} N N'$ as
  something like
  $\mathop{\mathrm{subst}^\mathrm{n}_{\alpha_0}}({N_0},\, x_0 \bdot N)
  = N'$, where ${x_0}$ is bound in $N$.}, we can see that our
commutativity theorem is equivalent to an apparently more general one
where the $\eta$-expanded variable is not the same as the
substituted-for variable.  For example, in the case of clause
(\ref{cl:left-n}), we would have that if $\jhsubst n {\expand \alpha
  x} y \alpha N {N'}$, then $\subst x y N = N'$.  We will freely make
use of this fact in what follows when convenient.

\begin{thm}[Expansion]
\label{thm:expansion}
If\/ $\jchecksortclass \Gamma S A$ and $\jsynthtermtype \Gamma R S$,
then $\jchecktermtype \Gamma {\expand A R} S$.
\end{thm}

\begin{proof}[Proof (sketch)]
By induction on $S$.  The $\sspi x {S_1} {A_1} {S_2}$ case relies on
Theorem~\ref{thm:substitution} (Substitution) to show that $\hsubst s
{\expand {A_1} x} x {A_1} {S_2}$ is defined and on
Lemma~\ref{lem:commutativity} (Commutativity) to show that it is equal to
${S_2}$.
\end{proof}

\begin{thm}[Identity]
\label{thm:identity}
If\/ $\jchecksortclass \Gamma S A$,
then $\jchecktermsort {\Gamma, \rdecl x S A} {\expand A x} S$.
\end{thm}
\begin{proof}
Corollary of Theorem~\ref{thm:expansion} (Expansion).
\end{proof}

\section{Subsorting at Higher Sorts}
\label{sect:higher-subsorting}

\noindent Our bidirectional typing discipline limits subsorting checks
to a single rule, the \rulename{switch} rule when we switch modes from
checking to synthesis.  Since we insist on typing only canonical
forms, this rule is limited to checking at atomic sorts $Q$, and
consequently, subsorting need only be defined on atomic sorts.  These
observations naturally lead one to ask, what is the status of
higher-sort subsorting in LFR\@?  How do our intuitions about things
like structural rules, variance, and distributivity---in particular,
the rules shown in Figure~\ref{fig:usual-rules}---fit into the LFR
picture?

\begin{figure}
\begin{mathpar}
\rulesheader{\jchecksubtype \Gamma {S_1} {S_2} A}
\and
\namedaxiom {\jchecksubtype \Gamma {S} {S} A}
            {refl}
\and
\namedrule{\jchecksubtype \Gamma {S_1} {S_2} A
        \\ \jchecksubtype \Gamma {S_2} {S_3} A}
          {\jchecksubtype \Gamma {S_1} {S_3} A}
          {trans}
\and
\namedrule{\jchecksubtype \Gamma {S_2} {S_1} A
        \\ \jchecksubtype {\Gamma, \rdecl x {S_2} A} {T_1} {T_2} {A'}}
          {\jchecksubtype \Gamma {\spi x {S_1} {T_1}} {\spi x {S_2} {T_2}}
                                 {\api x A {A'}}}
          {S-$\Pi$}
\and
\namedaxiom {\jchecksubtype \Gamma S \top A}
          {$\sstop$-R}
\and
\namedrule{\jchecksubtype \Gamma T {S_1} A
        \\ \jchecksubtype \Gamma T {S_2} A}
          {\jchecksubtype \Gamma T {\sinter {S_1} {S_2}} A}
          {$\intersect$-R}
\\
\and
\namedrule{\jchecksubtype \Gamma {S_1} T A}
          {\jchecksubtype \Gamma {\sinter {S_1} {S_2}} T A}
          {$\intersect$-L$_1$}
\and
\namedrule{\jchecksubtype \Gamma {S_2} T A}
          {\jchecksubtype \Gamma {\sinter {S_1} {S_2}} T A}
          {$\intersect$-L$_2$}
\\
\and
\namedaxiom {\jchecksubtype \Gamma {\sstop} {\spi x S \sstop} {\api x A {A'}}}
          {$\sstop/\Pi$-dist}
\and
\namedaxiom {\jchecksubtype \Gamma {\sinter {(\spi x S {T_1})}
                                          {(\spi x S {T_2})}}
                                 {\spi x S {(\sinter {T_1} {T_2})}}
                                 {\api x A {A'}}}
          {$\intersect/\Pi$-dist}
\end{mathpar}
\caption{Derived rules for subsorting at higher sorts.}
\label{fig:usual-rules}
\end{figure}


It turns out that despite not \textit{explicitly} including subsorting at
higher sorts, LFR \textit{implicitly} includes an intrinsic notion of
higher-sort subsorting through the $\eta$-expansion associated with
canonical forms.
The simplest way of formulating this intrinsic notion is as a variant of
the identity principle: $S$ is taken to be a subsort of $T$ if $\jchecktermtype
{\Gamma, \rdecl x S A} {\expand A x} T$.  This notion is equivalent to a
number of other alternate formulations, including a subsumption-based
formulation and a substitution-based formulation.



\begin{thm}[Alternate Formulations of Subsorting]
\label{thm:alternate}
Suppose that for some $\Gamma_0$, $\jchecksortclass {\Gamma_0} {S_1} A$ and $\jchecksortclass {\Gamma_0} {S_2} A$,
and define:
    \begin{enumerate}[\em(1)]
    \item $S_1 \subtype_{\textit{\arabic{enumi}}} S_2 \defeq$
        for all $\Gamma$ and $R$: \ 
        if\/ $\jsynthtermsort \Gamma R {S_1}$,
        then $\jchecktermsort \Gamma {\expand A R} {S_2}$.
        \label{ds:expan}
    \item $S_1 \subtype_{\textit{\arabic{enumi}}} S_2 \defeq$
        for all $\Gamma$: \ 
        $\jchecktermsort {\Gamma, \rdecl x {S_1} A} {\expand A x} {S_2}$.
        \label{ds:ident}
    \item $S_1 \subtype_{\textit{\arabic{enumi}}} S_2 \defeq$
        for all $\Gamma$ and $N$: \ 
        if\/ $\jchecktermsort \Gamma N {S_1}$,
        then $\jchecktermsort \Gamma N {S_2}$.
        \label{ds:subsum}
    \item $S_1 \subtype_{\textit{\arabic{enumi}}} S_2 \defeq$
        for all $\GammaL$, $\GammaR$, $N$, and $S$: \ 
        if\/ $\jchecktermsort {\GammaL, \rdecl x {S_2} A, \GammaR} N S$\\
        \phantom{$S_1 \subtype_{\textit{\arabic{enumi}}} S_2 \defeq$
        for all $\GammaL$, $\GammaR$, $N$, and $S$:\ if\/ }%
        \llap{then}\/ $\,\jchecktermsort {\GammaL, \rdecl x {S_1} A, \GammaR} N S$
        \label{ds:promote}
    \item $S_1 \subtype_{\textit{\arabic{enumi}}} S_2 \defeq$
        for all $\GammaL$, $\GammaR$, $N$, $S$, and $N_1$: \ 
        if\/ $\jchecktermsort {\GammaL, \rdecl x {S_2} A, \GammaR} N S$
         and $\jchecktermsort {\GammaL} {N_1} {S_1}$,
        then $\jchecktermsort {\GammaL, \hsubst \gamma {N_1} x A \GammaR}
                                     {\hsubst n {N_1} x A N}
                                     {\hsubst s {N_1} x A S}$.
        \label{ds:subst}
    \end{enumerate}
Then, $S_1 \subtype_{\textit{\ref{ds:expan}}} S_2
  \iff S_1 \subtype_{\textit{\ref{ds:ident}}} S_2
  \iff \cdots
  \iff S_1 \subtype_{\textit{\ref{ds:subst}}} S_2$.
\end{thm}
\begin{proof} Using the identity and substitution principles along with
Lemma~\ref{lem:commutativity}, the commutativity of substitution with
$\eta$-expansion.
\begin{enumerate}[(1) $\Rightarrow$ (2):]
\item[(\ref{ds:expan}) $\Rightarrow$ (\ref{ds:ident}):]
By rule, $\jsynthtermsort {\Gamma, \rdecl x {S_1} A} x {S_1}$.
By \ref{ds:expan},
        $\jchecktermsort {\Gamma, \rdecl x {S_1} A} {\expand A x} {S_2}$.

\item[(\ref{ds:ident}) $\Rightarrow$ (\ref{ds:subsum}):] 
Suppose $\jchecktermsort \Gamma N {S_1}$.
By \ref{ds:ident},
        $\jchecktermsort {\Gamma, \rdecl x {S_1} A} {\expand A x} {S_2}$.
By Theorem~\ref{thm:substitution} (Substitution),
        $\jchecktermsort \Gamma {\hsubst n N x A {\expand A x}} {S_2}$.
By Lemma~\ref{lem:commutativity} (Commutativity),
        $\jchecktermsort \Gamma N {S_2}$.

\item[(\ref{ds:subsum}) $\Rightarrow$ (\ref{ds:promote}):]
Suppose $\jchecktermsort {\GammaL, \rdecl x {S_2} A, \GammaR} N S$.
By weakening,
        $\jchecktermsort {\GammaL, \rdecl y {S_1} A,
                                   \rdecl x {S_2} A, \linebreak \GammaR}
                         N S$.
By Theorem~\ref{thm:identity} (Identity),
        $\jchecktermsort  {\GammaL, \rdecl y {S_1} A} {\expand A y} {S_1}$.
By \ref{ds:subsum},
        $\jchecktermsort {\GammaL, \rdecl y {S_1} A} {\expand A y} {S_2}$.
By Theorem~\ref{thm:substitution} (Substitution),
        $\jchecktermsort {\GammaL, \rdecl y {S_1} A, \linebreak
                          \hsubst \gamma {\expand A y} x A \GammaR}
                         {\hsubst n      {\expand A y} x A N}
                         {\hsubst s      {\expand A y} x A S}$.
By Lemma~\ref{lem:commutativity} (Commutativity) and $\alpha$-conversion,
        $\jchecktermsort {\GammaL, \rdecl x {S_1} A, \GammaR} N S$.

\item[\ref{ds:promote}) $\Rightarrow$ (\ref{ds:subst}):]
Suppose $\jchecktermsort {\GammaL, \rdecl x {S_2} A, \GammaR} N S$
    and $\jchecktermsort {\GammaL} {N_1} {S_1}$.
  By \ref{ds:promote},
        $\jchecktermsort {\GammaL, \rdecl x {S_1} A, \linebreak \GammaR} N S$.
By Theorem~\ref{thm:substitution} (Substitution),
        $\jchecktermsort {\GammaL, \hsubst \gamma {N_1} x A \GammaR}
                        {\hsubst n {N_1} x A N}
                        {\hsubst s {N_1} x A S}$.

\item[(\ref{ds:subst}) $\Rightarrow$ (\ref{ds:expan}):]
Suppose $\jsynthtermsort \Gamma R {S_1}$.
By Theorem~\ref{thm:expansion} (Expansion),
        $\jchecktermsort \Gamma {\expand A R} {S_1}$.
By Theorem~\ref{thm:identity} (Identity),
        $\jchecktermsort {\Gamma, \rdecl x {S_2} A} {\expand A x} {S_2}$.
  By \ref{ds:subst},
        $\jchecktermsort \Gamma {\hsubst n {\expand A R} x A {\expand A x}}
                                {S_2}$.
By Lemma~\ref{lem:commutativity} (Commutativity),
        $\jchecktermsort \Gamma {\expand A R} {S_2}$.
  \endproof
\end{enumerate}
\end{proof}

\noindent If we take ``subsorting as $\eta$-expansion'' to be our
\textit{model} of subsorting, we can show the ``usual'' presentation
in Figure~\ref{fig:usual-rules} to be both sound and complete with
respect to this model.  In other words, subsorting as $\eta$-expansion
\textit{really is} subsorting (soundness), and it is \textit{no more
  than} subsorting (completeness).  Alternatively, we can say that
completeness demonstrates that there are no subsorting rules missing
from the usual declarative presentation: Figure~\ref{fig:usual-rules}
accounts for everything covered intrinsically by $\eta$-expansion.  By
the end of this section, we will have shown both theorems: if $S
\subtype T$, then $\jchecktermtype {\Gamma, \rdecl x S A} {\expand A
  x} T$, and vice versa.



Soundness is a straightforward inductive argument.

\begin{thm}[Soundness of Declarative Subsorting]
\label{thm:sound-declarative}
If\/ $S \subtype T$, then $\jchecktermtype {\Gamma, \rdecl x S A} {\expand A
x} T$.
\end{thm}
\begin{proof}
By induction on the derivation of $S \subtype T$.  The alternate
formulations given by Theorem~\ref{thm:alternate} are useful in many cases.
\end{proof}


\noindent The proof of completeness is considerably more intricate.
We demonstrate completeness via a detour through an algorithmic
subsorting system very similar to the algorithmic typing system from
Section~\ref{sect:decidability}, with judgments $\jasubtype \Delta S$
and $\japply \Delta {\rdecl x {\Delta_1} {A_1}} {\Delta_2}$.
To show completeness, we show that intrinsic subsorting implies algorithmic
subsorting and that algorithmic subsorting implies declarative subsorting;
the composition of these theorems is our desired completeness result.
\begin{quote}
    If\/ $\jchecktermtype {\Gamma, \rdecl x S A} {\expand A x} T$,
    then $\jasubtype {\ssplit S} T$.
    (\textbf{Theorem~\ref{thm:intr-algo}} below.)
\end{quote}
\begin{quote}
    If\/ $\jasubtype {\ssplit S} T$,
    then $\jsubtype S T$.
    (\textbf{Theorem~\ref{thm:algo-decl}} below.)
\end{quote}
The following schematic representation of soundness and completeness may
help the reader to understand the key theorems.
\begin{diagram}
\fbox{\begin{minipage}{1.35in}
      \centering
      ``declarative'' \\
      $\jsubtype S T$
      \end{minipage}}
    &         & \rTo^{\text{soundness}}
                              &       & \fbox{\begin{minipage}{1.7in}
                                      \centering
                                      ``intrinsic'' \\
                                      $\jchecktermtype {\Gamma, \rdecl x S A}
                                                       {\expand A x} T$
                                      \end{minipage}}
                                                       \\
    & \luTo   & \text{completeness} & \ldTo &   \\
    &         &  \fbox{\begin{minipage}{1.4in}
                       \centering
                       ``algorithmic'' \\
                       $\jasubtype {\ssplit S} T$
                       \end{minipage}}  &       &
\end{diagram}

\noindent As mentioned above, the algorithmic subsorting system system
is characterized by two judgments: $\jasubtype \Delta S$ and $\japply
\Delta {\rdecl x {\Delta_1} {A_1}} {\Delta_2}$ ; rules defining them
are shown in Figure~\ref{fig:algorithmic-rules}.  As in
Section~\ref{sect:decidability}, $\Delta$ represents an
intersection-free list of sorts.  The interpretation of the judgment
$\jasubtype \Delta S$, made precise below, is roughly that the
intersection of all the sorts in $\Delta$ is a subsort of the sort
$S$.

\begin{figure}
\begin{mathpar}
\rulesheader{\jasubtype \Delta S}
\and
\inferaxiom {\jasubtype \Delta \sstop}
\and
\inferrule{\jasubtype \Delta {S_1}
        \\ \jasubtype \Delta {S_2}}
          {\jasubtype \Delta {\sinter {S_1} {S_2}}}
\and
\inferrule{Q' \in \Delta
        \\ \jsubtype {Q'} Q}
          {\jasubtype \Delta Q}
\and
\inferrule{\japply \Delta {\rdecl x {\ssplit {S_1}} {A_1}} {\Delta_2}
        \\ \jasubtype {\Delta_2} {S_2}}
          {\jasubtype \Delta {\sspi x {S_1} {A_1} {S_2}}}
\end{mathpar}
\begin{mathpar}
\rulesheader{\japply \Delta {\rdecl x {\Delta_1} {A_1}} {\Delta_2}}
\and
\inferaxiom {\japply \cdot {\rdecl x {\Delta_1} {A_1}} \cdot}
\and
\inferrule{\japply \Delta {\rdecl x {\Delta_1} {A_1}} {\Delta_2}
        \\ \jasubtype {\Delta_1} {S_1}
        \\ \jhsubst n {\expand {A_1} x} y {A_1} {S_2} {S_2'}}
          {\japply {(\Delta, \sspi y {S_1} {A_1} {S_2})}
                   {\rdecl x {\Delta_1} {A_1}}
                   {\Delta_2, \ssplit {S_2'}}}
\and
\inferrule{\japply \Delta {\rdecl x {\Delta_1} {A_1}} {\Delta_2}
        \\ \jnasubtype {\Delta_1} {S_1}}
          {\japply {(\Delta, \sspi y {S_1} {A_1} {S_2})}
                   {\rdecl x {\Delta_1} {A_1}}
                   {\Delta_2}}
\and
\inferrule{\japply \Delta {\rdecl x {\Delta_1} {A_1}} {\Delta_2}
        \\ \not\exists {S_2'} \bdot
           \jhsubst s {\expand {A_1} x} y {A_1} {S_2} {S_2'}}
          {\japply {(\Delta, \sspi y {S_1} {A_1} {S_2})}
                   {\rdecl x {\Delta_1} {A_1}}
                   {\Delta_2}}
\and
\inferrule{\japply \Delta {\rdecl x {\Delta_1} {A_1}} {\Delta_2}}
          {\japply {(\Delta, Q)} {\rdecl x {\Delta_1} {A_1}} {\Delta_2}}
\end{mathpar}
\caption{Algorithmic subsorting.}
\label{fig:algorithmic-rules}
\end{figure}

The rule for checking whether $\Delta$ is a subsort of a function type
makes use of the application judgment $\japply \Delta {\rdecl x {\Delta_1}
{A_1}} {\Delta_2}$ to extract all of the applicable function codomains from
the list $\Delta$.  As in Section~\ref{sect:decidability}, care is taken to
ensure that this latter judgment is defined even in seemingly
``impossible'' scenarios that well-formedness preconditions would rule out,
like $\Delta$ containing atomic sorts or hereditary substitution being
undefined.

Our first task is to demonstrate that the algorithm has the interpretation
alluded to above.  To that end, we define an operator $\binter{-}$ that
transforms a list $\Delta$ into a sort $S$ by ``folding'' $\intersect$ over
$\Delta$ with unit $\sstop$.
\[\binter{\cdot}=\binter{\Delta, S}=\sinter {\binter\Delta} S\]
Now our goal is to demonstrate that if the algorithm says $\jasubtype
\Delta S$, then declaratively $\jsubtype {\binter\Delta} S$.  First, we
prove some useful properties of the $\binter -$ operator.

\begin{lem}
\label{lem:binter-inter}
$\binter{\Delta_1} \intersect \binter{\Delta_2} \subtype
 \binter{\Delta_1, \Delta_2}$
\end{lem}

\begin{proof}
Straightforward induction on $\Delta_2$.
\end{proof}

\begin{lem}
\label{lem:binter-split}
$S \subtype \binter{\ssplit{S}}$.
\end{lem}

\begin{proof}
Straightforward induction on $S$.
\end{proof}

\begin{lem}
\label{lem:Q}
If\/ $Q' \in \Delta$ and $\jsubtype {Q'} Q$, then $\jsubtype {\binter\Delta} Q$.
\end{lem}

\begin{proof}
Straightforward induction on $\Delta$.
\end{proof}

\begin{thm}[Generalized Algorithmic $\Rightarrow$ Declarative] ~
\label{thm:sound-alg-subtyping}
  \begin{enumerate}[\em(1)]
    \item If\/ $\D :: \jasubtype \Delta T$, then $\jsubtype {\binter \Delta} T$.
          \label{cl:jsub}
    \item If\/ $\D :: \japply \Delta {\rdecl x {\Delta_1} {A_1}} \Delta_2$, then
               $\jsubtype {\binter\Delta} \sspi x {\binter{\Delta_1}} {A_1}
                                                  {\binter{\Delta_2}}$.
          \label{cl:japp}
  \end{enumerate}
\end{thm}

\begin{proof}[Proof (sketch)]
By induction on $\D$, using Lemmas~\ref{lem:binter-inter},
\ref{lem:binter-split}, and \ref{lem:Q}.  The derivable rules from
Figure~\ref{fig:deriv-rules} come in handy in the proof of clause
\ref{cl:japp}.
\end{proof}

\begin{figure}
\begin{mathpar}
\namedrule{\jsubtype {S_1} {T_1} \\ \jsubtype {S_2} {T_2}}
          {\jsubtype {\sinter {S_1} {S_2}} {\sinter {T_1} {T_2}}}
          {S-$\intersect$}
\and
\namedaxiom{\jsubtype {\sinter {S_1} {(\sinter {S_2} {S_3})}}
                      {\sinter {(\sinter {S_1} {S_2})} {S_3}}}
           {$\intersect$-assoc}
\and
\namedrule{\jsubtype S {\spi x {T_1} {T_2}}
        \\ \jsubtype {T_1} {S_1}}
          {\jsubtype {\sinter S {\spi x {S_1} {S_2}}}
                     {\spi x {T_1} {(\sinter {T_2} {S_2})}}}
          {$\intersect/\Pi$-dist$'$}
\end{mathpar}
\caption{Useful rules derivable from those in Figure~\ref{fig:usual-rules}.}
\label{fig:deriv-rules}
\end{figure}

\noindent Theorem~\ref{thm:sound-alg-subtyping} is sufficient to prove that
algorithmic subsorting implies declarative subsorting.

\begin{thm}[Algorithmic $\Rightarrow$ Declarative]
\label{thm:algo-decl}
If\/ $\jasubtype {\ssplit S} T$, then $\jsubtype S T$.
\end{thm}
\proof Suppose $\jasubtype {\ssplit S} T$.  Then, 

    \begin{indented}
    \begin{proofcase}
    $\jsubtype {\binter {\ssplit S}} T$
        \` By Theorem~\ref{thm:sound-alg-subtyping}.\quad\phantom{\qEd} \\
    $\jsubtype S {\binter {\ssplit S}}$
        \` By Lemma~\ref{lem:binter-split}.\quad\phantom{\qEd} \\
    $\jsubtype S T$
        \` By rule \rulename{trans}.\quad\qEd
    \\* \` 
    \end{proofcase}
    \end{indented}


\noindent Now it remains only to show that intrinsic subsorting
implies algorithmic.  To do so, we require some lemmas.  First, we
extend our notion of a sort $S$ refining a type $A$ to an entire list
of sorts $\Delta$ refining a type $A$ in the obvious way.
\begin{mathpar}
\inferaxiom {\jchecksortclass \Gamma \cdot A}
\and
\inferrule {\jchecksortclass \Gamma \Delta A
        \\  \jchecksortclass \Gamma S A}
           {\jchecksortclass \Gamma {(\Delta, S)} A}
\end{mathpar}
This new notion has the following important properties.

\begin{lem}
\label{lem:concat-refine}
If\/ $\jchecksortclass \Gamma {\Delta_1} A$ and $\jchecksortclass \Gamma
{\Delta_2} A$, then $\jchecksortclass \Gamma {\Delta_1, \Delta_2} A$.
\end{lem}
\begin{proof}
Straightforward induction on $\Delta_2$.
\end{proof}
\begin{lem}
\label{lem:split-refine}
If\/ $\jchecksortclass \Gamma S A$, then $\jchecksortclass \Gamma {\ssplit S} A$.
\end{lem}
\begin{proof}
Straightforward induction on $S$.
\end{proof}

\begin{lem}
\label{lem:japp}
If\/ $\D :: \jchecksortclass \Gamma \Delta {\api x {A_1} {A_2}}$ and
     $\E :: \under \Gamma \japply \Delta N {\Delta_2}$ and
     $\jhsubst a N x {A_1} {A_2} {A_2'}$,
then $\jchecksortclass \Gamma {\Delta_2} {A_2'}$.
\end{lem}

\begin{proof}[Proof (sketch)]
By induction on $\E$, using Theorem~\ref{thm:sound-typing} (Soundness of
Algorithmic Typing) to appeal to Theorem~\ref{thm:substitution}
(Substitution), along with Lemmas~\ref{lem:concat-refine} and
\ref{lem:split-refine}.
\end{proof}

We will also require an analogue of subsumption for our algorithmic typing
system, which relies on two lemmas about lists of sorts.

\begin{lem}
\label{lem:D-all}
If\/ $\jchecksortclass \Gamma \Delta A$, then for all $S \in \Delta$,
$\jchecksortclass \Gamma S A$.
\end{lem}
\begin{proof}
Straightforward induction on $\Delta$.
\end{proof}

\begin{lem}
\label{lem:checkall-binter}
If for all\/ $S \in \Delta$, $\jchecktermtype \Gamma N S$, then
$\jchecktermtype \Gamma N {\binter\Delta}$.
\end{lem}
\begin{proof}
Straightforward induction on $\Delta$.
\end{proof}

\begin{thm}[Algorithmic Subsumption]
\label{thm:alg-subs}
If\/ $\jasynth \Gamma R \Delta$ and $\jchecksortclass \Gamma \Delta A$ and
$\jasubtype \Delta S$, then $\jacheck \Gamma {\expand A R} S$.
\end{thm}
\proof
Straightforward deduction, using soundness and completeness of algorithmic
typing.

    \begin{indented}
    \begin{proofcase}
    $\forall S' \in \Delta \bdot \jsynthtermtype \Gamma R {S'}$
        \` By Theorem~\ref{thm:sound-typing} (Soundness of Alg.\ Typing).\quad\phantom{\qEd} \\
    $\forall S' \in \Delta \bdot \jchecksortclass \Gamma {S'} A$
        \` By Lemma~\ref{lem:D-all}.\quad\phantom{\qEd} \\
    $\forall S' \in \Delta \bdot \jchecktermtype \Gamma {\expand A R} {S'}$
        \` By Theorem~\ref{thm:expansion} (Expansion).\quad\phantom{\qEd} \\
    $\jchecktermtype \Gamma {\expand A R} {\binter\Delta}$
        \` By Lemma~\ref{lem:checkall-binter}.\quad\phantom{\qEd} \\
    \\
    $\jasubtype \Delta S$
        \` By assumption.\quad\phantom{\qEd}\\
    $\jsubtype {\binter\Delta} S$
        \` By Theorem~\ref{thm:sound-alg-subtyping}
           (Generalized Alg. $\Rightarrow$ Decl.).\quad\phantom{\qEd}  \\
    \\
    $\jchecktermtype \Gamma {\expand A R} S$
        \` By Theorem~\ref{thm:sound-declarative}
           (Soundness of Decl.\ Subsorting) and \quad\phantom{\qEd}\\
        \` Theorem~\ref{thm:alternate}
           (Alternate Formulations of Subsorting).\quad\phantom{\qEd} \\
    $\jacheck \Gamma {\expand A R} S$
        \` By Theorem~\ref{thm:complete-typing}
           (Completeness of Alg.\ Typing).\quad\qEd
    \\* \` 
    \end{proofcase}
    \end{indented}

\noindent Now we can prove the following main theorem, which
generalizes our desired ``Intrinsic $\Rightarrow$ Algorithmic''
theorem:
\begin{thm}[Generalized Intrinsic $\Rightarrow$ Algorithmic] ~
\label{thm:complete-alg-subtyping}
\begin{enumerate}[\em(1)]
\item If\/ $\jasynth \Gamma R \Delta$
       and $\E :: \jacheck \Gamma {\expand A R} S$
       and $\jchecksortclass \Gamma \Delta A$
       and $\jchecksortclass \Gamma S A$,
      then $\jasubtype \Delta S$.
  \label{cl:comp}
\item If\/ $\jasynth \Gamma x {\Delta_1}$
       and $\E :: \under \Gamma \japply \Delta {\expand {A_1} x} {\Delta_2}$
       and $\jchecksortclass \Gamma {\Delta_1} {A_1}$
       and $\jchecksortclass \Gamma \Delta {\api x {A_1} {A_2}}$,
      then $\japply \Delta {\rdecl x {\Delta_1} {A_1}} {\Delta_2}$.
  \label{cl:compaux}
\end{enumerate}
\end{thm}

\begin{proof}[Proof (sketch)]
By induction on $A$, $S$, and $\E$.

Clause \ref{cl:comp} is most easily proved by case analyzing the sort $S$
and applying inversion to the derivation $\E$.  The case when $S = \sspi x
{S_1} {A_1} {S_2}$ appeals to the induction hypothesis at an unrelated
derivation but at a smaller type, and Lemmas~\ref{lem:concat-refine} and
\ref{lem:split-refine} are used to satisfy the preconditions of the
induction hypotheses.

Clause \ref{cl:compaux} is most easily proved by case analyzing the
derivation $\E$.  In one case, we require the contrapositive of
Theorem~\ref{thm:alg-subs} (Algorithmic Subsumption) to convert a
derivation of $\jnacheck \Gamma {\expand {A_1} x} {S_1}$ into a
derivation of $\jnasubtype {\Delta_1} {S_1}$.

\end{proof}

Theorem~\ref{thm:complete-alg-subtyping} along with
Theorem~\ref{thm:complete-typing}, the Completeness of
Algorithmic Typing, gives us our desired result:
\begin{thm}[Intrinsic $\Rightarrow$ Algorithmic]
\label{thm:intr-algo}
If\/ $\jchecktermtype {\Gamma, \rdecl x S A} {\expand A x} T$,
then $\jasubtype {\ssplit S} T$.
\end{thm}
\proof
Suppose $\jchecktermtype {\Gamma, \rdecl x S A} {\expand A x} T$.  Then,

    \begin{indented}
    \begin{proofcase}
    $\jasynth {\Gamma, \rdecl x S A} x {\ssplit S}$
        \` By rule.\quad\phantom{\qEd} \\
    $\jacheck {\Gamma, \rdecl x S A} {\expand A x} T$
        \` By Theorem~\ref{thm:complete-typing} (Completeness of Alg.
        Typing).\quad\phantom{\qEd}  \\
    $\jasubtype {\ssplit S} T$
        \` By Theorem~\ref{thm:complete-alg-subtyping}.\quad\qEd
    \\* \` 
    \end{proofcase}
    \end{indented}

\noindent Finally, we have completeness as a simple corollary:
\begin{thm}[Completeness of Declarative Subsorting]
If\/ $\jchecktermtype {\Gamma, \rdecl x S A} {\expand A x} T$,
then $\jsubtype S T$.
\end{thm}
\begin{proof}
Corollary of Theorems~\ref{thm:intr-algo} and \ref{thm:algo-decl}.
\end{proof}


\section{Proof Irrelevance}
\label{sect:proof-irrelevance}

\noindent When constructive type theory is used as a foundation for
verified functional programming, we notice that many parts of proofs
are \emph{computationally irrelevant}, that is, their structure does
not affect the returned value we are interested in.  The role of these
proofs is only to guarantee that the returned value satisfies the
desired specification.  For example, from a proof of $\forall x{:}A.\,
\exists y{:}B.\, C(x,y)$ we may choose to extract a function $f : A
\rightarrow B$ such that $C(x,f(x))$ holds for every $x{:}A$, but
ignore the proof that this is the case.  The proof must be present,
but its identity is irrelevant.  Proof-checking in this scenario has
to ascertain that such a proof is indeed not needed to compute the
relevant result.

A similar issue arises when a type theory such as $\lambda^\Pi$ is
used as a logical framework.  For example, assume we would like to have an
adequate representation of prime numbers, that is, to have a
bijection between prime numbers $p$ and closed terms $M :
\var{primenum}$.  It is relatively easy to define a type family
\lst"prime : nat -> type" such that there exists a closed
\lst"M : prime N" if and only if $N$ is prime.  Then
$\var{primenum} = \asigma n {\var{nat}} {\aapp {\var{prime}} n}$ is a
candidate (with members
$\mpair N M$), but it is not actually in bijective correspondence with
prime numbers unless the proof $M$ that a number is prime is always
unique.  Again, we need the existence of $M$, but would like to ignore
its identity.  This can be achieved with \emph{subset types}~%
\cite{Constable86,Salvesen88lics} $\asubset x {\var{nat}}
{\var{prime}(x)}$
whose members are just the prime numbers $p$, but if the restricting
predicate is undecidable then type-checking would be undecidable,
which is not acceptable for a logical framework.

For LF, we further note that $\Sigma$ is not available as a type
constructor, so we instead introduce a new type \lst{primenum}
with exactly one constructor, \lst"primenum/i":
\begin{lstlisting}
    primenum : type.
    primenum/i : {N:nat} prime N -:> primenum.
\end{lstlisting}
Here the second arrow \lst{-:>} represents a function
that ignores the identity of its argument.  The inhabitants
of \lst"primenum", all of the form \lst"primenum/i N [[$M$]]",
are now in bijective correspondence with prime numbers since
\lst"primenum/i N [[$M$]]" $=$ \lst"primenum/i N [[$M'$]]" for all
$M$ and $M'$.

In the extension of LF with proof
irrelevance~\cite{Pfenning01lics,ReedPfenning}, or LFI,
we have a new form of hypothesis $\idecl x A$ ($x$ has type
$A$, but the identity of $x$ should be irrelevant).  In the
non-dependent case (the only one important for the purposes
of this paper), such an assumption is introduced by a
$\lambda$-abstraction:
\begin{mathpar}
\inferrule{\jchecktermtype {\Gamma, \idecl x A} M B}
          {\jchecktermtype \Gamma {\mlam x M} {\aiarrow A B}}~.
\end{mathpar}
We can use such variables only in places where their identity
doesn't matter, e.g., in the second argument to the
constructor \lst"primenum/i" in the prime number example.  More
generally, we can only use it in arguments to constructor
functions that do not care about the identity of their argument:
\begin{mathpar}
\inferrule{\jsynthtermtype \Gamma R {\aiarrow A B}
        \\ \jchecktermtype {\promote\Gamma} N A}
          {\jsynthtermtype \Gamma {\miapp R N} B}~.
\end{mathpar}
Here, $\promote\Gamma$ is the \emph{promotion} operator which converts
any assumption $\idecl x A$ to $\decl x A$, thereby making
$x$ usable in $N$.  Note that there is no direct way to
use an assumption $\idecl x A$.

The underlying definitional equality ``$=$'' (usually just
$\alpha$-conversion on canonical forms)
is extended so that $\miapp R N = \miapp {R'} {N'}$ if $R = R'$,
no matter what $N$ and $N'$ are.

The substitution principle (shown here only in its simplest,
non-dependent form) captures the proper typing
as well as the irrelevance of assumptions $\idecl x A$:
\begin{principle}[Irrelevant Substitution]
If $\jchecktermtype {\Gamma, \idecl x A} N B$
and $\jchecktermtype {\promote\Gamma} M A$
then $\jchecktermtype \Gamma {\subst M x N} B$
and $\subst M x N = N$ (under definitional equality).
\end{principle}

One typical use of proof irrelevance in type theory is to render the
typechecking of subset types~%
\cite{Constable86,Salvesen88lics} decidable.  A subset type $\asubset x A
{B(x)}$ represents the set of terms of type $A$ which also satisfy $B$;
typechecking is undecidable because to determine if a term $M$ has this
type, you must search for a proof of $B(M)$.  One might attempt to recover
decidability by using a dependent sum $\asigma x A {B(x)}$, representing
the set of terms $M$ of type $A$ paired with proofs of $B(M)$; typechecking
is decidable, since a proof of $B(M)$ is provided, but equality of terms is
overly fine-grained: if there are two proofs of $B(M)$, the two pairs will
be considered unequal.  Using proof irrelevance, one can find a middle
ground with the type $\asigma x A {[B(x)]}$, where $[-]$ represents the
proof irrelevance modality.  Type checking is decidable for such terms,
since a proof of the property $B$ is always given, but the identity of that
proof is ignored, so all pairs with the same first component will be
considered equal.

Our situation with the subset interpretation is similar: we would like to
represent proofs of sort-checking judgments without depending on the
identities of those proofs.  By carefully using proof irrelevance to hide the
identities of sort-checking proofs, we are able to make a translation that
is sound and complete, preserving the adequacy of representations.

\section{Interpretation}
\label{sect:interpretation}


\subsection{Overview}

We interpret LFR into LFI by representing sorts as predicates and
derivations of sorting as proofs of those predicates.  In this section, we
endeavor to explain our general translation by way of examples of it in
action.
The translation is
derivation-directed and compositional: for each judgment $\Gamma \vdash
\mathcal{J}$, there is a corresponding judgment $\Gamma \vdash \mathcal{J}
\leadsto X$ whose rules mimic the rules of $\Gamma \vdash \mathcal{J}$.
The syntactic class of $X$ and its precise interpretation vary from
judgment to judgment.  For reference, the various forms are listed in
Table~\ref{tab:judgments}, but we will explain them in turn as they arise
in our examples.

\setlength{\oldarrayrulewidth}{\arrayrulewidth}
\setlength{\arrayrulewidth}{0.2pt}
\setlength{\doublerulesep}{\arrayrulewidth}

\begin{table}[tb]
\begin{center}
\begin{tabular*}{0.92\textwidth}%
                {l@{\extracolsep{\fill}}cl}
\hline\hline\hline\hline
{\bf Judgment:}                             && {\bf Result:} \bigstrut \\
\hline\hline\hline\hline
$\jclasstform \Gamma L K {\metaapp{\what L\f} -} $
    && Type of proofs of the formation family \bigstrut[t] \\
$\kindtpred K {\metaapptwo {\what K\p} - -}$
    && Kind of the predicate family \\
$\kindtsub K {\metaappfive {\what K\s} - - - - -}$         
    && Type of coercions between families of kind $K$ \bigstrut[b] \\
\hline
$\jchecksortclasst \Gamma S A {\metaapp {\what S} -}$
    && Metafunction representing predicate \bigstrut[t] \\
$\jsynthsortclasst \Gamma Q P L {\what Q}$  && Proof that $Q$ is well-formed
                                               \bigstrut[b] \\
\hline
$\jchecktermsortt \Gamma N S {\what N}$     && Proof that $N$ has sort $S$
                                               \bigstrut[t] \\
$\jsynthtermsortt \Gamma R S {\what R}$     && Proof that $R$ has sort $S$
                                               \bigstrut[b] \\
\hline
$\jchecksubtypet \Gamma {Q_1} {Q_2} {\metaapptwo F - -}$     
    && Metacoercion from proofs of $Q_1$ to proofs of $Q_2$
                                               \bigstrut[t] \\
$\jsynthsubtypet {Q_1} {Q_2} {\whatQQ}$     && Coercion from proofs of $Q_1$ to proofs of $Q_2$
                                               \bigstrut[b] \\
\hline
$\jctxt \Gamma {\what \Gamma}$              && Translated context
                                               \bigstrut[t] \\
$\jsigt \Sigma {\what \Sigma}$              && Translated signature
                                               \bigstrut[b] \\
\hline\hline\hline\hline
\end{tabular*}
\end{center}
\caption{Judgments of the translation.}
\label{tab:judgments}
\end{table}

\setlength{\arrayrulewidth}{\oldarrayrulewidth}

Recall our simplest example of refinement types: the natural numbers,
where the even and odd numbers are isolated as refinements.
\begin{lstlisting}
    nat : type.
    z : nat.
    s : nat -> nat. `\medskip`
    even << nat.
    odd << nat.
    z :: even.
    s :: even -> odd  ^  odd -> even.
\end{lstlisting}
As described in the introduction, our translation represents \lst{even} and
\lst{odd} as \textit{predicates} on natural numbers, and the refinement
declarations for \lst{z} and \lst{s} become declarations for constants for
constructing proofs of those predicates.
\begin{lstlisting}
    even : nat -> type.
    odd : nat -> type.
    $\zhat\phantom{_1}$ : even z.
    $\shat_1$ : {x:nat} even x -> odd (s x).
    $\shat_2$ : {x:nat} odd x -> even (s x).
\end{lstlisting}
Starting simple, the proof constructor declaration for $\zhat$ can be read
as an assertion that the constant \lst{z} satisfies a certain predicate,
namely that of being even.

In fact, every sort $S$ will have a representation as a predicate, not just
the base sorts like \lst{even} and \lst{odd}.  Generally, a predicate is
just a \textit{type} with a hole for a term; conventionally, we write the
predicate representation of $S$ as a meta-level function $\metaapp {\what
S} -$, and we say that a term $N$ satisfies such a predicate if the type
$\metaapp {\what S} N$ is inhabited.  Predicates will be the output of the
sort translation judgment, $\jchecksortclasst \Gamma S A {\what S}$, which
mirrors the sort formation judgment, adding a translation as an output.

For example, the predicate corresponding to the sort \lst{even -> odd} is
the meta-function (\lst'{x:nat} even x -> odd ($\metalp-\metarp$ x)'), and
we see this predicate applied to the successor constant \lst{s} in the type
of the proof constructor $\shat_1$.  Thus the proof constructor declaration
for $\shat_1$ can \textit{also} be read as an assertion: the constant
\lst's' satisfies the predicate that, when applied to an even natural
number, it yields an odd one.

Our analysis suggests a general strategy for translating a refinement type
declaration: translate its sort into a predicate, and yield a declaration
of a proof constructor asserting that the predicate holds of the original
constant.
\begin{mathpar}
\inferrule{\jsigt \Sigma {\what \Sigma}
        \\ \decl c A \in \Sigma
        \\ \jchecksortclasst [\Sigma] \cdot S A {\what S}}
          {\jsigt {\Sigma,\ \mrefdecl c S}
                  {\what \Sigma,\ \decl {\what c}
                                       {\metaapp {\what S} {\expand A c}}}}
\end{mathpar}
As a reflection of the fact that in general these predicates may be applied
to arbitrary terms, not just atomic ones, we fully $\eta$-expand the
constant before applying the predicate.

How do arrow sorts like \lst{even -> odd} translate in general?  Recall
that $\sarrow S T$ is just shorthand for the dependent function sort $\spi
x S T$ when $x$ does not occur in $T$.  The general rule for translating
dependent function sorts is:
\begin{mathpar}
\namedrule{\jchecksortclasst \Gamma {S} {A} {\what S}
        \\ \jchecksortclasst {\Gamma, \rdecl x {S} {A}} {T} {B} {\what {T}}}
          {\jchecksortclasst \Gamma {\sspi x {S} {A} {T}} {\api x {A} {B}}
                             {\metalam N
                              {\api x A
                               {\api {\what x} {\metaapp {\what S} {\expand A x}} 
                                {\metaapp {\what {T}} {\revapp N x}}}}}}
          {$\Pi$-F}
\end{mathpar}
There are two points of note in this rule.  First, writing predicates as
types with holes becomes cumbersome, so we instead write metafunctions
explicitly using meta-level abstraction, written as a bold
$\boldsymbol\lambda$; we continue to write meta-level application using
bold $\metalp$parens$\metarp$.  Second, since as we noted above, the term
argument of a predicate is in general a canonical term, and canonical terms
may not appear in application position, we appeal to an auxiliary judgment
that applies a canonical term to an atomic one, $\revapp N R$.  It is
defined by the single clause,
\begin{mathpar}
    \revapp {(\mlam x N)} R = \subst R x N ,
\end{mathpar}
where the right-hand side is an ordinary non-hereditary substitution.
Now we can read the translation output as the predicate of a term $N$
which holds if there is a function from objects $x:A$ satisfying
predicate $\what S$ to proofs that $N$ applied to $x$ satisfies predicate
$\what {T}$.

But what about the fact that \lst's' only had one declaration in the
original signature, but there are two proof constructor declarations
asserting predicates that hold of it?
For compositionality's sake, we would like to translate the single
refinement declaration for \lst{s} into a single proof constructor
declaration, but one that can effectively serve the roles of both
$\shat_1$ and $\shat_2$.  To this end, we use a product type.
\begin{lstlisting}
    $\shat$ : ({x:nat} even x -> odd (s x))
     * ({x:nat} odd x -> even (s x)).
\end{lstlisting}
Now $\pi_i\, \shat$ may be used anywhere $\shat_i$ was used before.
Generally, an intersection sort will translate to a conjunction of
predicates, represented as a type-theoretic product.  Similarly, the
nullary intersection $\sstop$ will translate to a unit type.%
\footnote{Strictly speaking, this means our translation targets an
extension of LFI with product and unit types.  Such an extension is
orthogonal to the addition of proof irrelevance, and has been studied by
many people over the years, including Sch\"{u}rmann \cite{Schurmann03} and
Sarkar \cite{Sarkar09}.  Alternatively, products may be eliminated after
translation by a simple currying transformation, but that is beyond the
scope of this article.}
\begin{mathpar}
\namedrule{\jchecksortclasst \Gamma {S_1} A {\what {S_1}}
        \\ \jchecksortclasst \Gamma {S_2} A {\what {S_2}}}
          {\jchecksortclasst \Gamma {\sinter {S_1} {S_2}} A
                                    {\metalam N {\aprod
                                        {\metaapp {\what {S_1}} N}
                                        {\metaapp {\what {S_2}} N}}}}
          {$\intersect$-F}
\and
\namedaxiom {\jchecksortclasst \Gamma \sstop A {\metalam N \aunit}} {$\top$-F}
\end{mathpar}


\noindent What kinds of proofs inhabit these predicates?  Such proofs
are the output of the term translation judgment $\jchecktermsortt
\Gamma N S {\what N}$, which mirrors the sort checking judgment,
adding a translation as an output.  Generally, a derivation that a
term $N$ has sort $S$ will translate to a proof $\what N$ that the
predicate $\what S$ holds of $N$ (where $\what S$ is as usual the
interpretation of $S$ as a predicate), or symbolically, if $S \refines
A \leadsto \what S$ and $N \checks S \leadsto \what N$, then $\what N
\checks \metaapp {\what S} N$---ignoring for a moment the question of
what happens to the contexts.  This expectation begins to hint at the
soundness theorem we will demonstrate below, but for now we will use
it just to guide our intuitions.

For example, since an intersection sort is represented by a product
of predicates,  we should expect that a term judged to have an
intersection sort should translate to a proof of a product, or a pair.
Similarly, since the sort $\sstop$ translates to a trivially true unit
predicate, a term judged to have sort $\sstop$ should translate to a
trivial unit element.
\begin{mathpar}
\namedrule{\jchecktermsortt \Gamma N {S_1} {\what {N_1}}
        \\ \jchecktermsortt \Gamma N {S_2} {\what {N_2}}}
          {\jchecktermsortt \Gamma N {\sinter {S_1} {S_2}}
                                     {\mpair {\what {N_1}} {\what {N_2}}}}
          {$\intersect$-I}
\and
\namedaxiom {\jchecktermsortt \Gamma N \sstop \munit}
          {$\top$-I}
\end{mathpar}
Intuitively, knowing that a term has an intersection sort $\sinter {S_1}
{S_2}$ gives us two pieces of information about it, while knowing that a
term has sort $\sstop$ tells us nothing new.
\newcommand{\lambdatmeet}{\Lambda^{\textsf{t}}_{\intersect}}
This aspect of our translation is similar in spirit to Liquori and Ronchi
Della Rocca's $\lambdatmeet$~\cite{Liquori07intersection}, a Church-style
type system for intersections in which derivations are explicitly
represented as proofs and intersections as products, though in their
setting the proofs are viewed as part of a program rather than the output
of a translation.

We can similarly intuit the appropriate proof for an implication predicate
by examining the rule for translating $\spi x S T$ above.
We start from the sort-checking rule \rulename{$\Pi$-I}, which shows that
a term $\mlam x N$ has sort $\spi x S T$.
To prove that the corresponding $\Pi$ predicate holds of $\mlam x N$, we
will have to produce a function taking an object $x$ of type $A$ and a
proof that $x$ satisfies $\what S$ and yielding a proof that $\revapp
{(\mlam x {N})} x = \subst x x {N} = N$ satisfies $\what T$.  This is
easily done: the translation of the body $N$ is precisely the proof we
require about $N$, and we wrap this in two $\lambda$-abstractions to get a
proof of the $\Pi$ predicate.
\begin{mathpar}
\namedrule{\jchecktermsortt {\Gamma, \rdecl x {S} {A}} N {T} {\what N}}
          {\jchecktermsortt \Gamma {\mlam x N} {\sspi x {S} {A} {T}}
                            {\mlam x \mlam {\what x} {\what N}}}
          {$\Pi$-I}
\end{mathpar}

\noindent Careful examination of the \rulename{$\Pi$-I} rule reveals a
subtlety: it is clear from our understanding of the sort-checking part
of the rule that the free variables of $N$ and $T$ may include $x$,
but we seem to have indicated by our $\lambda$-abstraction that the
proof $\what N$ may depend not only on the variable $x$, but also on a
variable $\what x$.  Where did this second variable come from?

The answer---as hinted above---is that we have not yet specified with
respect to what context the translation of a term is to be interpreted.
This context should in fact be the translation of the context $\Gamma$
associated with the original term $N$, and by convention we write it as
$\what \Gamma$.  The judgment translating contexts is an annotated version
of the context-formation judgment, written $\jctxt \Gamma {\what \Gamma}$.
\begin{mathpar}
\inferaxiom{\jctxt \cdot \cdot}
\and
\inferrule{\jctxt \Gamma {\what \Gamma}
        \\ \jchecksortclasst \Gamma S A {\what S}}
          {\jctxt {\Gamma, \rdecl x S A}
                  {\what \Gamma, \decl x A, 
                                 \decl {\what x}
                                       {\metaapp {\what S} {\expand A x}}}}
\end{mathpar}
The second rule is quite similar to the translation rule we
have seen for signature declarations $\decl c A$: each declaration $\rdecl
x S A$ splits into a typing declaration $\decl x A$ and a proof declaration
$\decl {\what x} {\metaapp {\what S} {\expand A x}}$.  Now it is easily
seen why the proof $\what N$ in the translation rule \rulename{$\Pi$-I}
may depend on $\what x$: our soundness criterion will tell us that
$\jchecktermtype {\what\Gamma, \decl x A, \decl {\what x} {\metaapp {\what
S} {\expand A x}}} {\what N} {\metaapp {\what T} N}$.

There is just one sort checking rule remaining: the \rulename{switch} rule
for checking an atomic term at a base sort.  This rule appeals to
subsorting, so we postpone discussion of it until we discuss the
translation of subsorting judgments in Section~\ref{sect:trans-subsorting}.
For now, the reader may think of the rule as simply returning the result of
the sort synthesis translation judgment, $\jsynthtermsortt \Gamma R S
{\what R}$.  At the base cases, this judgment returns the hatted proof
constants ${\what c}$ and variables ${\what x}$ we have seen in the
translations of signature declarations and contexts.  The other rules
correspond to elimination forms, and they follow straightforwardly by the
same intuitions we used to derive the introduction rules in the sort
checking translation.  All the rules for this judgment are shown in
Figure~\ref{fig:ssynth}.

\begin{figure}
\hrule
\begin{mathpar}
\rulesheader{\jsynthtermsortt [\Sigma] \Gamma {R^+} {S^-} {{\what R}^-}}
\and
\namedrule{\mrefdecl c S \in \Sigma
          }
          {\jsynthtermsortt \Gamma c S {\what c}}
          {const}
\and
\namedrule{\rdecl x S A \in \Gamma}
          {\jsynthtermsortt \Gamma x S {\what x}}
          {var}
\and
\namedrule{\jsynthtermsortt \Gamma {R_1} {\sspi x {S_2} {A_2} {S}} {\what {R_1}}
        \\ \jchecktermsortt \Gamma {N_2} {S_2} {\what {N_2}}
        \\ \jhsubst s {N_2} x {A_2} {S} {S'}
        }
          {\jsynthtermsortt \Gamma {\mapp {R_1} {N_2}} {S'}
                            {\mapp {\mapp {\what {R_1}} {N_2}} {\what {N_2}}}} 
          {$\Pi$-E}
\\
\and
\namedrule{\jsynthtermsortt \Gamma R {\sinter {S_1} {S_2}} {\what R}}
          {\jsynthtermsortt \Gamma R {S_1} {\mfst {\what R}}}
          {$\intersect$-E$_1$}
\and
\namedrule{\jsynthtermsortt \Gamma R {\sinter {S_1} {S_2}} {\what R}}
          {\jsynthtermsortt \Gamma R {S_2} {\msnd {\what R}}}
          {$\intersect$-E$_2$}
\end{mathpar}
\hrule
\caption{Translation rules for atomic term sort synthesis}
\label{fig:ssynth}
\end{figure}

There is also just one sort formation rule remaining: the rule for
translating base sorts $Q$.  Although this translation seems
straightforward in the case of simple sorts like \lst'even' and \lst'odd',
it is rather subtle when it comes to dependent sort families due to a
problem of coherence.  To explain, we return to another early example, the
doubling relation on natural numbers.

\subsection{Dependent Base Sorts}

Recall the double relation defined as a type family in LF:
\begin{lstlisting}
    double : nat -> nat -> type.
    dbl/z : double z z.
    dbl/s : {N:nat} {N2:nat} double N N2 -> double (s N) (s (s N2)).
\end{lstlisting}
As we saw earlier, we can use LFR \textit{refinement kinds}, or
\textit{classes}, to express and enforce the property that the second
subject of any doubling relation is always even, no matter what properties
hold of the first subject.  To do so we define a sort \doublestar{} which
is isomorphic to \lst{double}, but has a more precise class.%
\footnote{Earlier, we used the name \lst{double} for both the type family
and the sort family refining it, but in what follows it will be important
to distinguish the two.}
\begin{lstlisting}
    $\doublestar$ << double :: # -> even -> sort.
    dbl/z :: $\doublestar$ z z.
    dbl/s :: {N::#} {N2::even} $\doublestar$ N N2 -> $\doublestar$ (s N) (s (s N2)).
\end{lstlisting}
Successfully sort-checking the declarations for \lst{dbl/z} and \lst{dbl/s}
demonstrates that whenever \lst{$\doublestar$ M N} is inhabited, the second
argument, \lst{N}, is even.

There is a crucial difference between refinements like \lst{even} or
\lst{odd} and refinements like \doublestar: while \lst{even} and \lst{odd}
denote particular subsets of the natural numbers, the inhabitants of
the refinement \lst{$\doublestar{}$ M N} are identical to
those 
of the ordinary type \lst{double M N}.  What is important is not whether a
particular instance \lst{$\doublestar$ M N} is inhabited, but rather
whether it is \textit{well-formed at all}.

For this reason, we separate the \textit{formation} of a dependent
refinement type family from its \textit{inhabitation}.  Simple sorts
like \lst'even' and \lst'odd' are always well-formed, but we would like a
way to explicitly represent the formation of an indexed sort like
\lst'$\doublestar$ M N'.  Therefore, we translate \lst'$\doublestar$' into
two parts: a \textit{formation family}, written \doublestarhat{}, and a
\textit{predicate family}, written using the original name of the sort,
\doublestar.

There are two declarations involving the formation family.  First, the
declaration of the formation family itself:
\begin{lstlisting}
    $\doublestarhat$ : nat -> nat -> type.
\end{lstlisting}
The formation family has the same kind as the original refined type.
Intuitively, the formation family \lst{$\doublestarhat$ M N} should be
inhabited whenever the sort \lst{$\doublestar$ M N} would have been a
well-formed sort pre-translation.  For example, \lst{$\doublestarhat$ z z}
will be inhabited, since \lst{$\doublestar$ z z} was a well-formed sort.

Next, we have a constructor for the formation family:
\begin{lstlisting}
    $\doublestarhat$/i : {x:nat} {y:nat} even y -> $\doublestarhat$ x y.
\end{lstlisting}
The constructor takes all the arguments to \doublestar{} along with
evidence that they have the appropriate sorts and yields a member of the
formation family, i.e. a proof that \doublestar{} applied to those arguments
was well-formed pre-translation.  For example,
\lst{$\doublestarhat/i$ z z $\zhat$} is a proof that \lst{$\doublestar$ z z}
was well-formed, since it contains the necessary evidence: a proof that
the second argument \lst'z' is \lst'even'.

Finally, we have a declaration for the predicate family itself:
\begin{lstlisting}
    $\doublestar$ : {x:nat} {y:nat} $\doublestarhat$ x y -:> double x y -> type.
\end{lstlisting}
For any $M$ and $N$, the predicate family will be inhabited by proofs that
derivations of \lst{double M N} have the refinement \lst{$\doublestar$ M N},
provided that \lst{$\doublestar$ M N} is well-formed in the first place.
In our doubling example, all derivations of \lst{double M N} satisfy the
refinement \lst{$\doublestar$ M N}, so the predicate family will have one
inhabitant for each of them.  As before, these inhabitants come from the
translation of the refinement declarations for \lst{dbl/z} and \lst{dbl/s}.
Writing arguments in irrelevant position in \lst{[[ square brackets ]]}, we
get:
\begin{lstlisting}
    $\dblzhat$ : $\doublestar$ z z
                [[ $\doublestarhat$/i z z $\zhat$ ]]
                dbl/z.
    $\dblshat$ : {N:nat} {N2:nat} {$\Ntwohat$:even N2} {D:double N N2}
             $\doublestar$ N N2 [[ $\doublestarhat$/i N N2 $\Ntwohat$ ]] D
         -> $\doublestar$ (s N) (s (s N2))
                [[ $\doublestarhat$/i (s N) (s (s N2)) ($\shat_2$ (s N2) ($\shat_1$ N2 $\Ntwohat$)) ]]
                (dbl/s N N2 D).
\end{lstlisting}
As is evident even from this short and abbreviated example, the
interpretation leads to a significant blowup in the size and complexity of
a signature, underscoring the importance of a primitive understanding of
refinement types.

Note that in the declaration of the predicate family \doublestar{}, the
proof of well-formedness is made irrelevant using a proof-irrelevant
function space $\aiarrow A B$, representing functions from $A$ to $B$ that
are insensitive to the identity of their argument.  Using irrelevance
ensures that a given sort has a unique translation, up to equivalence.  We
elaborate on this below.

Generalizing from the above example, a sort declaration translates into
three declarations: one for the \textit{formation family}, one for the
\textit{proof constructor} for the formation family, and one for the
\textit{predicate family}.
\begin{mathpar}
\inferrule{\jsigt \Sigma {\what \Sigma}
        \\ \decl a K \in \Sigma
        \\ \jclasstform [\Sigma] \cdot L K {\what L\f}
        \\ \kindtpred K {\what K\p}}
          {\jsigt {\Sigma, \arefdecl s a L}
                  {\what \Sigma,
                   \ \decl {\what s} K,
                   \ \decl {\what s / i}
                        {\metaapp {\what L\f} {\what s}},
                   \ \decl s {\metaapptwo {\what K\p} {\what s} a}}}
\end{mathpar}
The class formation judgment $\jclasstform \Gamma L K {\what L\f}$
yields a metafunction describing the type of proofs of formation family,
while an auxiliary kind translation judgment $\kindtpred K {\what K\p}$
yields a metafunction describing the kind of the predicate family. 
As in the example, the kind of the formation family is the same as the kind
of the refined type, $K$.

The metafunction $\what L\f$ takes as input the formation family so far,
initially just $\what s$.  The translation of $\Pi$ classes adds an
argument, and the base case returns the formation family so constructed.
\begin{mathpar}
\inferrule{\jchecksortclasst \Gamma S A {\what S}
        \\ \jclasstform {\Gamma, \rdecl x S A} L K {\what L}}
          {\jclasstform \Gamma {\ccpi x S A L} {\kpi x A K}
                    {\metalam {Q\f} {\kpi x A
                                {\kpi {\what x} {\metaapp {\what S} {\expand A x}} 
                                 {\metaapp {\what L} {\aapp {Q\f} {\expand A x}}}}}}}
\and
\inferaxiom{\jclasstform \Gamma \csort \ktype {\metalam {Q\f} {Q\f}}}
\end{mathpar}
Employing a similar trick as we did with intersection sorts, we will
translate intersection and $\ctop$ classes to unit and product types.
\begin{mathpar}
\inferrule{\jclasstform \Gamma {L_1} K {\what {L_1}}
        \\ \jclasstform \Gamma {L_2} K {\what {L_2}}}
          {\jclasstform \Gamma {\cinter {L_1} {L_2}} K
                    {\metalam {Q\f} {\kprod {\metaapp {\what {L_1}} {Q\f}}
                                            {\metaapp {\what {L_2}} {Q\f}}}}}
\and
\inferaxiom {\jclasstform \Gamma \ctop K {\metalam {Q\f} \kunit}}
\end{mathpar}
Intersection classes give multiple ways for a sort to be well-formed, and a
product of formation families gives multiple ways to project out a proof of
well-formedness.

The metafunction $\what K\p$ takes two arguments: one for the formation
family so far (initially $\what s$) and one for the refined type so far
(initially $a$).  The rule for $\Pi$ kinds just adds an argument to each:
\begin{mathpar}
 \inferrule{\kindtpred K {\what K}}
           {\kindtpred {\kpi x A K}
                {\metalamtwo {Q\f} P
                    {\kpi x A {\metaapptwo {\what K}
                                           {\aapp {Q\f} {\expand A x}}
                                           {\aapp P {\expand A x}}}}}} 
\end{mathpar}
while the translation is really characterized by its behavior on the
base kind, $\ktype$:
\begin{mathpar}
\inferaxiom{\kindtpred \ktype
                {\metalamtwo {Q\longform} P
                    {\aiarrow {Q\longform} {\aarrow P \ktype}}}}
\end{mathpar}
The kind of the predicate family for a base sort $Q$ refining $P$ is
essentially a one-place judgment on terms of type $P$, along with an
irrelevant argument belonging to the formation family of $Q$.

Finally, we are able to make sense of the rule for translating base sorts:
\begin{mathpar}
\namedrule{\jsynthsortclasst \Gamma Q {P'} L {\what Q}
        \\ P' = P
        \\ L = \csort
        }
          {\jchecksortclasst \Gamma Q P
            {\metalam N {\aapp {\aiapp Q {\what Q}} N}}}
        {$Q$-F}
\end{mathpar}
The class synthesis translation judgment $\jsynthsortclasst \Gamma Q P L
{\what Q}$ (similar to the sort synthesis judgment; see
Figure~\ref{fig:lsynth}) yields a proof of $Q$'s formation family; thus
the predicate for a base sort $Q$, given an argument $N$, is simply the
predicate family $Q$ applied to an irrelevant proof $\what Q$ that $Q$ is
well-formed and the argument itself, $N$.

\begin{figure}
\hrule
\begin{mathpar}
\rulesheader{\jsynthsortclasst [\Sigma] \Gamma {Q^+} {P^-} {L^-} {\what Q^-}}
\and
\inferrule{\arefdecl s a L \in \Sigma
        }
          {\jsynthsortclasst \Gamma s a L {\what s / i}}
\and
\inferrule{\jsynthsortclasst \Gamma Q P {\ccpi x S A L} {\what Q}
        \\ \jchecktermsortt \Gamma N S {\what N}
        \\ \jhsubst l N x A L {L'}
        }
          {\jsynthsortclasst \Gamma {\sapp Q N} {\aapp P N} {L'}
                             {\aapp {\aapp {\what Q} N} {\what N}}} 
\and
\inferrule{\jsynthsortclasst \Gamma Q P {\cinter {L_1} {L_2}} {\what Q}}
          {\jsynthsortclasst \Gamma Q P {L_1} {\afst {\what Q}}}
\and
\inferrule{\jsynthsortclasst \Gamma Q P {\cinter {L_1} {L_2}} {\what Q}}
          {\jsynthsortclasst \Gamma Q P {L_2} {\asnd {\what Q}}}
\end{mathpar}
\hrule
\caption{Translation rules for base sort class synthesis}
\label{fig:lsynth}
\end{figure}


What if we hadn't made the proofs of
formation irrelevant?  Then if there were more than one proof that
$Q$ were well-formed, a soundness problem would arise.  To see how, let us
return to the doubling example.  Imagine extending our encoding of natural
numbers with a sort distinguishing zero as a refinement.
\begin{lstlisting}
    zero << nat.
    z :: even  ^  zero.
\end{lstlisting}
As with \lst{even} and \lst{odd}, the sort \lst{zero} turns into a predicate.
Now that \lst{z} has two sorts, it translates to two proof constructors.%
\footnote{For the sake of simplicity, we will continue our example with the
slightly unfaithful assumptions we've been making all along.  Strictly
speaking, \lst{zero} should also have a formation family with a single
trivial member, and the two declarations $\zhat_1$ and $\zhat_2$ should be
one declaration of product type.  The point we wish to make will be the
same nonetheless.}
\begin{lstlisting}
    zero : nat -> type.
    $\zhat_1$ : even z.
    $\zhat_2$ : zero z.
\end{lstlisting}

\noindent Next, we can observe that zero always doubles to itself and
augment the declaration of \doublestar{} using an intersection class:
\begin{lstlisting}
    $\doublestar$ << double ::  # -> even -> sort
                           ^  zero -> zero -> sort.
\end{lstlisting}
After translation, since there are potentially two ways for
\lst{$\doublestar$ x y} to be well-formed, there are two introduction
constants for the formation family.
\begin{lstlisting}
    $\doublestarhat/i_1$ : {x:nat} {y:nat} even y -> $\doublestarhat$ x y.
    $\doublestarhat/i_2$ : {x:nat} zero x -> {y:nat} zero y -> $\doublestarhat$ x y.
\end{lstlisting}
The declarations for $\doublestarhat$ and $\doublestar$ remain the same.

Now recall the refinement declaration for doubling zero,
\begin{lstlisting}
    dbl/z :: $\doublestar$ z z ,
\end{lstlisting}
and observe that it is valid for two reasons, since \lst{$\doublestar$ z z}
is well-formed for two reasons.  Consequently, after translation, there
will be two proofs inhabiting the formation family \lst{$\doublestarhat$ z z},
but only one of them will be used in the translation of the \lst{dbl/z}
declaration.  Supposing it is the first one, we'll have
\begin{lstlisting}
    $\dblzhat$ : $\doublestar$ z z [[ $\doublestarhat/i_1$ z z $\zhat_1$ ]] dbl/z ,
\end{lstlisting}
but our soundness criterion will still require that the constant $\dblzhat$
check at the type
\lst{$\doublestar$ z z [[ $\doublestarhat/i_2$ z $\zhat_2$ z $\zhat_2$ ]] dbl/z},
the other possibility.
The apparent mismatch is resolved by the fact that the formation proofs are
irrelevant, and so the two types are considered equal.  Without proof
irrelevance, the two types would be distinct and we would have a
counterexample to the soundness theorem (Theorem~\ref{thm:sound-interp}) we
prove below.



\subsection{Subsorting}
\label{sect:trans-subsorting}

We now return to the question of how the translation handles subsorting.
Recall that an LFR signature can include subsorting declarations between
sort family constants, $\subdecl {s_1} {s_2}$.  For instance, continuing
with our running example of the natural numbers, we might note that any
\lst'nat' that is \lst'zero' is \lst'even' by declaring:
\begin{lstlisting}
    zero <: even.
\end{lstlisting}
Such a declaration may seem redundant, since the only thing declared to
have sort \lst'zero' has \textit{already} been declared to have sort
\lst'even', but it may be necessary given the inherently open-ended nature
of an LF signature.  We may find ourselves later in a situation where we
have a new hypothesis \lst{x : zero}, and without the inclusion, we would
not be able to conclude that \lst{x : even}.  For example the derivation of
$\jchecktermsort \cdot {\mlam x x} {\sarrow {\var{zero}} {\var{even}}}$
requires the inclusion to satisfy the second premise of the
\rulename{switch} rule.
\begin{mathpar}
\infer[\rulename{$\Pi$-I}]
      {\jchecktermsort \cdot {\mlam x x} {\sarrow {\var{zero}} {\var{even}}}}
      {\infer[\rulename{switch}]
             {\jchecktermsort {\decl x {\var{zero}}} x {\var{even}}}
             {\infer[\rulename{var}]
                    {\jsynthtermsort {\decl x {\var{zero}}} x {\var{zero}}}
                    {}
            & \infer{\var{zero} \subtype \var{even}}
                    {\subdecl {\var{zero}} {\var{even}} \in \Sigma}}}
\end{mathpar}

\noindent How should we translate that derivation into a proof?  As we
saw earlier, the representation of \lst{zero -> even} as a predicate
is $\metalam N ${\lst'{x:nat} zero x -> even (N @ x)'}, and applying
this predicate to $\mlam x x$ yields the type we need the proof to
have: \lst'{x:nat} zero x -> even x'.  It is not much of a leap of the
imagination to see that one solution is simply to posit a constant of
the appropriate type:
\begin{lstlisting}
    zero-even : {x:nat} zero x -> even x.
\end{lstlisting}
Now the translation of $\mlam x x \checks {\sarrow {\var{zero}}
{\var{even}}}$ can be simply the $\eta$-expansion of this constant: $\mlam
x \mlam {\what x} {\var{zero-even}\ x\ \what x}$.  This makes intuitive
sense: the constant \lst'zero-even' witnesses the meaning of the
declaration \lst'zero <: even' under the subset interpretation.

Our example leads us to a rule: a subsorting declaration $\subdecl {s_1}
{s_2}$ will will translate into a declaration for a coercion constant
${s_1{\textit-}s_2}$.
\begin{mathpar}
\inferrule{\jsigt \Sigma {\what \Sigma}
        \\ \arefdecl {s_1} a L \in \Sigma
        \\ \arefdecl {s_2} a L \in \Sigma
        \\ \decl a K \in \Sigma
        \\ \kindtsub K {\what K\s}}
          {\jsigt {\Sigma, \subdecl {s_1} {s_2}}
                  {\what \Sigma,
                    \decl {s_1{\textit-}s_2}
                    {\metaappfive {\what K\s} a {\what{s_1}} {s_1}
                                                {\what{s_2}} {s_2}}}}
\end{mathpar}
The auxiliary judgment $\kindtsub K {\what K\s}$ yields a metafunction
describing the type of proof coercions between sorts that refine a type
family of kind $K$.  The metafunction $\what K\s$ takes five arguments: the
refined type, the formation family and predicate family for the domain of
the coercion, and the formation family and predicate family for the
codomain of the coercion.
As before, the $\Pi$ translation adds an argument to each of the
meta-arguments.
\begin{mathpar}
\inferrule{\kindtsub K {\what K}}
          {\kindtsub {\kpi x A K}
            {\metalamfive P {{Q_1}\f} {Q_1} {{Q_2}\f} {Q_2}
             {\api x A
              {\metaappfive {\what K}
                {P'} {{Q_1}\f'} {Q_1'} {{Q_2}\f'} {Q_2'}}}}
        \\ \text{(where, for each $P$, $P' = \aapp P {\expand A x}$})}
\end{mathpar}
At the base kind $\ktype$, the rule outputs the type of the coercion:
\begin{mathpar}
\inferaxiom{\kindtsub \ktype
            {\metalamfive P {{Q_1}\f} {Q_1} {{Q_2}\f} {Q_2}
                {\api {f_1} {{Q_1}\f} 
                 \api {f_2} {{Q_2}\f} 
                 \api x P {\aarrow {\aapp {\aiapp {Q_1} {f_1}} x}
                                   {\aapp {\aiapp {Q_2} {f_2}} x}}}}}
\end{mathpar}
Essentially, this is the type of coercions, given $x$, from proofs of $Q_1\
x$ to proofs of $Q_2\ x$, but in the general case, we must pass the
predicates $Q_1$ and $Q_2$ evidence that they are well-formed, so the
coercion requires formation proofs as inputs as well.


How do these coercions work?  Recall
that subsorting need only be defined at
base sorts $Q$, and there, it is simply the application-compatible,
reflexive, transitive closure of the declared relation.  For the purposes
of the translation, we employ an equivalent algorithmic formulation of
subsorting.  Following the inspiration of bidirectional typing, there are
two judgments: a \textit{checking} judgment that takes two base sorts as
inputs and a \textit{synthesis} judgment that takes one base sort as input
and outputs another base sort that is one step higher in the subsort
hierarchy.

The synthesis judgment constructs a coercion from the new coercion
constants in the signature.
\begin{mathpar}
\inferrule{\subdecl {s_1} {s_2} \in \Sigma}
          {\jsynthsubtypet {s_1} {s_2} {s_1{\textit-}s_2}}
\and
\inferrule{\jsynthsubtypet {Q_1} {Q_2} {\whatQQ}
          }
          {\jsynthsubtypet {\sapp {Q_1} N} {\sapp {Q_2} N}
                {\mapp \whatQQ N}
          }
\end{mathpar}
The checking judgment, on the other hand, constructs a \textit{meta-level}
coercion between proofs of the two sorts.  It is defined by two rules: a
rule of reflexivity and a rule to climb the subsort hierarchy.
\begin{mathpar}
\namedrule{Q_1 = Q_2}
          {\jchecksubtypet \Gamma {Q_1} {Q_2} {\metalamtwo R {R_1} {R_1}}}
          {refl}
\and
\namedrule{\jsynthsubtypet {Q_1} {Q'} \whatQQp
        \\ \jsynthsortclasst \Gamma {Q_1} P {\csort} {\what{Q_1}}
      \\\\ \jchecksubtypet \Gamma {Q'} {Q_2} F
        \\ \jsynthsortclasst \Gamma {Q'} P {\csort} {\what{Q'}}
        }
          {\jchecksubtypet \Gamma {Q_1} {Q_2}
                {\metalamtwo R {R_1}
                 {\metaapptwo F R {\mapp {\mapp {\mapp {\mapp \whatQQp 
                 {\what {Q_1}}} {\what{Q'}}} R} {R_1}}} }}
          {climb}
\end{mathpar}
The reflexivity rule's metacoercion simply returns the proof it is given,
while the climb rule composes the actual coercion $\whatQQp$ with the
metacoercion $F$.  Two extra premises generate the necessary formation
proofs.

Finally, we have described enough of the translation to explain the rule
most central to the design of LFR, the \rulename{switch} rule.
\begin{mathpar}
\namedrule{\jsynthtermsortt \Gamma R {Q'} {\what R}
        \\ \jchecksubtypet \Gamma {Q'} Q F}
          {\jchecktermsortt \Gamma R Q {\metaapptwo F R {\what R}}}
          {switch}
\end{mathpar}
The first premise produces a proof $\what R$ that $R$ satisfies property
$Q'$, and the second premise generates the meta-level proof coercion that
transforms such a proof into a proof that $R$ satisfies property $Q$.

Having sketched the translation and the role of proof irrelevance, we now
review some metatheoretic results.


\subsection{Correctness}


Our translation is both sound and complete with respect to the original
system of LF with refinement types, and so our correctness criteria will
come in two flavors.

Soundness theorems tell us that the result of a translation is well-formed.
But even more importantly than telling us that our translation is on some level
correct, they serve as an independent means of understanding the translation.
In a sense, a soundness theorem can be read as the meta-level type of a
translation judgment---a specification of its intended behavior---and just
as types serve as an organizing principle for the practicing programmer, so
too do soundness theorems serve the
thoughtful theoretician.  We explain our soundness theorems, then, not only
to demonstrate the sensibility of our translation, but also to aid the
reader in understanding its purpose.

In what follows, $\form{Q}$ represents the formation family for a base sort
$Q$.
\begin{align*}
    \form{s} &= \what s     &   \form{\aapp Q N} &= \aapp {\form{Q}} N
\end{align*}
\begin{thm}[Soundness]
\label{thm:sound-interp}
Suppose $\jctxt \Gamma {\what\Gamma}$ and $\jsigt \Sigma {\what\Sigma}$.
Then:
    \begin{enumerate}[\em(1)]
    \item If $\jchecksortclasst \Gamma S A {\what S}$
          and $\jchecktermsortt \Gamma N S {\what N}$,
          then $\jchecktermtype[\what\Sigma] {\what\Gamma} {\what N}
                {\metaapp {\what S} N}$.
        \label{cl:tr-n}
    \item If $\jsynthtermsortt \Gamma R S {\what R}$,
          then 
               $\jchecksortclasst \Gamma S A {\what S}$
          and $\jsynthtermtype[\what\Sigma] {\what\Gamma} {\what R}
                {\metaapp {\what S} {\expand A R}}$
                \\
               (for some $A$ and $\what S$).
        \label{cl:tr-r}
    \medskip
    \item If $\jchecksortclasst \Gamma S A {\what S}$
          and $\jchecktermtype \Gamma N A$,
          then $\jchecktypekind[\what\Sigma] {\what\Gamma}
                {\metaapp {\what S} N}$.
    \item If $\jsynthsortclasst \Gamma Q P L {\what Q}$,
          then 
               for some $K$, $\what L\f$, and $\what K\p$,
               \begin{enumerate}[$\bullet$]
                \item $\jclasstform \Gamma L K {\what L\f}$ and
                      $\jsynthtermtype[\what\Sigma] {\what\Gamma} {\what Q}
                        {\metaapp {\what L\f} {\form Q}}$, and
                \item $\kindtpred K {\what K\p}$ and
                      $\jsynthtypekind[\what\Sigma] {\what\Gamma} Q
                        {\metaapptwo {\what K\p} {\form Q} P}$.
               \end{enumerate}
        \label{cl:tr-q}
    \medskip
    \item If $\jclasstform \Gamma L K {\what L\f}$
          and $\jsynthtypekind \Gamma P K$,
          then $\jchecktypekind[\what\Sigma] {\what\Gamma}
                {\metaapp {\what L\f} P}$.
    \item If $\kindtpred K {\what K\p}$,\ \ 
          $\jsynthtypekind \Gamma {Q\longform} K$,
          and $\jsynthtypekind \Gamma P K$,
          then $\jkind[\what\Sigma] {\what\Gamma} {\metaapptwo {\what K\p}
                {Q\longform} P}$.
    \medskip
    \item If $\jsynthsubtypet {Q_1} {Q_2} {\whatQQ}$,\ \ 
             $\jsynthsortclass \Gamma {Q_1} P L K$,\ \ 
             $\jsynthtypekind \Gamma P K$,
          and $\kindtsub K {\what K\s}$,\ \ 
          then $\jsynthsortclass \Gamma {Q_2} P L K$
          and $\jsynthtermtype {\what\Gamma} {\whatQQ}
                    {\metaappfive {\what K\s} P {\form {Q_1}} {Q_1}
                                                {\form {Q_2}} {Q_2}}$.
    \item If $\jsynthtermtype \Gamma R P$,\ \ 
             $\jchecksortclasst \Gamma {Q_i} P {\what {Q_i}}$,\ \ 
             $\jchecksubtypet \Gamma {Q_1} {Q_2} F$,
          and $\jsynthtermtype {\what\Gamma} {R_1} {\metaapp {\what {Q_1}} R}$,\ \ 
          then $\jsynthtermtype {\what\Gamma} {\metaapptwo F R {R_1}}
                                              {\metaapp {\what {Q_2}} R}$.
        \label{cl:tr-f}
    \item If $\kindtsub K {\what K\s}$,\ \ 
             $\kindtpred K {\what K\p}$,\ \ 
             $\jsynthtypekind {\Gamma} P K$,\ \ 
             $\jsynthtypekind {\Gamma} {{Q_i}\f} K$, and
             $\jsynthtypekind {\what\Gamma} {Q_i} {\metaapptwo {\what K\p}
                                                    {{Q_i}\f} P}$,
          then $\jchecktypekind {\what\Gamma}
                    {\metaappfive {\what K\s} P {{Q_1}\f} {Q_1} {{Q_2}\f}
                    {Q_2}}$.
    \end{enumerate}
\end{thm}
\proof
By induction on each clause's main input derivation.  Several clauses must
be proved mutually; for instance, clauses \ref{cl:tr-n}, \ref{cl:tr-r},
\ref{cl:tr-f}, and \ref{cl:tr-q} are all mutual, since the rules for
translating terms refer to the translation of subsorting, the rules for
translating subsorting refer to the class synthesis translation, and since
sorts may be dependent, the rules for class synthesis translation refer
back to the term translation.
\qed

\noindent The proofs use entirely standard syntactic methods, but they
appeal to several key lemmas about the structure of the translation.

\begin{lem}[Erasure]
If $\judgetrans J X$, then $\judge J$.
\end{lem}
\proof
Straightforward induction on the structure of the translation derivation.
The translation rules are premise-wise strictly more restrictive than the
original LFR rules, except for the subsorting rules, which are also more
restrictive in the sense that they force rules to be applied in a certain
order.
\qed

\begin{lem}[Reconstruction]
If $\judge J$, then for some $X$, $\judgetrans J X$.
\end{lem}
\proof
By induction on the structure of the LFR derivation.  The cases for the
subsorting rules require us to demonstrate that an LFR subsorting
derivation can be put into ``algorithmic form'', with all uses of
reflexivity and transitivity outermost and right-nested, like the
algorithmic translation rules \rulename{refl} and \rulename{climb}.
We also make use of the tacit assumption that the judgment $\judge J$
itself is well-formed, e.g. if $\mathcal J = N \checks S$, then
$\jchecksortclass \Gamma S A$, which ensures that we will have the
necessary formation premises when we need to apply the \rulename{climb}
rule.
\qed

\noindent Erasure and reconstruction substantiate the claim that our
translation is der\-i\-va\-tion-directed by allowing us to move freely
between translation judgments and ordinary ones.  Using erasure and
reconstruction, we can leverage all of the LFR metatheory without
reproving it for translation judgments.  For example, several cases
require us to substitute into a translation derivation: we can apply
erasure, appeal to LFR's substitution theorem, and invoke
reconstruction to get the output we require.

But since reconstruction only gives us \textit{some} output $X$, we may not
know that it is the one that suits our needs.  Therefore, we usually require
another lemma, compositionality, to tell us that the translation commutes
with substitution.  There are several such lemmas; we show here the one for
sort translation.
\begin{lem}[Compositionality]
Let $\sigma$ denote $\hsubst {} M x A {}$.
\begin{enumerate}[\em(1)]
\item If $\jchecksortclasst {\GammaL, \mrefdecl x \_, \GammaR} S A {\what S}$
      and $\jchecksortclasst {\GammaL, \sigma \GammaR} {\sigma S} {\sigma A}
                             {\what {S'}}$,
      then $\sigma \metaapp {\what S} N
           = \metaapp {\what {S'}} {\sigma N}$,
\item If $\jclasstform {\GammaL, \mrefdecl x \_, \GammaR} L K {\what L}$
      and $\jclasstform {\Gamma, \sigma \GammaR} {\sigma L} {\sigma K}
                    {\what {L'}}$,
      then $\sigma \metaapp {\what L} P
            = \metaapp {\what {L'}} {\sigma P}$,
\end{enumerate}
and similarly for $\kindtsub K {\what K\s}$ and $\kindtpred K {\what K\p}$.
\end{lem}
\proof
Straightforward induction using functionality of hereditary substitution.
The base case of the first clause leverages the irrelevance introduced in
the \rulename{$Q$-F} translation rule: both sort formation derivations will
have a premise outputting evidence for the well-formedness of the sort, and
there is no guarantee they will output the \textit{same} evidence, but
since the evidence is relegated to an irrelevant position, its identity is
ignored.  The second clause's $\Pi$ case appeals to the first clause, since
$\Pi$ classes contain sorts.
\qed

Finally, there is a lemma demonstrating that proof variables only ever
occur irrelevantly, so substituting for them cannot change the identity of
a sort or class meta-function output by the translation.
\begin{lem}[Proof Variable Substitution]
\hfill
\begin{enumerate}[\em(1)]
    \item If $\jchecksortclasst {\GammaL, \rdecl x {S_0} {A_0}, \GammaR}
                    S A {\what S}$
          then $\hsubst a M {\what x} {A_0} {\metaapp {\what S} N}
               = \metaapp {\what S} {\hsubst n M {\what x} {A_0} N}$.
    \item If $\jclasstform {\GammaL, \rdecl x {S_0} {A_0}, \GammaR} L K {\what L}$
          then $\hsubst a M {\what x} {A_0} {\metaapp {\what L} P}
               = \metaapp {\what L} {\hsubst p M {\what x} {A_0} P}$.
\end{enumerate}
\end{lem}
\proof
Straightforward induction, noting in the base case, the \rulename{$Q$-F}
rule, the only term that could depend on $\what x$ is in an irrelevant
position.
\qed

\noindent Completeness theorems tell us that our target is not too
rich: that everything we find evidence of in the codomain of the
translation actually holds true in its domain.  While important for
establishing general correctness, completeness theorems are not as
informative as soundness theorems, so we give here only the cases for
terms---and in any case, those are the only cases we require to
fulfill our goal of preserving adequacy.

In stating completeness, we syntactically isolate the set of terms that
could represent proofs using metavariables $\what R$ and $\what N$.
\begin{align*}
  \what R &::= \what c \mid \what x \mid {\what R}\ N\ {\what N}
            \mid \mfst {\what R} \mid \msnd {\what R} \\
  \what N &::= \what F \mid \mlam x \mlam {\what x} {\what N}
            \mid \mpair {\what {N_1}} {\what {N_2}} \mid \munit \\
  \what F &::= \what R \mid \whatQQ\ {\what {Q_1}}\ {\what {Q_2}}\ R\ F \\
  \whatQQ &::= s_1{\textit-}s_2 \mid \whatQQ\ N \\
  \what Q &::= \what s/i \mid \what Q\ N\ \what N
            \mid \mfst {\what Q} \mid \msnd {\what Q}
\end{align*}

\begin{thm}[Completeness]
Suppose $\jctxt \Gamma {\what\Gamma}$ and $\jsigt \Sigma {\what\Sigma}$.
Then:
    \begin{enumerate}[\em(1)]
    \item If $\jchecksortclasst \Gamma S A {\what S}$
          and $\jchecktermtype[\what\Sigma] {\what\Gamma} {\what N}
                {\metaapp {\what S} N}$,
          then $\jchecktermsort \Gamma N S$.
    \item If $\jsynthtermsort[\what\Sigma] {\what\Gamma} {\what R} B$,
          then
               $\jchecksortclasst \Gamma S A {\what S}$,
               $B = \metaapp {\what S} {\expand A R}$,
           and $\jsynthtermsort \Gamma R S$
               (for some $S$, $A$, $\what S$, and $R$).
    \item If $\jsynthtermsort {\what\Gamma} {\what F}
                    {\aapp {\aiapp Q {\what Q}} R}$,
          then $\jchecktermsort {\what\Gamma} {\what R} Q$.
    \item If $\jsynthtermtype {\what\Gamma} {\whatQQ} B$,
          then
              $\kindtsub K {\what K\s}$,\ \ 
              $B = \metaappfive {\what K\s} P {\form {Q_1}} {Q_1}
                                              {\form {Q_2}} {Q_2}$,
          and $\jsubtype {Q_1} {Q_2}$
          (for some $K$, $\what K\s$, $P$, ${Q_1}$, and ${Q_2}$).
    \item If $\jsynthtermsort {\what\Gamma} {\what Q} B$,
          then 
              $\jclasstform \Gamma L K {\what L\f}$,\ \ 
              $B = \metaapp {\what L\f} {\form Q}$,
          and $\jsynthsortclass \Gamma Q P L K$
               (for some $L$, $K$, $\what L\f$, and $Q$).
    \end{enumerate}
\end{thm}
\proof
By induction over the structure of the proof term.
\qed


\noindent Adequacy of a representation is generally shown by exhibiting a
compositional bijection between informal entities and terms of certain
LFR sorts.
Since we have undertaken a subset interpretation, the set of terms of any
LFR sort are unchanged by translation, and so any bijective correspondence
between those terms and informal entities remains after translation.
Furthermore, soundness and completeness tell us that our interpretation
preserves and reflects the derivability of any refinement type judgments
over those terms.
Thus, we have achieved our main goal: any adequate LFR representation can
be translated to an adequate LFI representation.


\section{Conclusion}
\label{sect:conclusion}

\noindent Logical frameworks are metalanguages specifically designed
so that common concepts and notations in logic and the theory of
programming languages can be represented elegantly and concisely.
LF~\cite{Harper93jacm} intrinsically supports $\alpha$-conversion,
capture-avoiding substitution, and hypothetical and parametric
judgments, but as with any such enterprise, certain patterns fall out
of its scope and must be encoded indirectly. 
SPACE One pattern is the ability to form regular subsets of types
already defined.  We address this by extending LF with type
refinements, leveraging the modern view of LF as a calculus of
canonical forms to obtain a metatheoretically simple yet expressive
system, LFR\@.  Another pattern is to ignore the identities of proofs,
relying only on their existence.  This is addressed in LF extended
with proof irrelevance, LFI~\cite{Pfenning01lics,ReedPfenning}.  We
have shown that our system of refinement types can be mapped into LFI
in a bijective manner, preserving adequacy theorems for LFR
representations in LFI\@.

In the methodology of logical frameworks research, it is important to
understand the cost of such a translation: how much more complicated
are encodings in the target framework, and how much more difficult is
it to work with them?  We cannot measure this cost precisely, but we
hope it is evident from the definition of the translation and the
examples that the price is considerable.  Even if in special cases
more direct encodings are possible, we believe our general translation
could not be simplified much, given the explicit goal to preserve the
adequacy of representations.  Other translations from programming
languages, such as coercion interpretations where sorts are translated
to distinct types and subsorting to coercions, appear even more
complex because adequacy depends on certain functional equalities
between coercions.  Our preliminary conclusion is that refinement
types in logical frameworks provide elegant and immediate
representations that are not easy to simulate without them, providing
a solid argument for their inclusion in the next generation of
frameworks.

Of course, much work remains to be done before refinement types can be
considered a practical addition.  First, it will be necessary to develop a
sufficiently complete algorithm for reconstructing the sorts of implicitly
$\Pi$-quantified metavariables in order to allow the elegant encodings we
imagine without burdensome redundancy.  Furthermore, it would be useful to
have a logic programming interpretation of LFR declarations and the ability
to perform analyses like coverage and termination checking on declarations
\textit{qua} programs; to enable such an interpretation, we will have to
develop an algorithm for sorted unification, generalizing existing work on
pattern unification in the context of logical frameworks.  It may also be a
worthwhile endeavor to formalize the metatheory of LFR and its subset
interpretation in a metalogical framework or proof assistant; although we
have avoided doing so due to the high cost of working around current
technological limitations in proof assistants, the present work has been
carried out in sufficient detail that formalization should not be
particularly difficult beyond the technical challenge of representing a
dependently typed calculus.

Refinement types have been also been proposed for functional
programming~\cite{Freeman94,Dunfield04:Tridirectional,Davies05}, most
recently in conjunction with a limited form of dependent
types~\cite{Dunfield07thesis}.  Proof irrelevance is already integrated in
this setting, and also available in general type theories such as NuPrl or
Coq.  One can ask the same question here: Can we simply eliminate
refinement types and just work with dependent types and proof irrelevance?
The results in this paper lend support to the conjecture that this can be
accomplished by a uniform translation.  On the other hand, just as here, it
seems there would likely be a high cost in terms of brevity in order to
maintain a bijection between well-sorted data in the source and dependently
well-typed data in the target of the translation.


\subsubsection*{Acknowledgements.}
Thanks to Jason Reed for many fruitful discussions on
the topic of proof irrelevance.
Thanks to the anonymous referees for offering insightful commentary
on how to clarify our presentation.

\bibliographystyle{alpha}
\bibliography{refs}

\vfill\eject

\appendix
\section{Complete LFR Rules}
\label{app:lfr-rules}

\noindent In the judgment forms below, superscript $+$ and $-$
indicate a judgment's ``inputs'' and ``outputs'', respectively.

\subsection{Grammar}
\begin{align*}
\grammarheader{Kind level} \\
K &::= \ktype
    \mid \kpi x A K
  & \text{kinds}
\\
L &::= \csort
    \mid \ccpi x S A L
    \mid \ctop
    \mid \cinter {L_1} {L_2}
  & \text{classes}
\\ \\
\grammarheader{Type level} \\
P &::= a
    \mid \aapp P N
  & \text{atomic type families}
\\
A &::= P
    \mid \api x {A_1} {A_2}
  & \text{canonical type families}
\\ \\
Q &::= s
    \mid \sapp Q N
  & \text{atomic sort families}
\\
S &::= Q
    \mid \sspi x {S_1} {A_1} {S_2}
    \mid \sstop
    \mid \sinter {S_1} {S_2}
  & \text{canonical sort families}
\\ \\
\grammarheader{Term level} \\
R &::= c
    \mid x
    \mid \mapp R N
  & \text{atomic terms}
\\
N &::= R
    \mid \mlam x N
  & \text{canonical terms}
\\ \\
\grammarheader{Signatures and contexts} \\
\Sigma &::= \cdot
    \mid \Sigma, D
  & \text{signatures}
\\
D &::= \decl a K
    \mid \decl c A
    \mid \arefdecl s a L
    \mid \subdecl {s_1} {s_2}
    \mid \mrefdecl c S
  & \text{declarations}
\\
\Gamma &::= \cdot
    \mid \Gamma, \rdecl x S A
  & \text{contexts}
\end{align*}

\subsection{Expansion and Substitution}

All bound variables are tacitly assumed to be sufficiently fresh.

\rulesheadernobreak{\erasedeps A = \alpha}
\begin{align*}
\alpha, \beta &::= a \mid \aarrow {\alpha_1} {\alpha_2} \\ \\
\erasedeps a &= a \\
\erasedeps {\aapp P N} &= \erasedeps P \\
\erasedeps {\api x A B} &= {\aarrow {\erasedeps A} {\erasedeps B}}
\end{align*}

\rulesheadernobreak{\expand \alpha R = N}
\begin{align*}
\expand a R &= R \\
\expand {\aarrow \alpha \beta} R
    &= {\mlam x {\expand \beta {\mapp R {\expand \alpha x}}}}
\end{align*}


\begin{mathpar}
\rulesheader{\jhsubst n {N_0} {x_0} {\alpha_0} N {N'}}
\and
\inferrule{\jhsubst {rn} {N_0} {x_0} {\alpha_0} R {(N, a)}}
          {\jhsubst n {N_0} {x_0} {\alpha_0} R {N}}
\and
\inferrule{\jhsubst {rr} {N_0} {x_0} {\alpha_0} R {R'}}
          {\jhsubst n {N_0} {x_0} {\alpha_0} R {R'}}
\and
\inferrule{\jhsubst n {N_0} {x_0} {\alpha_0} N {N'}}
          {\jhsubst n {N_0} {x_0} {\alpha_0} {\mlam x N} {\mlam x {N'}}}
\end{mathpar}

\begin{mathpar}
\rulesheader{\jhsubst {rr} {N_0} {x_0} {\alpha_0} R {R'}}
\and
\inferrule{x \neq x_0}
          {\jhsubst {rr} {N_0} {x_0} {\alpha_0} x x}
\and
\inferaxiom {\jhsubst {rr} {N_0} {x_0} {\alpha_0} c c}
\and
\inferrule{\jhsubst {rr} {N_0} {x_0} {\alpha_0} {R_1} {R_1'}
        \\ \jhsubst n {N_0} {x_0} {\alpha_0} {N_2} {N_2'}}
          {\jhsubst {rr} {N_0} {x_0} {\alpha_0} {\mapp {R_1} {N_2}}
                                                {\mapp {R_1'} {N_2'}}}
\end{mathpar}

\begin{mathpar}
\rulesheader{\jhsubst {rn} {N_0} {x_0} {\alpha_0} R {(N', \alpha')}}
\and
\namedaxiom {\jhsubst {rn} {N_0} {x_0} {\alpha_0} {x_0} {(N_0, \alpha_0)}}
            {subst-rn-var}
\and
\namedrule{\jhsubst {rn} {N_0} {x_0} {\alpha_0} {R_1}
                            {(\mlam x {N_1}, \aarrow {\alpha_2} {\alpha_1})}
        \\\\ \jhsubst n {N_0} {x_0} {\alpha_0} {N_2} {N_2'}
        \\ \jhsubst n {N_2'} x {\alpha_2} {N_1} {N_1'}}
          {\jhsubst {rn} {N_0} {x_0} {\alpha_0} {\mapp {R_1} {N_2}}
                                                {(N_1', \alpha_1)}}
          {subst-rn-$\beta$}
\and
\text{(Substitution for other syntactic categories (q, p, s, a, l, k,
$\gamma$) is compositional.)}
\end{mathpar}

\newpage
\subsection{Kinding}

\begin{mathpar}
\rulesheader{\jclass [\Sigma] \Gamma {L^+} {K^+}}
\and
\inferaxiom{\jclass \Gamma \csort \ktype}
\and
\inferrule{\jchecksortclass \Gamma S A
        \\ \jclass {\Gamma, \rdecl x S A} L K}
          {\jclass \Gamma {\ccpi x S A L} {\kpi x A K}}
\\
\and
\inferaxiom {\jclass \Gamma \ctop K}
\and
\inferrule{\jclass \Gamma {L_1} K
        \\ \jclass \Gamma {L_2} K}
          {\jclass \Gamma {\cinter {L_1} {L_2}} K}
\end{mathpar}

\begin{mathpar}
\rulesheader{\jsynthsortclass [\Sigma] \Gamma {Q^+} {P^-} {L^-} {K^-}}
\and
\inferrule{\arefdecl s a L \in \Sigma
        }
          {\jsynthsortclass \Gamma s a L K}
\\
\and
\inferrule{\jsynthsortclass \Gamma Q P {\ccpi x S A L} {\kpi x A K}
        \\ \jchecktermsort \Gamma N S
        \\ \jhsubst l N x A L {L'}
        }
          {\jsynthsortclass \Gamma {\sapp Q N} {\aapp P N} {L'} {K'}}
\\
\and
\inferrule{\jsynthsortclass \Gamma Q P {\cinter {L_1} {L_2}} K}
          {\jsynthsortclass \Gamma Q P {L_1} K}
\and
\inferrule{\jsynthsortclass \Gamma Q P {\cinter {L_1} {L_2}} K}
          {\jsynthsortclass \Gamma Q P {L_2} K}
\end{mathpar}

\begin{mathpar}
\rulesheader{\jchecksortclass [\Sigma] \Gamma {S^+} {A^+}}
\and
\namedrule{\jsynthsortclass \Gamma Q {P'} L K
        \\ P' = P
        \\ L = \csort
        }
          {\jchecksortclass \Gamma Q P}
          {$Q$-F}
\\
\and
\namedrule{\jchecksortclass \Gamma {S} {A}
        \\ \jchecksortclass {\Gamma, \rdecl x {S} {A}} {S'} {A'}}
          {\jchecksortclass \Gamma {\sspi x {S} {A} {S'}} {\api x {A} {A'}}}
          {$\Pi$-F}
\\
\and
\namedaxiom {\jchecksortclass \Gamma \sstop A} {$\top$-F}
\and
\namedrule{\jchecksortclass \Gamma {S_1} A
        \\ \jchecksortclass \Gamma {S_2} A}
          {\jchecksortclass \Gamma {\sinter {S_1} {S_2}} A}
          {$\intersect$-F}
\end{mathpar}

\textbf{Note:} no intro rules for classes $\ctop$ and $\cinter {L_1} {L_2}$.

\subsection{Typing}

\begin{mathpar}
\rulesheader{\jsynthtermsort [\Sigma] \Gamma {R^+} {S^-}}
\and
\namedrule{\mrefdecl c S \in \Sigma
          }
          {\jsynthtermsort \Gamma c S}
          {const}
\and
\namedrule{\rdecl x S A \in \Gamma}
          {\jsynthtermsort \Gamma x S}
          {var}
\and
\namedrule{\jsynthtermsort \Gamma {R_1} {\sspi x {S_2} {A_2} {S}}
        \\ \jchecktermsort \Gamma {N_2} {S_2}
        \\ \jhsubst s {N_2} x {A_2} {S} {S'}
        }
          {\jsynthtermsort \Gamma {\mapp {R_1} {N_2}} {S'}}
          {$\Pi$-E}
\\
\and
\namedrule{\jsynthtermsort \Gamma R {\sinter {S_1} {S_2}}}
          {\jsynthtermsort \Gamma R {S_1}}
          {$\intersect$-E$_1$}
\and
\namedrule{\jsynthtermsort \Gamma R {\sinter {S_1} {S_2}}}
          {\jsynthtermsort \Gamma R {S_2}}
          {$\intersect$-E$_2$}
\end{mathpar}

\begin{mathpar}
\rulesheader{\jchecktermsort [\Sigma] \Gamma {N^+} {S^+}}
\and
\namedrule{\jsynthtermsort \Gamma R {Q'}
        \\ \jchecksubtype \Gamma {Q'} Q P}
          {\jchecktermsort \Gamma R Q}
          {switch}
\\
\and
\namedrule{\jchecktermsort {\Gamma, \rdecl x {S} {A}} N {S'}}
          {\jchecktermsort \Gamma {\mlam x N} {\sspi x {S} {A} {S'}}}
          {$\Pi$-I}
\\
\and
\namedaxiom {\jchecktermsort \Gamma N \sstop}
          {$\top$-I}
\and
\namedrule{\jchecktermsort \Gamma N {S_1}
        \\ \jchecktermsort \Gamma N {S_2}}
          {\jchecktermsort \Gamma N {\sinter {S_1} {S_2}}}
          {$\intersect$-I}
\end{mathpar}

\begin{mathpar}
\rulesheader{\jchecksubtype \Gamma {Q_1^+} {Q_2^+} {P^+}}
\and
\inferrule{Q_1 = Q_2}
          {\jchecksubtype \Gamma {Q_1} {Q_2} P}
\and
\inferrule{\jsynthsubtype \Gamma {Q_1} {Q'} {P'} L K
        \\ \jchecksubtype \Gamma {Q'} {Q_2} P}
          {\jchecksubtype \Gamma {Q_1} {Q_2} P}
%
\and
\inferrule{\subdecl {s_1} {s_2} \in \Sigma
        }
          {\jsynthsubtype \Gamma {s_1} {s_2} a L K}
\and
\inferrule{\jsynthsubtype \Gamma {Q_1} {Q_2} P {\cpi x S L} {\kpi x A K}
        }
          {\jsynthsubtype \Gamma {\sapp {Q_1} N} {\sapp {Q_2} N}
                                 {\aapp P N} {L'} {K'}}
\end{mathpar}

\subsection{Signatures and Contexts}

\begin{mathpar}
\rulesheader{\jsig \Sigma}
\and
\inferaxiom{\jsig \cdot}
\and
\inferrule{\jsig \Sigma
        \\ \jkind [\eraserefs\Sigma] \cdot K
        \\ \decl a {K'} \not\in \Sigma}
          {\jsig {\Sigma, \decl a K}}
\and
\inferrule{\jsig \Sigma
        \\ \jchecktypekind [\eraserefs\Sigma] \cdot A
        \\ \decl c {A'} \not\in \Sigma}
          {\jsig {\Sigma, \decl c A}}
\and
\inferrule{\jsig \Sigma
        \\ \decl a K \in \Sigma
        \\ \jclass [\Sigma] \cdot L K
        \\ \arefdecl s {a'} {L'} \not\in \Sigma}
          {\jsig {\Sigma, \arefdecl s a L}}
\and
\inferrule{\jsig \Sigma
        \\ \decl c A \in \Sigma
        \\ \jchecksortclass [\Sigma] \cdot S A
        \\ \mrefdecl c {S'} \not\in \Sigma}
          {\jsig {\Sigma, \mrefdecl c S}}
\and
\inferrule{\jsig \Sigma
        \\ \arefdecl {s_1} a L \in \Sigma
        \\ \arefdecl {s_2} a L \in \Sigma}
          {\jsig {\Sigma, \subdecl {s_1} {s_2}}}
\end{mathpar}

\begin{mathpar}
\rulesheader{\jctx [\Sigma] \Gamma}
\and
\inferaxiom{\jctx \cdot}
\and
\inferrule{\jctx \Gamma
        \\ \jchecksortclass \Gamma S A}
          {\jctx {\Gamma, \rdecl x S A}}
\end{mathpar}



\newpage
\section{Complete Translation Rules}
\label{app:trans-rules}

\noindent In the judgment forms below, superscript $+$ and $-$
indicate a judgment's ``inputs'' and ``outputs'', respectively.

\subsection{Kinding}


\begin{mathpar}
\rulesheader{\jclasstform [\Sigma] \Gamma {L^+} {K^+} {\what{L}^-}}
\and
\inferaxiom{\jclasstform \Gamma \csort \ktype {\metalam {Q\f} {Q\f}}}
\and
\inferrule{\jchecksortclasst \Gamma S A {\what S}
        \\ \jclasstform {\Gamma, \rdecl x S A} L K {\what L}}
          {\jclasstform \Gamma {\ccpi x S A L} {\kpi x A K}
                    {\metalam {Q\f} {\kpi x A
                                {\kpi {\what x} {\metaapp {\what S} {\expand A x}} 
                                 {\metaapp {\what L} {\aapp {Q\f} {\expand A x}}}}}}}
\\
\and
\inferaxiom {\jclasstform \Gamma \ctop K {\metalam {Q\f} \kunit}}
\and
\inferrule{\jclasstform \Gamma {L_1} K {\what {L_1}}
        \\ \jclasstform \Gamma {L_2} K {\what {L_2}}}
          {\jclasstform \Gamma {\cinter {L_1} {L_2}} K
                    {\metalam {Q\f} {\kprod {\metaapp {\what {L_1}} {Q\f}}
                                            {\metaapp {\what {L_2}} {Q\f}}}}}
\end{mathpar}



\begin{mathpar}
\rulesheader{\jsynthsortclasst [\Sigma] \Gamma {Q^+} {P^-} {L^-} {\what Q^-}}
\and
\inferrule{\arefdecl s a L \in \Sigma
        }
          {\jsynthsortclasst \Gamma s a L {\what s / i}}
\and
\inferrule{\jsynthsortclasst \Gamma Q P {\ccpi x S A L} {\what Q}
        \\ \jchecktermsortt \Gamma N S {\what N}
        \\ \jhsubst l N x A L {L'}
        }
          {\jsynthsortclasst \Gamma {\sapp Q N} {\aapp P N} {L'}
                             {\aapp {\aapp {\what Q} N} {\what N}}} 
\and
\inferrule{\jsynthsortclasst \Gamma Q P {\cinter {L_1} {L_2}} {\what Q}}
          {\jsynthsortclasst \Gamma Q P {L_1} {\afst {\what Q}}}
\and
\inferrule{\jsynthsortclasst \Gamma Q P {\cinter {L_1} {L_2}} {\what Q}}
          {\jsynthsortclasst \Gamma Q P {L_2} {\asnd {\what Q}}}
\end{mathpar}


\begin{mathpar}
\rulesheader{\jchecksortclasst [\Sigma] \Gamma {S^+} {A^+} {\what S^-}}
\and
\namedrule{\jsynthsortclasst \Gamma Q {P'} L {\what Q}
        \\ P' = P
        \\ L = \csort
        }
          {\jchecksortclasst \Gamma Q P
            {\metalam N {\aapp {\aiapp Q {\what Q}} N}}}
        {$Q$-F}
\and
\namedrule{\jchecksortclasst \Gamma {S} {A} {\what S}
        \\ \jchecksortclasst {\Gamma, \rdecl x {S} {A}} {S'} {A'} {\what {S'}}}
          {\jchecksortclasst \Gamma {\sspi x {S} {A} {S'}} {\api x {A} {A'}}
                             {\metalam N
                              {\api x A
                               {\api {\what x} {\metaapp {\what S} {\expand A x}} 
                                {\metaapp {\what {S'}} {\revapp N x}}}}}}
          {$\Pi$-F}
\\
\and
\namedaxiom {\jchecksortclasst \Gamma \sstop A {\metalam N \aunit}} {$\top$-F}
\and
\namedrule{\jchecksortclasst \Gamma {S_1} A {\what {S_1}}
        \\ \jchecksortclasst \Gamma {S_2} A {\what {S_2}}}
          {\jchecksortclasst \Gamma {\sinter {S_1} {S_2}} A
                             {\metalam N {\aprod {\metaapp {\what {S_1}} N}
                                                 {\metaapp {\what {S_2}} N}}}}
          {$\intersect$-F}
\end{mathpar}

\textbf{Note:} no intro rules for classes $\ctop$ and $\cinter {L_1} {L_2}$.


\begin{mathpar}
\rulesheader{\kindtpred {K^+} {\what K}^-}
\and
\inferaxiom{\kindtpred \ktype
                {\metalamtwo {Q\longform} P
                    {\aiarrow {Q\longform} {\aarrow P \ktype}}}}
\and
 \inferrule{\kindtpred K {\what K}}
           {\kindtpred {\kpi x A K}
                {\metalamtwo {Q\f} P
                    {\kpi x A {\metaapptwo {\what K}
                                           {\aapp {Q\f} {\expand A x}}
                                           {\aapp P {\expand A x}}}}}} 
\end{mathpar}



\begin{mathpar}
\rulesheader{\kindtsub {K^+} {{\what K}^-}}
\and
\inferaxiom{\kindtsub \ktype
            {\metalamfive P {{Q_1}\f} {Q_1} {{Q_2}\f} {Q_2}
                {\api {f_1} {{Q_1}\f} 
                 \api {f_2} {{Q_2}\f} 
                 \api x P {\aarrow {\aapp {\aiapp {Q_1} {f_1}} x}
                                   {\aapp {\aiapp {Q_2} {f_2}} x}}}}}
\and
\inferrule{\kindtsub K {\what K}}
          {\kindtsub {\kpi x A K}
            {\metalamfive P {{Q_1}\f} {Q_1} {{Q_2}\f} {Q_2}
             {\api x A
              {\metaappfive {\what K}
                {P'} {{Q_1}\f'} {Q_1'} {{Q_2}\f'} {Q_2'}}}}
        \\ \text{(where, for each $P$, $P' = \aapp P {\expand A x}$})}
\end{mathpar}

\newpage
\subsection{Typing}

\begin{mathpar}
\rulesheader{\jsynthtermsortt [\Sigma] \Gamma {R^+} {S^-} {{\what R}^-}}
\and
\namedrule{\mrefdecl c S \in \Sigma
          }
          {\jsynthtermsortt \Gamma c S {\what c}}
          {const}
\and
\namedrule{\rdecl x S A \in \Gamma}
          {\jsynthtermsortt \Gamma x S {\what x}}
          {var}
\and
\namedrule{\jsynthtermsortt \Gamma {R_1} {\sspi x {S_2} {A_2} {S}} {\what {R_1}}
        \\ \jchecktermsortt \Gamma {N_2} {S_2} {\what {N_2}}
        \\ \jhsubst s {N_2} x {A_2} {S} {S'}
        }
          {\jsynthtermsortt \Gamma {\mapp {R_1} {N_2}} {S'}
                            {\mapp {\mapp {\what {R_1}} {N_2}} {\what {N_2}}}} 
          {$\Pi$-E}
\\
\and
\namedrule{\jsynthtermsortt \Gamma R {\sinter {S_1} {S_2}} {\what R}}
          {\jsynthtermsortt \Gamma R {S_1} {\mfst {\what R}}}
          {$\intersect$-E$_1$}
\and
\namedrule{\jsynthtermsortt \Gamma R {\sinter {S_1} {S_2}} {\what R}}
          {\jsynthtermsortt \Gamma R {S_2} {\msnd {\what R}}}
          {$\intersect$-E$_2$}
\end{mathpar}


\begin{mathpar}
\rulesheader{\jchecktermsortt [\Sigma] \Gamma {N^+} {S^+} {\what N^-}}
\and
\namedrule{\jsynthtermsortt \Gamma R {Q'} {\what R}
        \\ \jchecksubtypet \Gamma {Q'} Q F}
          {\jchecktermsortt \Gamma R Q {\metaapptwo F R {\what R}}}
          {switch}
\and
\namedrule{\jchecktermsortt {\Gamma, \rdecl x {S} {A}} N {S'} {\what N}}
          {\jchecktermsortt \Gamma {\mlam x N} {\sspi x {S} {A} {S'}}
                            {\mlam x \mlam {\what x} {\what N}}}
          {$\Pi$-I}
\\
\and
\namedaxiom {\jchecktermsortt \Gamma N \sstop \munit}
          {$\top$-I}
\and
\namedrule{\jchecktermsortt \Gamma N {S_1} {\what {N_1}}
        \\ \jchecktermsortt \Gamma N {S_2} {\what {N_2}}}
          {\jchecktermsortt \Gamma N {\sinter {S_1} {S_2}}
                            {\mpair {\what N_1} {\what N_2}}}
          {$\intersect$-I}
\end{mathpar}
\begin{mathpar}
\rulesheader{\jchecksubtypet \Gamma {Q_1^+} {Q_2^+} {F^-}}
\and
\namedrule{Q_1 = Q_2}
          {\jchecksubtypet \Gamma {Q_1} {Q_2} {\metalamtwo R {R_1} {R_1}}}
          {refl}
\and
\namedrule{\jsynthsubtypet {Q_1} {Q'} \whatQQp
        \\ \jsynthsortclasst \Gamma {Q_1} P {\csort} {\what{Q_1}}
      \\\\ \jchecksubtypet \Gamma {Q'} {Q_2} F
        \\ \jsynthsortclasst \Gamma {Q'} P {\csort} {\what{Q'}}
        }
          {\jchecksubtypet \Gamma {Q_1} {Q_2}
                {\metalamtwo R {R_1}
                 {\metaapptwo F R {\mapp {\mapp {\mapp {\mapp \whatQQp 
                 {\what {Q_1}}} {\what{Q'}}} R} {R_1}}} }}
          {climb}
\end{mathpar}


\begin{mathpar}
\rulesheader{\jsynthsubtypet {Q_1^+} {Q_2^-} {\whatQQ^-}}
\and
\inferrule{\subdecl {s_1} {s_2} \in \Sigma}
          {\jsynthsubtypet {s_1} {s_2} {s_1{\textit-}s_2}}
\and
\inferrule{\jsynthsubtypet {Q_1} {Q_2} {\whatQQ}
          }
          {\jsynthsubtypet {\sapp {Q_1} N} {\sapp {Q_2} N}
                {\mapp \whatQQ N}
          }
\end{mathpar}

\newpage
\subsection{Signatures and Contexts}

\begin{mathpar}
\rulesheader{\jsigt {\Sigma^+} {{\what \Sigma}^-}}
\and
\inferaxiom{\jsigt \cdot \cdot}
\and
\inferrule{\jsigt \Sigma {\what \Sigma}
        \\ \jkind [\eraserefs\Sigma] \cdot K
        \\ \decl a {K'} \not\in \Sigma}
          {\jsigt {\Sigma, \decl a K} {\what \Sigma, \decl a K}}
\and
\inferrule{\jsigt \Sigma {\what \Sigma}
        \\ \jchecktypekind [\eraserefs\Sigma] \cdot A
        \\ \decl c {A'} \not\in \Sigma}
          {\jsigt {\Sigma, \decl c A} {\what \Sigma, \decl c A}}
\and
\inferrule{\jsigt \Sigma {\what \Sigma}
        \\ \decl a K \in \Sigma
        \\ \jclasstform [\Sigma] \cdot L K {\what {L_\text{f}}}
        \\ \kindtpred K {\what {K_\text{p}}}
        \\ \arefdecl s {a'} {L'} \not\in \Sigma}
          {\jsigt {\Sigma, \arefdecl s a L}
                  {\what \Sigma,
                   \ \decl {\what s} K,
                   \ \decl {\what s / i}
                        {\metaapp {\what {L_\text{f}}} {\what s}},
                   \ \decl s {\metaapptwo {\what {K_\text{p}}} {\what s} a}}}
\and
\inferrule{\jsigt \Sigma {\what \Sigma}
        \\ \decl c A \in \Sigma
        \\ \jchecksortclasst [\Sigma] \cdot S A {\what S}
        \\ \mrefdecl c {S'} \not\in \Sigma}
          {\jsigt {\Sigma, \mrefdecl c S}
                  {\what \Sigma, \decl {\what c}
                                       {\metaapp {\what S} {\expand A c}}}}
\and
\inferrule{\jsigt \Sigma {\what \Sigma}
        \\ \arefdecl {s_1} a L \in \Sigma
        \\ \arefdecl {s_2} a L \in \Sigma
        \\ \decl a K \in \Sigma
        \\ \kindtsub K {\what K}}
          {\jsigt {\Sigma, \subdecl {s_1} {s_2}}
                  {\what \Sigma,
                    \decl {s_1{\textit-}s_2}
                    {\metaappfive {\what K} a {\what{s_1}} {s_1}
                                              {\what{s_2}} {s_2}}}}
\end{mathpar}

\begin{mathpar}
\rulesheader{\jctxt [\Sigma] {\Gamma^+} {{\what \Gamma}^-}}
\and
\inferaxiom{\jctxt \cdot \cdot}
\and
\inferrule{\jctxt \Gamma {\what \Gamma}
        \\ \jchecksortclasst \Gamma S A {\what S}}
          {\jctxt {\Gamma, \rdecl x S A}
                  {\what \Gamma, \decl x A, 
                                 \decl {\what x}
                                       {\metaapp {\what S} {\expand A x}}}}
\end{mathpar}



\end{document}

%% file: lfr-listings.tex
\usepackage[usenames]{color}

\lstdefinelanguage{Twelf}{%
    identifierstyle=\it,                %
    commentstyle=\rm,                   
    keepspaces=true,                    %
    keywordstyle=[1]\textbf,            %
    keywordstyle=[2]\textsf,            %
    alsoletter={-+},                    %
    morecomment=[l]{\%},                %
    morecomment=[n]{\%\{}{\}\%},        %
    otherkeywords={\%total,\%worlds,\%mode,\%infix},  %
    morekeywords={[2]type},             %
    literate=*{->}{$\rightarrow$}{2}%
             {<-}{$\leftarrow$}{2}%
             {\{}{$\Pi$}{1}%
             {\}}{$.$}{1}%
             {[}{$\lambda$}{1}%
             {]}{$.$}{1}%
}

\lstdefinelanguage{LFR}[]{Twelf}{%
    keepspaces=true,                    %
    morekeywords={[2]sort},             %
    literate=*
             {-}{-}{1}%
             {->}{$\rightarrow$}{2}%
             {<-}{$\leftarrow$}{2}%
             {\{}{$\Pi$}{1}%
             {\}}{$.$}{1}%
             {[}{$\lambda$}{1}%
             {]}{$.\!$}{1}%
              {<}{$\langle$}{1}%
              {>}{$\rangle$}{1}%
             {<:}{$\leq$}{2}%
             {<<}{$\sqsubset$}{2}%
             {^}{$\wedge$}{1}%
             {\#}{$\top$}{1}%
             {_1}{$_1$}{2}%
             {_2}{$_2$}{2}%
             {X_1}{$X_1$}{2}%
             {X_2}{$X_2$}{2}%
             {al_1}{$\alpha_1$}{2}%
             {al_2}{$\alpha_2$}{2}%
             {*}{$\times$}{1}%
             {1}{$\mathbf{1}$}{1}%
             {\#1}{$\pi_1$}{2}%
             {\#2}{$\pi_2$}{2}%
             {:-}{$\div$}{1}%
             {-:>}{$\iarrow$}{3}%
             {[[}{[}{1}
             {]]}{]}{1}
}

\lstdefinelanguage{ExtendedTwelf}[]{Twelf}{%
    literate=*{->}{$\rightarrow$}{2}%
             {<-}{$\leftarrow$}{2}%
             {\{}{$\Pi$}{1}%
             {\}}{$.$}{1}%
             {[}{$\lambda$}{1}%
             {]}{$.$}{1}%
             {*}{$\times$}{1}%
             {1}{$\mathbf{1}$}{1}%
             {\#1}{$\pi_1$}{2}%
             {\#2}{$\pi_2$}{2}%
}